\documentclass[prx, tightenlines,aps,epsf,showpacs,superscriptaddress,nofootinbib,longbibliography,onecolumn]{revtex4-2}
\usepackage{algorithm} 
\usepackage{algpseudocode} 
\usepackage{float}
\usepackage{graphicx}  % needed for figures
\usepackage{xcolor}
\usepackage{bm}        % for math
\usepackage{braket}
\usepackage{amssymb}   % for math
\usepackage{epstopdf}
\usepackage{xfrac} 
\usepackage{cancel}
\usepackage{soul}
\usepackage{amsmath}
\usepackage{amsthm}
\usepackage{amsfonts}
\usepackage{bbm}
\usepackage{comment}
\usepackage{dsfont}
\usepackage{physics}
\usepackage[utf8]{inputenc}
\usepackage[english]{babel}
\usepackage[T1]{fontenc}
\usepackage{amsmath,amsthm,amsfonts,amssymb,xcolor}
\definecolor{darkblue}{rgb}{0.1,0.2,0.6}
\definecolor{darkred}{rgb}{0.8,0.1,0.2}
\definecolor{darkgreen}{rgb}{0.31,0.62,0.24}
\definecolor{bleudefrance}{rgb}{0.19, 0.55, 0.91}
\usepackage[colorlinks,citecolor=darkblue,linkcolor=darkblue,urlcolor=darkblue]{hyperref} 
\usepackage{enumerate}
\usepackage{setspace}
\usepackage{url}  % This makes \url work
\usepackage{mathrsfs}
\usepackage{mathtools} 
\usepackage{orcidlink}
\newtheoremstyle{italicroman}
  {3pt}   % Space above theorem
  {3pt}   % Space below theorem
  {\normalfont} % Body font (Roman)
  {}      % Indent amount
  {\itshape} % Theorem head font (italic)
  {.}     % Punctuation after theorem head
  { }     % Space after theorem head
  {}      % Theorem head spec

\theoremstyle{italicroman}
\newtheorem{theorem}{Theorem}[section]
\newtheorem{lemma}[theorem]{Lemma}
\newtheorem{definition}[theorem]{Definition}

\newtheorem{proposition}[theorem]{Proposition}
\newtheorem{corollary}[theorem]{Corollary}

\newtheorem{problem}[theorem]{Problem}

\date{\today}

\newcommand{\<}{\langle}

\renewcommand{\>}{\rangle}
\makeatletter
\renewcommand{\ALG@name}{Box} %Change the name Algorithm to Algoritme
\makeatother
\DeclareMathOperator{\poly}{poly}
\begin{document}
\title{Observation of a non-Hermitian supersonic mode on a trapped-ion quantum computer}
    \author{Yuxuan Zhang\orcidlink{0000-0001-5477-8924}}
    \thanks{quantum.zhang@utoronto.ca}
    
    \affiliation{Department of Physics, University of Toronto, Toronto, ON, Canada}
    \author{Juan Carrasquilla\orcidlink{0000-0001-7263-3462}}

    \affiliation{Institute for Theoretical Physics, ETH Zürich, Zürich, Switzerland}
    \author{Yong Baek Kim}
     \affiliation{Department of Physics, University of Toronto, Toronto, ON, Canada}

\begin{abstract}
Quantum computers have long been anticipated to excel in simulating quantum many-body physics. In this work, we demonstrate the power of variational quantum circuits for resource-efficient simulations of dynamical and equilibrium physics in non-Hermitian systems.
Using a variational quantum compilation scheme for fermionic systems, we reduce gate count, save qubits, and eliminate the need for postselection, a major challenge in simulating non-Hermitian dynamics via standard Trotterization. On the Quantinuum H1 trapped-ion processor, we experimentally observed a supersonic mode on an $ n = 18 $ fermionic chain after a non-Hermitian, nearest-neighbor interacting quench, which would otherwise be forbidden in a Hermitian system. Additionally, we investigate sequential quantum circuits generated by tensor networks for ground-state preparation using a variance minimization scheme, accurately capturing correlation functions and energies across an exceptional point on a dissipative spin chain up to length $ n = 20 $ using only 3 qubits. On the other hand, we provide an analytical example demonstrating that simulating single-qubit non-Hermitian dynamics for $\Theta(\log(n))$ time from certain initial states is exponentially hard on a quantum computer. Our work raises many intriguing questions about the intrinsic properties of non-Hermitian systems that permit efficient quantum simulation.
\end{abstract}
\maketitle
\section{Introduction}
In quantum mechanics, Hamiltonian Hermiticity is typically considered a fundamental postulate. However, this requirement can be relaxed in open systems, such as quantum hardware experiencing noise, atoms undergoing spontaneous decay, or other scenarios involving measurement and postselection---the process of monitoring a quantum system and analyzing its behavior following the quantum trajectory of specific measurement outcomes (see Fig.~\ref{fig:quench}a). This selective process enables an effective description of the system's dynamics using pure states evolving under a non-Hermitian Hamiltonian, providing new strategies for efficient numerical simulation and an alternative physical perspective for studying open quantum systems~\cite{dalibard1992wave}.

%{\color {blue}{YB: it may be useful to explain what you mean by "postselection". Nature physics aims at rather general audience.}} 
The study of non-Hermitian physics traces back to the complex field theory approach by Yang and Lee, who used it to explore new phases of matter~\cite{yang1952statistical}. Since then, it has been established that Hamiltonian non-Hermicity %unitary state evolution due to dissipation 
leads to unconventional phase transitions and unique phenomena~\cite{fisher1978yang,hatano1996localization,hatano1997vortex,ashida2020non}, such as exceptional points~\cite{bender1969anharmonic}, the non-Hermitian skin effect~\cite{yao2018edge}, supersonic modes~\cite{ashida2018full,dora2020quantum}, topological phases~\cite{rudner2009topological}, and unique entanglement behaviors~\cite{gopalakrishnan2021entanglement,kawabata2023entanglement}. Although these discoveries have generated new possibilities and excitement in the era of quantum information, non-Hermiticity poses an engineering challenge for conventional electrical, mechanical, photonic and cold atom platforms, and experimental demonstrations have been largely confined to few-particle problems~\cite{guo2009observation,ruter2010observation,feng2011nonreciprocal}.

%Since then, it has been known that the non-unitary state evolution due to dissipation leads to unconventional phase transitions and gives rise to unique behaviors such as exceptional points~\cite{bender1969anharmonic,berry1998diffraction,heiss2012physics}, skin effect~\cite{}, supersonic modes~\cite{} and entanglement phase transitions, creating new possibilities and excitement in the era of quantum computing~\cite{bender1998real,el2018non,yamamoto2019theory,hamazaki2019non,ashida2020non,guo2021entanglement,yamamoto2022universal,shen2023observation}.

The recent progress of programmable quantum computers is expected to offer new opportunities to solve computational tasks~\cite{shor1994algorithms}, enable cryptography~\cite{pirandola2020advances}, sample hard distributions~\cite{aaronson2011computational}, and, crucially, to simulate many-body physics~\cite{feynman2018simulating}. However, the presence of physical noise creates a significant obstacle in the faithful execution of such tasks. More essentially, the postselection nature of non-Hermitian dynamical simulations brings the challenge to another level: In a ``direct'' simulation by Trotterization, the success rate of postselection would vanish exponentially in time, creating a huge sampling overhead. To date, digital quantum simulations of non-Hermitian physics have also been limited in scale compared to those of Hermitian quantum systems~\cite{wen2019experimental,shen2023observation}.
 
On the other hand, variational quantum algorithms (VQAs)~\cite{cerezo2021variational} such as the quantum approximate optimization algorithm (QAOA)~\cite{farhi2014quantum} and the variational quantum eigensolver (VQE)~\cite{peruzzo2014variational} promise to offer advantages in the simulation of quantum systems with near-term quantum hardware~\cite{preskill2012quantum}. %Even though an analytical study could be hard to prove~\cite{}, 
Combined with machine learning~\cite{carrasquilla2017machine} and tensor-network techniques~\cite{schon2005sequential}, VQAs have already yielded successful applications in reducing the resource cost in quantum time evolution~\cite{lin2021real,zhang2024scalable} and quantum state preparation~\cite{foss2021holographic}. It is then natural to ask whether one could utilize these methods in the quantum simulation of non-Hermitian systems. Intuitively, VQAs explore the large degree of freedom in the Hilbert space and could mitigate the postselection issue, and even reduce the circuit gate count.

%It is worth noticing that upper-bound in the advantage one could get is very sensitive to the type of Hamiltonian and initial state. 
%We analytically prove a no-go theorem when even single qubit non-Hermitian time evolution on certain initial states for merely $t\sim\log(n)$ can be exponentially hard to simulate using a quantum circuit. 
In this work, we demonstrate the power of variational quantum algorithms through the analytical, numerical, and experimental study of various non-Hermitian systems. First, combining variational quantum compilation (VQC) and Gaussian matrix product state (GMPS) method, we simulate the non-Hermitan quench dynamics in a strongly correlated fermionic system~\cite{dora2020quantum} without post-selection, drastically reducing the quantum resources required compared to a Trotterization scheme. Most significantly, this includes a $\sim10^{19}$ sampling overhead due to post-selection, as we estimate. On Quantumtinuum's H1 trapped-ion quantum processor, we observe a non-Hermitian supersonic mode where correlation functions propagate faster than the conventional light cone. Second, by integrating the data compression capabilities of quantum matrix product states (qMPS)~\cite{schon2005sequential,perez2006matrix} with a variance minimization algorithm, we devise an eigenstate-finding algorithm for non-Hermitian systems suitable for near-term quantum computers. Testing our algorithm on Quantumtinuum's H1 trapped-ion quantum computer, we accurately capture the energy and correlation functions of a $n = 20$ non-Hermitian spin chain around an exceptional point (EP), where the real-valued ground energy merges with the first excited state to form a Hermitian conjugate pair. This result leverages the qMPS compression capabilities, wherein a classical MPS with bond dimension $\chi$ can be stored on a quantum computer using memory scaling as $q \sim \log\chi$~\cite{schon2005sequential,foss2021holographic}, extending the quantum memory advantages of qMPS to the study of non-Hermitian systems. 

\begin{figure*}
    \centering
    \includegraphics[width=.95\linewidth]{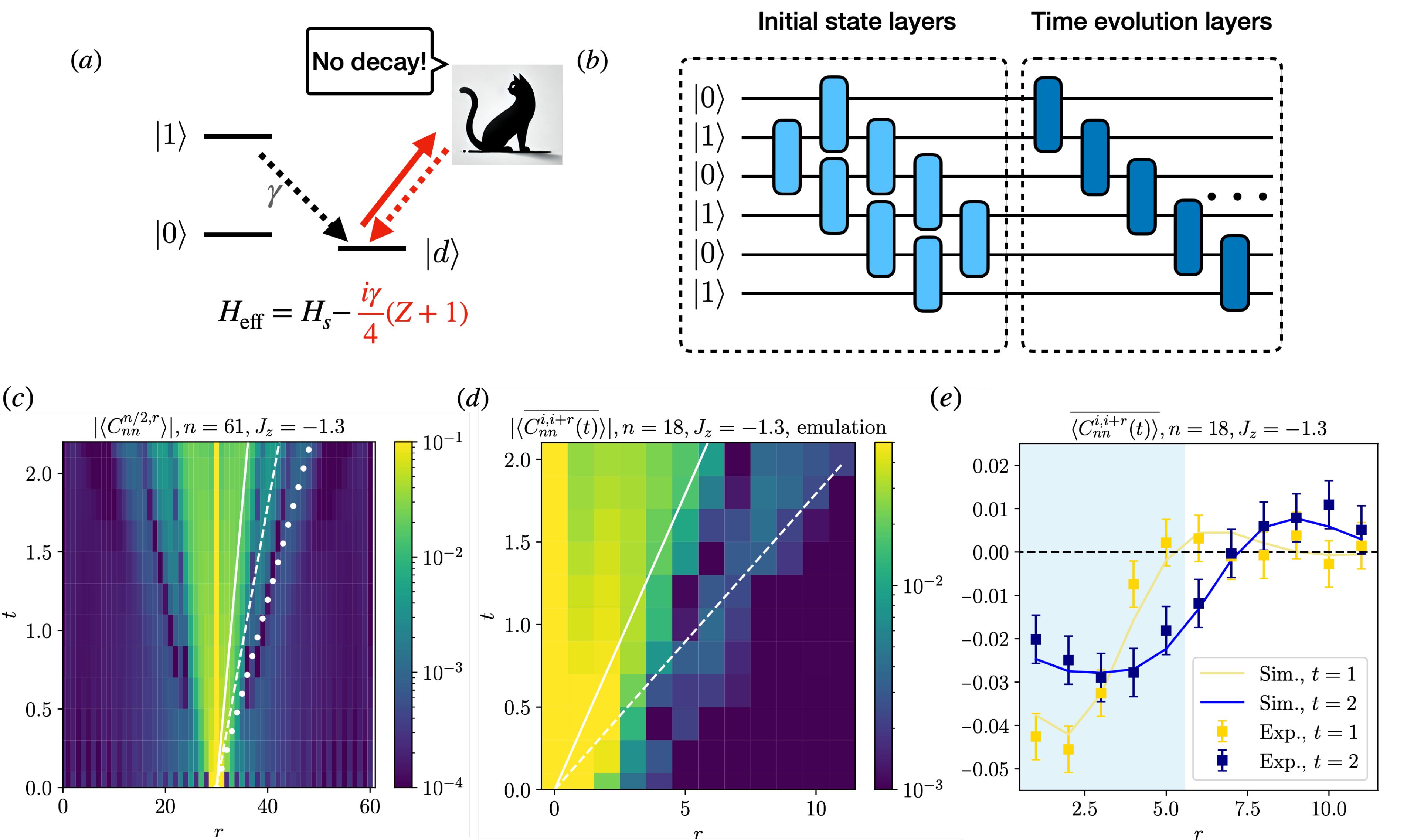}
    \caption{(a) A minimal example of a non-Hermitian system: Consider a two-level atom with eigenstates $\{\ket{0},\ket{1}\}$. When in the excited state $\ket{1}$, the atom spontaneously decays to a dark state $\ket{d}$ at a rate $\gamma$. A cat observer closely monitors the population in the dark state, and by postselecting on the absence of decay, the atom can be effectively thought of as undergoing non-Hermitian time evolution~\cite{dalibard1992wave}. 
    (b) The use of VQAs in simulating the non-Hermitian quantum quench of Eq.~\ref{eq: fermi} resolves the postselection issue: The half-filled free-fermion ground state is first prepared with a circuit generated by a Gaussian matrix product state. Additional layers of variational unitary circuits are then attached to compress the non-Hermitian time evolution.
    (c) Numerical observation of a supersonic mode in non-Hermitian quench dynamics. The different modes can be distinguished from the intensity `valleys' of the correlation function. The white solid line denotes the Lieb-Robinson bound and the dashed lines represent the first two supersonic modes.
    (d) For the same physical parameter as in (c) at $n = 18$, we emulate the compressed quench dynamics using 100,000 shots and report the connected density-density correlator averaged over bulk sites to enhance signal level. 
    (e) Experimental data of the same circuit in (d) collected using Quantinuum's H1 quantum computer at $t = 1$ and $t = 2$ for 5,000 shots each. We compare the averaged connected density-density correlator to a classical TEBD simulation. The error bars represent two standard deviations and the shaded blue region indicates the conventional lightcone at $r = 2vt$ where $t = 2$.} 
    
    \label{fig:quench}
\end{figure*}

\section{Results}
\subsection{Observation of a supersonic mode}~\label{sec:sm}
One of the most striking dynamical properties of non-Hermitian physics is the violation of the Lieb-Robinson (LR) bound~\cite{lieb1972finite}, a cornerstone of quantum dynamics that imposes strict limits on the propagation of information in many-body systems. Although LR upper bounds the information spreading velocity in a locally interacting Hermitian quantum system, the dynamics under strictly geometrically two local non-Hermitian Hamiltonians can allow supersonic information propagation modes~\cite{ashida2018full}, without introducing any (quasi) long-range interactions~\cite{eisert2013breakdown, richerme2014non}. This can be probed by studying the two-point correlation functions of the quenched dynamics, starting from the ground state of some Hermitian Hamiltonian $\ket{\psi_{0}}$ and time-evolved with a different non-Hermitian Hamiltonian $H$: 
\begin{align}
    \langle\mathcal{O}_0\mathcal{O}_r\rangle(t) &\equiv \bra{\psi(t)}\mathcal{O}_0\mathcal{O}_r\ket{\psi(t)}/N(t)\\
    &= \bra{\psi_{0}}e^{iH^\dag t}\mathcal{O}_0\mathcal{O}_r e^{-iHt}\ket{\psi_{0}}/N(t)\\
    &= \bra{\psi_{0}}e^{iH^\dag t}e^{-iHt}e^{iHt}\mathcal{O}_0e^{-iHt}e^{iHt}\mathcal{O}_r e^{-iHt}\ket{\psi_{0}}/N(t)\\
    &= \bra{\psi_{0}}e^{iH^\dag t}e^{-iHt}\mathcal{O}_0(t)\mathcal{O}_r(t)\ket{\psi_{0}}/N(t). \label{eq:corr}
\end{align}

where $N(t) \equiv \bra{\psi(t)}\ket{\psi(t)}$ is the normalization factor, and $\mathcal{O}(t) = e^{iHt}\mathcal{O} e^{-iHt}$ is pseudo-Heisenberg evolution. In~(\ref{eq:corr}), the pseudo-Heisenberg evolved operators are restricted by the Lieb-Robinson bound, yet in general the term $e^{iH^\dag t}e^{-iHt}$ can be long-range, enabling supersonic modes even in non-interacting models. In this work, we consider the following non-Hermitian, interacting fermionic chain with open boundary conditions, whose low energy spectrum corresponds to a $\mathcal{PT}$-symmetric Luttinger liquid~\cite{dora2020quantum,moca2021universal}:
\begin{align}\label{eq: fermi} 
H_{\rm fermi} = \sum_{i}\frac{1+iJ_z}{2}(c^\dag_i c_{i+1}+h.c.) - i\frac{\pi J_z}{2} n_in_{i+1}.
\end{align}

Although $H$ is local, it is possible to show that the system has modes corresponding to the velocity at $x = 2kvt$ through a bosonized effective theory calculation~\cite{dora2020quantum}, where $k$ is an integer, and 
\begin{align}
    v = 1+\frac{\pi^2-8}{8}J_z^2
\end{align} is the LR velocity. This is in strong contrast with the Hermitian case or systems in a full Linbladian evolution~\cite{ashida2018full}, where only the first ($k=1$) mode is permitted. We study the quench dynamics of the system starting with the free-fermion ground state at $J_z = 0$ and time evolving the state with a non-zero $J_z$ Hamiltonian. We observe the supersonic modes through measuring a connected density-density correlator $C_{nn}^{0,r}(t) = [\langle{n}_0{n}_{r}\rangle(t) - \langle{n}_0\rangle(t)\langle{n}_{r}\rangle(t) ]$, which reveals both the presence of supersonic light cones and their decay, as theoretically predicted and numerically shown in Fig.~\ref{fig:quench}c. 

Non-Hermitian time evolution can result from conditional Lindblad-type dynamics, where continuous environmental monitoring is applied to prevent quantum jumps. However, in a direct simulation with, e.g., Trotterizaiton~\cite{trotter1959product} and block encoding~\cite{gilyen2019quantum}, the likelihood of post-selection diminishes over time, and thus observing non-Hermitian dynamics becomes increasingly challenging. In fact, the post-selection rate would decay exponentially in $t$ even if one could simulate open system dynamics exactly with an analog simulator. Nevertheless, the suppressed supersonic modes in the correlation function shown in Fig.~\ref{fig:quench}(a) suggest that the post-quench state still exerts some form of locality at least for $t \sim \log (n)$, providing an opportunity to be compressed by a low complexity state. For a finite given accuracy, it suffices to truncate the correlation function at a finite mode number and set the cut-off supersonic mode as the effective light cone, which would still be efficiently simulable.

In this work, we circumvent this challenge by utilizing two tensor-network-inspired techniques: GMPS~\cite{fishman2015compression} for preparing the initial state and VQC~\cite{lin2021real} for approximating the non-Hermitian time evolution. 
As we explain in the supplement, GMPS provides an efficient approximation to the Gaussian mean-field state by exploiting the near-area-law entanglement nature of the ground state, and VQC dramatically compresses time-evolved states that have high entanglement but low complexity.
%{\color {blue}{need to explain briefly what GMPS and VQC are.}} 
Sketched in Fig.~\ref{fig:quench}(b), we start with a compressed representation of the initial free fermion state and variationally compute a circuit representation at each Trotterized time step, as detailed in Algorithm~\ref{algo:compressed_circuit}. In practice, we found that it suffices to use a shallow, unitary variational circuit to capture the action of the non-Hermitian dynamics, without the necessity to introduce ancillary qubits.

\fbox{%
\begin{minipage}{1\linewidth}
\begin{algorithm}[H]
	\caption{-- Compressed quantum quench with VQC} 
	\begin{algorithmic}[1]
        \State Find the GMPS representation that prepares the free-fermion ground state. This could be done either deterministically or variationally~\cite{fishman2015compression,niu2022holographic}. 
        Denote this state as $|\psi_0\rangle$.
		\For {$t=\Delta t, 2\Delta t, ..., T$}
			\State Calculate the normalized target state after a Trotterized time evolution, $\ket{\psi_{t}} \propto e^{-iHt}\ket{\psi_{t_0}}$
			\State Initialize the parametrized quantum circuit.$\ket{\Tilde{\psi}}$ at $\bm\theta_0$; then use a gradient descent method to minimize infidelity $1- \mathcal{F} = 1-|\langle\psi_{t}|\Tilde{\psi}(\bm\theta)\rangle|^2$ wrt. $\bm\theta$
                \State Record optimized parameter $\bm\theta^*(t)$, $\bm\theta_{0}\leftarrow\bm\theta^*$
                \If{$1- \mathcal{F} > \epsilon$}
                \State Add a layer of PQC to the time evolution layers which we initialize at the identity
                \State Return to 4
			\EndIf
		\EndFor
	\end{algorithmic} \label{algo:compressed_circuit}
\end{algorithm}
 \end{minipage}
}

% For the parameterized quantum circuit (PQC), we use a %{\color {blue}{what do you by "number" here ?}} 
%  parameterized gate that preserves the $U(1)$ symmetry called `fSim' gates~\cite{arute2019quantum}, detailed in Appx.~\ref{app:experimental}. This has three benefits: it reduces the number of variational parameters from 15 to 5 per gate, saves the number of CNOT gates from 3 to 2 per gate, and allows error mitigation at essentially no additional cost.
 
 For the parameterized quantum circuit (PQC), we use a parameterized gate known as the `fSim' gate~\cite{arute2019quantum}, which preserves the $U(1)$ symmetry, as detailed in Supplementary Note 1. This choice offers three key benefits: it reduces the number of variational parameters from 15 to 5 per gate, decreases the number of CNOT gates from 3 to 2 per gate, and enables error mitigation with essentially no additional cost. A similar circuit architecture has been proposed in~\cite{niu2022holographic} and is denoted as the GMPS+X ansatz. Notice that, generally, the imaginary time evolution does not necessarily conserve charge---only when the initial state is in a single charge sector---%, except when the initial state is in a single charge sector, 
 which is the situation we consider here. Fixing the target infidelity $1-\mathcal{F} = 0.02$ and comparing our VQC to a naive Trotterization implementation, we find a reduction of CNOT gate count by a factor of 9 and the elimination of auxiliary qubits, reducing their number from 9 to 0. Most significantly, the compressed evolution method does not require postselection, unlike block encoding. As estimated in the Supplementary Note 1, this avoids a sampling overhead of $10^{19}$.

Another experimental challenge is that supersonic modes have relatively low amplitudes. Increasing the magnitude of $J_z$ can increase the intensity of the supersonic modes, but there is a threshold $\sim -1.33$, above which the Lutingger Liquid prediction would fail~\cite{dora2020quantum} and thus we set our $J_z = -1.3$. Still, observing the signal level is demanding: In the classical simulation result with time-evolving block decimation (TEBD) Fig.~\ref{fig:quench} at $n = 61$, the correlation strength of the first supersonic model already drops below $\sim0.01$.

To resolve this, instead of reporting the correlation function $C_{nn}^{i,i+r}(t)$ for a specific initial position in the chain, we report the correlation function averaged over all bulk sites $\overline{C_{nn}^{i,i+r}(t)}$ because the supersonic mode should be a universal property within the system. Asymptotically, this post-processing technique effectively increased the amount of data by a factor of $\sim(n-r)$, roughly an order of magnitude in this case. We illustrate our experimental details in Supplementary Note 1. 

Comparing the TEBD and emulation results between panels a) and b) of Fig.~\ref{fig:quench}, we find the supersonic mode is preserved in $\overline{C_{nn}^{i,i+r}(t)}$. Limited by the number of shots as trapped-ion simulators tend to have a slower clock rate in exchange for higher fidelity, we choose $t = 1$ and $t = 2$ for experimental implementation on Quantinuum's H1 machine. 
As demonstrated in Fig.~\ref{fig:quench} c), the vast majority ($>95\%$) of data points align with the classical prediction within two standard deviations with merely 5,000 shots. This would not have been possible without the usage of variational circuits and proper data processing. To compare, a recent dynamical digital simulation on a merely $n=6$ non-interacting fermionic chain took 160,000 shots~\cite{shen2023observation} to implement for each Trotter step. %We present a more detailed resource comparison in Appx.~\ref{app:experimental}. 

Two clear trends emerge from the data: the supersonic mode gains amplitude over time, while the LR mode loses amplitude. Additionally, by observing the shift in the LR peak between the two time slices (from $r = 2$ to $r = 3.5$), we estimate the LR velocity to be $v_{\rm exp} = 1.5$, which closely matches the theoretical prediction of $v = 1.4$.

It is worth noticing that for the time evolution generated by a generic non-Hermitian Hamiltonian, completely circumventing postselection with $\poly(n)$ gates is nearly impossible even with the use of variational compilation unless complexity classes collapses: $\mathsf{BQP = PostBQP = PP}$~\cite{aaronson2005quantum}. But what if we make strong assumptions about the Hamiltonians, such as they are physically local? Surprisingly, for certain initial states, we find an example where the time evolution of a Hamiltonian consisting only of single qubit terms could be too powerful to simulate, even for a very short period: 
\begin{theorem}[Single qubit non-Hermitian dynamics for logarithmic time can be hard to simulate with a quantum circuit]\label{theo: hard}
    There exist $\ket{\psi_{0}}$ and a non-Hermitian Hamiltonian $H$ consisting with single-qubit terms only, such that the circuit $C$ returns 
    \begin{align}
        C \ket{\psi_{0}} =  \frac{e^{-iHt}\ket{\psi_{0}}}{\sqrt{N(t)}}
    \end{align} for merely $t = \Theta(\log(n))$ requires $e^{\Omega(n)}$ two-qubit gates to implement.
\end{theorem}
One example is to take $\ket{\psi_{0}}$ to be a Haar random state and $H = -i\sum_i Z_i$: running the time evolution for $t = \Theta(\log(n))$ allows one to distinguish the Haar random state to a maximally mixed state, which should be exponentially hard even for quantum circuits~\cite{brandao2021models}. We provide a detailed analytical proof of this no-go theorem in the Methods section. As a corollary, we also show that the theorem can be thought of as an oracle separation between complexity classes $\mathsf{BQP}$ and $\mathsf{PostBQP}$, as well as between $\mathsf{BQP}$ and $\mathsf{PostQNC_0}$ (to be defined in the Methods). %

\subsection{Crossing the exceptional point}\label{sec:ep}

We now focus on studying eigenstates of non-Hermitian Hamiltonians through variational circuits. We first recall that the matrix spectral theorem does not hold for non-Hermitian Hamiltonians ($H \neq H^{\dagger}$), which implies the eigenvalues of $H$ are not necessarily real, and the left eigenvectors are not the Hermitian conjugates of the right eigenvectors:
\begin{align}
    H\ket{\psi_i^R} = E\ket{\psi_i^R},\ \ \ \ \ H^\dag\ket{\psi_i^L} = E^*\ket{\psi_i^L}.
\end{align} 
\begin{figure*}
    \centering
    \includegraphics[width=.8\linewidth]{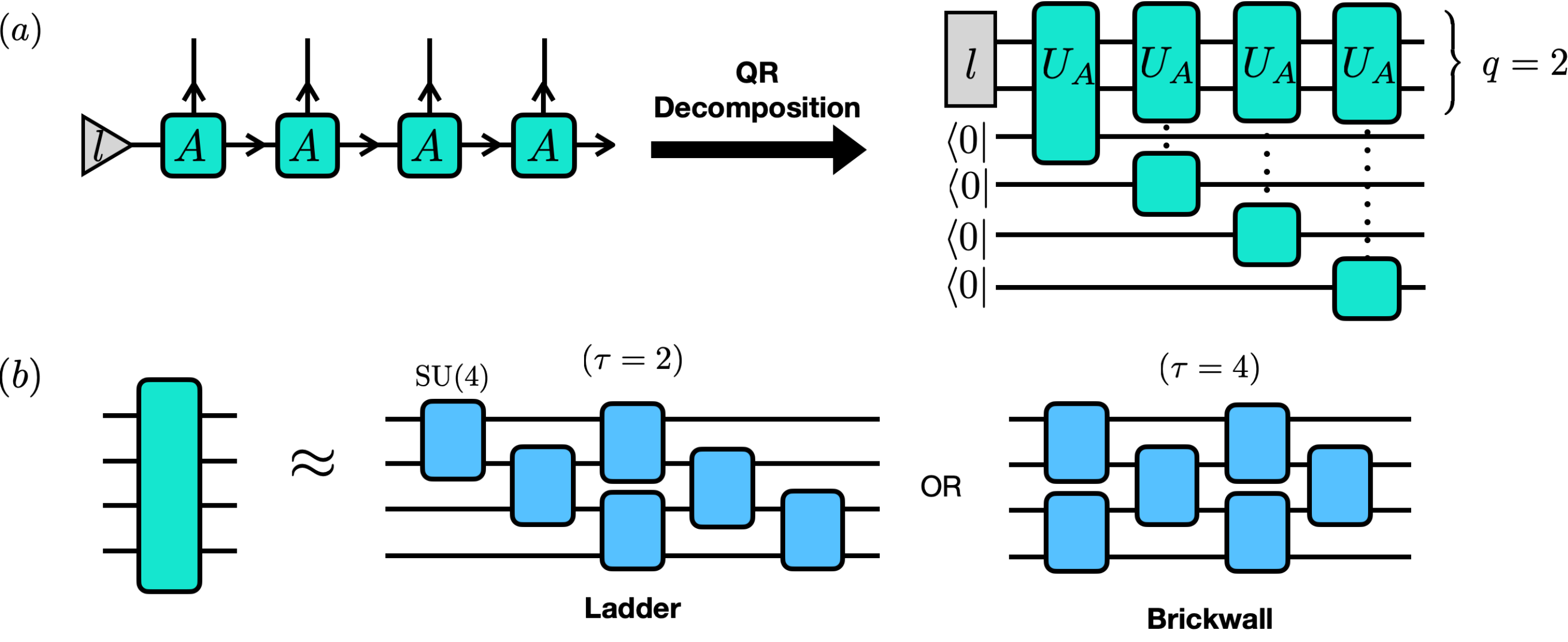}
    \caption{(a) While any many-body wavefunction can be cast as an MPS, it permits especially efficient ground state representations for $1d$ systems. In the right canonical form, each tensor in the MPS is an isometry~\cite{schon2005sequential,schon2007sequential}. To implement on a quantum computer, each isometry is embedded as a $2\chi \times 2\chi$ unitary through a QR decomposition. The MPS is then cast into a sequential quantum circuit. (b) Each $U_A$ can be synthesized from SU(4) gates with certain local geometries. Here we demonstrate a ladder geometry and a brickwall geometry. The depth of the local circuit is denoted $\tau$. The sequential quantum circuit can then be used as a variational ansatz. Due to having limited access to quantum resources, we first classically variationally optimize each SU(4) gate based on the cost function, e.g. Eq.~\ref{eq:cost}, and then compile each SU(4) gate into native gates of the Quantinuum H1 processor before executing. The energy of the spin system is calculated from sampling correlation functions in the real experiment. See also Supplementary Note 1 for more experimental details.}
    \label{fig:qmps}
\end{figure*}
The eigenstates within each left and right set are no longer orthogonal to each other. Instead, we have the biorthogonal relationship that exists between left and right eigenstates:
\begin{align}
    \braket{\psi_i^R}{\psi_j^R} \neq \delta_{ij},\ \ \ \ \ \braket{\psi_i^L}{\psi_j^R} = \delta_{ij}.
\end{align} 
In designing a variational algorithm targeting eigenstates of non-Hermitian Hamiltonians, these observations imply that the traditional objective function in a ground-state search, namely $\mathcal{L} = \braket{\psi({\bm\theta})|H}{\psi({\bm\theta})}$ is no longer suitable. Fortunately, a `variance minimization' method has been recently proposed~\cite{xie2023variational}. Instead of directly minimizing the energy function, it optimizes the following variance quantity, which is semi-positive definite:  
\begin{align}\label{eq:cost}
\mathcal{L}'(H,E,{\bm\theta}) = \bra{\psi({\bm\theta})}(H-E)(H^\dag-E^*)\ket{\psi({\bm\theta})}.
\end{align}
Here $E$ is a complex variable we optimize over. When $\mathcal{L}' = 0$ we are guaranteed that $\ket{\psi({\bm\theta})}$ is a right eigenstate of $H$ and $E$ is an eigenvalue (for the rest of the paper, we focus on the property of right eigenstates, and the left eigenstates can be found by substituting $H\leftarrow H^\dag$). Using this method, the authors devised a numerical optimizer and were able to numerically compute the left and right eigenstates, verifying the biorthogonal relations, as well as evaluating many observables~\cite{xie2023variational}.

Although the variance minimization scheme presented in Ref.~\onlinecite{xie2023variational} includes all the fundamental components of a variational quantum algorithm, the simulated system sizes were limited to only 7 qubits due to two main obstacles: First, they used a direct, full-state preparation method that is demanding on both classical and quantum memory. Capturing long-range correlations would then require a large circuit volume, making it impractical for larger systems. Second, the termination condition for variance minimization only ensures that some eigenstate is found, but not necessarily the one of interest. Given that the number of eigenvalues increases exponentially with system size, a random guess for $E$ is likely to result in an undesired eigenstate.

\begin{figure*}
    \centering
        \centering
        \includegraphics[width=.95\linewidth]{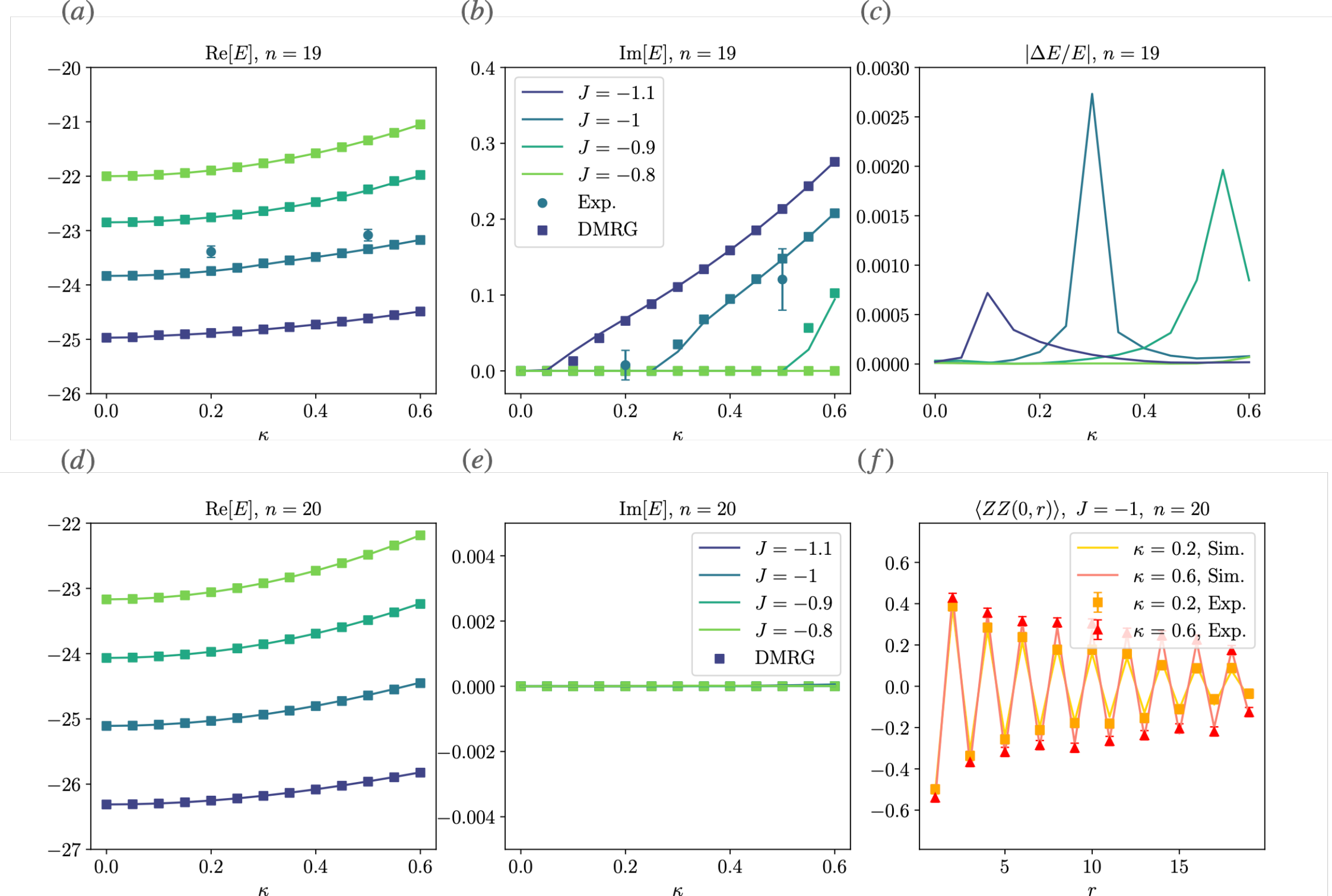}
        \caption{(a) The real and (b) imaginary parts of variational right ground state preparation results of a $n = 19$ non-Hermitian Ising chain are compared with that from a $\chi = 100$ NH-DMRG. The exceptional point can be read from the imaginary part of $E$. We experimentally execute the $J = -1, \kappa \in \{0.2, 0.5\}$ circuits for 2,000 measurement shots and report the energy outcomes, which are both within 2\% error of the DMRG values. (c) The absolute value of relative energy difference. Notice that the peak on each curve corresponds to the EP read from (b). (d) and (e) repeat the study of (a) and (b), except now the system size is set to 20 so that the system is EP-free. (f) Comparing $ZZ$ correlation functions measured in Quantinuum's H1 processor (2,000 shots) to DMRG results. All results are prepared with a $q = 2$, $\tau = 2$ ladder architecture as defined in Fig.~\ref{fig:qmps}, and error bars represent one statistical standard deviation. More experimental details on the implementation can be found in the Supplementary Note 1.}
        \label{fig:ising}
\end{figure*}

To solve the first obstacle, we make use of quantum circuit tensor network state (qTNS) techniques~\cite{schon2005sequential}, to sequentially simulate many-body quantum states, as shown in Fig.~\ref{fig:qmps}. Rather than a direct one-to-one encoding between a spin and a qubit, sequential simulation introduces $q$ ancilla qubits, or bond qubits, with local circuit depth $\tau$ to faithfully capture the near-area law entanglement structure of a physical state, generating a sequence of quantum operations that allows one to sample properties of the many-body state along a spatial direction without storing the full state in quantum memory; we defer a more detailed explanation to Supplementary Note 1. 

For the second obstacle, we designed a `warm start' method, as widely used in many other variational algorithms~\cite{lin2021real}. Namely, we first turn off the non-Hermitian field and solve for the Hamiltonian's ground state with a VQE algorithm and gradually turn on the imaginary field, feeding the previous optimization result as an initialization. Concretely, we consider the $1d$-transversal field Ising model with an imaginary longitudinal field~\cite{von1991critical}:

\begin{align}\label{eq: ising}
    H_{\rm Ising}(J, \kappa) = -\sum_i [JZ_{i}Z_{i+1} + X_{i} +i\kappa Z_{i}], 
\end{align}
where $J$ and $\kappa$ are real numbers. To find the ground state at a certain $\kappa_{\rm targ}$, we execute the procedure in Algo.~\ref{algo:vqe}:
\fbox{%
\begin{minipage}{1\linewidth} % Adjust width as needed
 \begin{algorithm}[H]
	\caption{-- VQE algorithm for dissipative Ising model with variance minimizaiton} 
	\begin{algorithmic}[1]
        \State Set the target Hamiltonian to be $H = H_{\rm Ising}(J, \kappa = 0) $ 
        \State Initialize $\ket{\psi({\bm\theta})}$ where the circuit with a variational circuit generated by quantum matrix product state (qMPS)~\cite{foss2021holographic}
        \State Find the Hamiltonian's ground state by minimizing $\mathcal{L} = \braket{\psi({\bm\theta})|H}{\psi({\bm\theta})}$ and record the optimized circuit parameter $\bm\theta^*$ and energy $E^*$;
		\For {$\kappa=\Delta \kappa,2\Delta \kappa,\ldots \kappa_{\rm targ}$}
			\State Initialize the circuit at $\bm\theta_{0}\leftarrow\bm\theta^*$ and $E_{0}\leftarrow E^*$
                \State Minimize the target function $\bra{\psi({\bm\theta})}(H-E)(H^\dag-E^*)\ket{\psi({\bm\theta})}$ over $\bm\theta$ and $E$.
                \State Record optimized parameter $\bm\theta^*(\kappa)$ and energy $E^*(\kappa)$ 
		\EndFor
	\end{algorithmic} \label{algo:vqe}
\end{algorithm}
\end{minipage}}
The dissipative Ising model exerts $\mathcal{PT}$-symmetry, but it goes through spontaneous symmetry breaking at $\kappa_c$, the so-called exceptional point: at $\kappa<\kappa_c$, the ground state energy remains real despite the non-Hermicity; at $\kappa>\kappa_c$, the ground state energy merges with the first excited state to form a complex conjugate pair in their energy values. The exceptional point is very sensitive to $J$ and the parity of the length of the chain. 

In the antiferromagnetic phase and open boundary condition, there will be a real to complex transition only when the system size is odd; when $n$ is even, such transition does not happen due to the ground state and first excited state being in different charge sectors~\cite{von1991critical}. In Fig.~\ref{fig:ising}, we compare simulation results for our numerical results to density-matrix renormalization group generalized to non-Hermitian systems (NH-DMRG) with iTensor~\cite{fishman2022itensor} at $\chi = 100$. Setting $\Delta \kappa = 0.05$ and using a ladder sequential circuit with $q = 2$ and $\tau = 2$, the variance minimization algorithm accurately captures the EP physics for different $J$ and $\kappa$ at system sizes $n = 19$ and $n = 20$. At each $J$, the optimizer returns a relative energy error below $0.01\%$ until $\kappa$ gets close to the EP, where the error rises to as high as $\sim 0.3\%$ and drops again, suggesting an increase in difficulty in representing the states near EP.

Additional numerical results for a non-Hermitian XXZ chain with size $n = 64$ can be found in the Supplementary Note 1. In Hermitian physics, it is known that quantum computers offer a ``data compression'' advantage in representing physical states~\cite{perez2006matrix,foss2021holographic,niu2022holographic}, as a quantum computer can store a MPS with bond dimension $\chi$ with merely $\sim \log\chi$ bond qubits. Our result extends this quantum advantage to non-Hermitian quantum material simulations.

%\YZ{finish the discussion part once we have experimental data, possibly for another 200 words. Need to talk about various boundary conditions giving different exceptional point behavior...} 

%should also look at left-right correlator $XX~L^{\frac{2}{5}}$
% to include:
%\subsection{Crossing the exceptional point on a quantum computer}
%A non-Hermitian Hamiltonian that describes an open system generically has complex eigenvalues, which must be studied on the complex plane, which leads to the emergence of eigenvalue topology, or spectral topology. This additional ‘layer’ of topology is a unique feature for non-Hermitian systems.

%Spectral topology fundamentally affects the parallel-transport behaviors of eigenvectors of a non-Hermitian Hamiltonian. Special attention is also needed to exceptional points — branch-point singularities on the complex eigenvalue manifolds

%We consider ground-state finding problem for the Ising Hamiltonian in Eq.~\ref{eq: ising}. Using a brick-wall local geometry on the sequential circuit, we classically optimize the circuit parameters and consider hardware implementation on a real quantum computer to take into account the shot noise and gate imperfections, which is simulated through qiskit~\cite{anis2021qiskit}. 

\section{Discussion}

We have combined two quantum state compression techniques, GMPS and VQC, to experimentally study the quenched dynamics of an extended strongly-correlated non-Hermitian system. Using a trapped-ion implementation on the Quantinuum H1 quantum processor, we have observed a supersonic mode in the connected density-density correlation function of a fermionic chain following a non-Hermitian quantum quench---a phenomenon conventionally forbidden by the Lieb-Robinson bound in Hermitian systems. The system sizes, time scales, and low sample complexity achieved in our experiments are enabled by the convergence of the advantages offered by universal quantum computers and the variational formulation of the dynamics employed in our experiments, allowing us to bypass costly postselection-based implementations with an affordable approach based on unitary dynamics. Similarly, we exploited the efficient compression offered by qMPS to experimentally prepare eigenstates of a non-hermitian spin chain where we experimentally measured correlation functions and energies across an exceptional point on a dissipative spin chain up to length $ n = 20 $ using only 3 qubits. These results show that universal quantum computers have great potential and offer an advantage as compared to traditional non-universal experimental platforms in the study of non-Hermitian quantum matters.

Our work also raises a number of intriguing questions regarding the opportunities and challenges of using digital quantum computers for simulating non-Hermitian systems. In particular, in Thm.~\ref{theo: hard} we established the hardness of the simulation of quench dynamics of certain initial states under single qubit non-Hermitian dynamics. A key open question arising from our experiments and theoretical findings is to determine the properties of the initial state and the non-Hermitian Hamiltonian that enable efficient unitary simulation. Here, a natural direction is to explore the use of VQC to alleviate the postselection requirements which are often encountered in systems such as those exhibiting measurement-induced phase transitions~\cite{li2018quantum}. Similarly, these and other related questions can be reframed in the context of quantum algorithms for solving dissipative differential equations. Specifically, under what conditions can quantum computers provide a significant advantage in solving these equations? This problem has been studied from an algorithmic perspective~\cite{liu2021efficient} and can potentially offer new insights into the question of simulability of non-hermitian systems more broadly. 

A natural extension of our work is to explore dynamics and eigenstate properties beyond one-dimensional systems where interesting non-Hermitian phenomena emerge~\cite{yoshida2019non}. One possibility is, that instead of doing a sudden quench, one could use our protocol to simulate non-adiabatic passage across an EP and investigate the proposed unconventional Kibble-Zurek mechanism~\cite{yin2017kibble,dora2019kibble}. Another interesting direction would be exploring the asymptotic quantum memory costs to approximate the ground states in different non-Hermitian phases of matter. In Hermitian systems, one could relate the hardness of compression to entanglement entropy laws~\cite{hastings2007area}. It is worth pursuing what metric one should use in the non-Hermitian case, as there does not exist a good entanglement measure due to the lack of an orthonormal basis~\cite{ashida2020non}.  As we push the boundaries of quantum computation, the techniques and insights gained from our study not only illuminate new pathways for exploring non-Hermitian physics but also pave the way for future advancements in quantum simulation aided by the convergence of innovative techniques such as compression and variational compilation with cutting-edge quantum hardware. 

\section{Methods}\label{app:methods}
This section details the numerical and analytical methods used throughout the work. We give a motivating example of non-Hermitian physics, followed by the discussion of sequential circuits generated by Matrix Product States (MPS) and Gaussian Matrix Product States (GMPS), explaining how they enable efficient quantum representations of near-area-law quantum states. %Next, we highlight the importance of using a proper warm start in all numerical simulations. 
Additionally, we provide an analytical proof of Theorem~I.1.
\subsection{Non-Hermitian physics: a motivating example}
We argued in the main text that a physical way to motivate non-Hermitian physics is from open quantum system dynamics. Usually, the Linbladian equation is used to describe the Markovian evolution of a system interacting with a thermal bath~\cite{manzano2020short}:

\begin{align}
    \frac{d\rho}{dt} &= -i[H,\rho] +\sum_i\gamma_i(L_i\rho L^\dag_i -\frac{1}{2}\{L^\dag_iL_i,\rho\})\\
     &= -i(H_{\rm eff}\rho - \rho H_{\rm eff}^\dag) + \sum_i\gamma_i L_i\rho L^\dag_i. \label{eq: linbladian}
\end{align} 
Here, the system density matrix $\rho$ evolves under $H$, the system Hamiltonian, and $\{L_i\}$ is a set of jump operators corresponding to the dissipations due to the bath. In the second line, we have redefined $H_{\rm eff} = H - \frac{i}{2}\sum_i\gamma_i L^\dag L$. 

The last term in Eq.~\ref{eq: linbladian}, $\sum_i\gamma_i L_i\rho L^\dag_i$, describes quantum jump processes that can be associated with a measurable physical quantity, such as a spontaneous emission of a photon. The stochastic time evolution can be classified into different quantum trajectories depending on the number of quantum jumps. Now, if we postselect on the absence of quantum jumps, we end up with a Schrodinger-like time evolution with a non-Hermitian Hamiltonian: 
\begin{align}\label{eq: schrodinger}
    \frac{d\rho}{dt} 
     &= -i(H_{\rm eff}\rho - \rho H_{\rm eff}^\dag)
\end{align}

\subsection{Review of sequential circuits}

In a quantum system consisting of $n$ qubits, the sequential quantum circuits are defined as follows:
\begin{definition}[$\tau$-sequential quantum circuit]
    Given a local universal gate set $ \mathcal{G} \subseteq U(4)$, a quantum circuit consisting of local gates in the set is called a $\tau$-sequential quantum circuit if each qubit is at most acted upon by $\tau$ gates in the circuit.
\end{definition}
Various interesting examples emerge at $\tau\ll n$. For example, when $\tau$ is a constant, from the definition, we can bound its computational power:
\begin{enumerate}
    \item Strictly contains all constant depth circuits;
    \item Is strictly contained in the set of linear depth circuits.
\end{enumerate}
The first inclusion holds because any constant-depth circuit is a sequential circuit and sequential circuits can prepare long-range correlates that cannot be accessed from constant-depth circuits such as the GHZ state. The second bound can be deduced from a) these sequential circuits are at most depth $O(n)$, as their circuit volume is $O(n)$ by definition and b) they cannot generate volume law entanglement states, which are permitted in linear depth circuits.

As such, sequential circuits are of particular interest in the NISQ era: Compared with a constant depth circuit, it has more representation power as its light cone can cover the whole system; compared to a `dense' constant depth circuit, it permits a compressed representation yet accurately captures certain classes of states, such as thermal states of locally interacting spin chains~\cite{zhang2022qubit}, electronic mean-field ground states~\cite{niu2022holographic}, chiral topological orders~\cite{chen2024sequential2}, maps between gapped phases~\cite{chen2024sequential}, projected entangled-pair states~\cite{wei2022sequential}, etc.; they can be combined with adaptive circuits to improve their representation power~\cite{lu2022measurement,foss2023experimental,malz2024preparation}, and experimental proposals on cQED devices have been proposed~\cite{osborne2010holographic,zhang2024sequential}.

\paragraph{Matrix product states} 
The discussion of sequential circuits originated from a celebrated quantum state compression method: MPS. Any pure quantum state $|\Psi\>$ can be expressed as a matrix-product state by sequentially performing Schmidt decompositions between local sites, turning wave-function amplitudes into a 1D tensor train (in physical states considered in this work, we focus on open boundary 1D chains with sites $x\in\{1,2,\dots n\}$, although the method can be generalized to infinite systems and larger on-site dimensions):
\begin{align}
	|\Psi\> = \sum_{\{s_x\}_{x=1}^n} \ell^T A^{s_1}A^{s_2}\dots |s_1s_2\dots \>.
\end{align} 

In this context, $ s_x \in \{0,1\} $ labels the basis states for site $ x $, and $ A^{s_x} $ are $\chi \times \chi$ matrices. The vectors $\ell$ are $\chi$-dimensional and determine the left boundary conditions. The memory and computational cost of Matrix Product State (MPS) computations scale with the bond dimension, $\chi$, which is lower bounded by the bipartite entanglement entropy across a cut through the bond.

Although representing a generic quantum state, such as the Haar random state still requires the bond dimension $\chi = e^{O(n)}$, many quantum states of physical interests such as 1D short-range correlated, area-law entangled states~\cite{hastings2007area, foss2021entanglement, niu2022holographic}, or thermal mixed states~\cite{hauschild2018finding, zhang2022qubit}, it is possible to truncate the entanglement spectrum to $\chi = O(1)$ independent of system size, allowing for efficient classical simulations of 1D gapped ground states.

Higher-dimensional systems can also be represented as MPS by treating them as a 1D stack of $(d-1)$-dimensional cross-sections. Yet even for area-law entangled states the required bond dimension grows exponentially with the cross-sectional area, making classical simulation impractical. In such cases including classically intractable cases such as 2D and 3D ground-states with symmetry-breaking~\cite{niu2022holographic} or (non-chiral) topological order~\cite{soejima2020isometric}, and finite-time quantum dynamics from any qMPS~\cite{foss2021holographic,chertkov2021holographic}, applying tensor network methods on a quantum computer could offer a significant advantage. 
\paragraph{Sequential quantum circuits generated by MPS (qMPS)}
Properties of any MPS in right-canonical form (RCF)~\cite{perez2006matrix} can be measured by sampling on a quantum computer and implementing its transfer-matrix as a quantum channel~\cite{gyongyosi2012properties} acting on $1$ ``physical" qubit and $q=\log_2\chi$ bond qubits~(see Fig.~2 for a graphical representation). 
Each tensor $A$ is then embedded as a block of larger unitary operator $U_A$ acting on a fixed initial state $|0\>$ of the physical qubits. After the application of $U_A$, the physical qubits can be measured in any desired basis while entanglement information is consistently stored on the bond qubits. 
The process is then repeated for each site in sequence from left to right. In this way, one can measure any product operator of the form $\prod_{x=1}^n \mathcal{O}_x$, which forms a complete basis for general observables. 
Crucially, once measured, the physical qubit for site $x$ can be reset to $|0\>$ and reused as the physical qubits for site $x+1$, enabling a small quantum processor to achieve quantum simulation tasks with sizes far larger than the number of qubits available ~\cite{foss2021holographic}.

To summarize, the qMPS procedure for sampling an observable of the form $\<\psi|\prod_{x=1}^n O_x|\psi\>$ is:

\fbox{%
\begin{minipage}{1\linewidth} % Adjust width as needed
 \begin{algorithm}[H]
	\caption{-- Generating sequential quantum circuits from MPS} 
 \begin{algorithmic}[1]
\State Prepare the bond qubits in a state corresponding to the left boundary vector $\ell$. This can be done with up to $\log \chi$ ancilla qubits
 \For {$x=1\dots n$}
\State Perform a synthesized quantum circuit $U_A$ at site $x$, entangling the physical and bond qubits. 
\State Measure the physical qubit in the eigenbasis of $O_x$ and weight the measurement outcome by the corresponding eigenvalue of that observable.
\State Reset the physical qubit for site $x$ in a fixed reference state, $|0\>$.
\EndFor
\State Discard the bond qubits.
 \end{algorithmic}
\end{algorithm}
\end{minipage}}

Moreover, the entanglement spectrum of the bond-qubits in between sites $x$ and $x+1$ coincides with the bipartite entanglement spectrum of the physical MPS at that entanglement cut, further enabling measurement of non-local entanglement observables, as recently demonstrated experimentally~\cite{foss2021holographic}. The left boundary vector $\ell$ is prepared by a unitary circuit acting on the bond space. For an open chain, there is no need to specify the right boundary condition as no entanglement is to be stored beyond $x = n$, and thus the bond qubits are disentangled with the physical qubit and can be traced out. When all $U_A$ matrices are set to be the same, the lack of right boundary conditions in the formalism describes a semi-infinite wire~\cite{foss2021holographic}. 

%qMPS methods enable qubit-efficient access to a subset of MPS with exponentially large bond dimension (in qubit number), including classically intractable cases such as 2D and 3D ground-states with symmetry-breaking~\cite{niu2022holographic} or (non-chiral) topological order~\cite{soejima2020isometric}, and finite-time quantum dynamics from any qMPS~\cite{foss2021holographic,chertkov2021holographic}.

By exploiting the efficient compression~\cite{orus2019tensor} of physically interesting states, such as low-energy states of local Hamiltonians, qTNS methods enable simulation of many-body systems relevant to condensed-matter physics and materials science with much smaller quantum memory than would be required to directly encode the many-body wave-function.

\paragraph{Gaussian MPS}%
While MPS is a generic approach to quantum state compression, a subclass of MPS, the Gaussian MPS (GMPS), explores the near-area law entanglement of free fermion systems, enabling even more efficient representations. The Hamiltonian of a Gaussian (i.e. non-interacting) fermion system with $ n $ sites has the form 
\begin{align}
    H = \sum_{i,j=1}^{n} c^\dagger_i h_{ij} c^{\vphantom\dagger}_j
\end{align} 
This system can be fully characterized by its $ n \times n $ two-point Green's function $ G_{ij} = \langle c^\dagger_i c^{\vphantom\dagger}_j \rangle $ with highly degenerate eigenvalues of either 0 (unoccupied) or 1 (occupied sites). Crucially, $G_{ij}$ remains invariant under any unitary transformation as long as one does not mix occupied and unoccupied states.

The compression scheme presented by Fishman and White~\cite{fishman2015compression} exploits this unitary invariance by progressively disentangling local degrees of freedom in blocks of $ B $ adjacent sites. Moreover, the ground states of Gaussian fermionic systems have near-area-law entanglement. Therefore, choosing block size $ B $ large enough, most of the block eigenvalues must be exponentially close to 0 or 1. This enables a sequence of operations that compresses the correlation matrix:
\fbox{%
\begin{minipage}{1\linewidth}
 \begin{algorithm}[H]
	\caption{-- GMPS compression} 
 \begin{algorithmic}[1]
 \For {$x=1\dots n$}
\State Start with $G_{xx}$, examine the next $B \times B $ sub-matrix of $ G $ to the bottom right.
\State Identify the eigenvector with an eigenvalue closest to 0 or 1. 
\State Apply a series of $ 2 \times 2 $ single-particle unitary rotations $\prod_{\alpha=1}^{B-1} u_{x,\alpha}^\dagger $ on $G$ to move this eigenvector to the first site of the block. Denote the resultant correlation matrix as $G'$.
\State Set $G\leftarrow G'$
\EndFor
 \end{algorithmic}
\end{algorithm}
\end{minipage}}
%The process begins with the upper-left $ B \times B $ block of $ G $, where the eigenvector with an eigenvalue closest to 0 or 1 can be identified. A series of $ 2 \times 2 $ single-particle unitary rotations is then performed to move this eigenvector to the first site of the block, effectively separating the site from the rest of the system. This procedure is iterated over the remaining $(n-1) \times (n-1)$ sites until the Green's function is approximately diagonalized.

%The composition of these basis rotations results in an $ n \times n $ unitary matrix, $ u^\dagger = (\prod_{\alpha=1}^{(B-1)(N_o - \frac{B}{2})} u_\alpha)^\dagger $, which is made up of $ 2 \times 2 $ single-particle unitaries indexed by $ \alpha $. This unitary approximately diagonalizes the Green's function. 

In each iteration, the compression algorithm separates a site from the rest of the system. At the end of execution, Green's function is approximately diagonalized. Reversing the process allows one to transform a product state into the desired Gaussian fermionic state. 

\paragraph{GMPS as a sequential quantum circuit} To simulate fermions on a quantum computer, these single-particle operations $u_\alpha$ in the fermionic language, can be converted into a circuit for the many-particle Hilbert space of size $ 2^n $ by replacing each $ 2 \times 2 $ unitary $ u_{x,\alpha}^\dagger $ with a two-qubit gate: 
\begin{align}
    U_{x,\alpha} = \exp[\sum_{ij} c^\dagger_i (\log u_{x,\alpha}^\dagger)_{ij} c^{\vphantom\dagger}_j] = \exp[\sum_{ij} \sigma^+_i (\log u_{x,\alpha}^\dagger)_{ij} \sigma^-_j] .
\end{align}

As pictured in Fig.~1(b), the resulting ladder circuit $ U = \prod_{x,\alpha} U_{x,\alpha} $ can be interpreted as a $B$-sequential quantum circuit generated by MPS (Fig.~2) with bond dimension $ \chi = 2^B $ whose causal cone slightly differentiates from generic non-Gaussian (q)MPS of the same bond dimension. In the 1D Hamiltonian considered in Eq.~5 and $J_z = 0$, we numerically find that it suffices to choose block size $B = 3$ to prepare the ground state to energy infidelity $<1\%$ on a $n=18$ chain. 

Directly preparing an arbitrary Gaussian state with a ladder circuit on $ n $ qubits requires $ O(n^2) $ two-qubit gates (as shown in~\cite{arute2020hartree}). To compare, a compressed GMPS ground state can be prepared with $ O(nB) $ two-qubit gates acting on $ O(B) $ qubits when implemented sequentially. The efficiency of this compression depends on the block size $ B $ needed for accurate state approximation. 

Numerical evidence and entanglement-based arguments suggest that for ground states of local Hamiltonians in 1D systems of length $ L $ with a target error threshold $ \epsilon = 1 - \frac{1}{L} \sum_{i,j} |G_{ij}^{(\text{GMPS})} - G_{ij}| $, the required block size $ B $ scales as $ \log \epsilon^{-1} $ for gapped states or $ \log L \log \epsilon^{-1} $ for gapless metallic states. In Niu et al.~\cite{niu2022holographic}, these results are extended to $ d $-dimensional systems, where $ B $ generally scales with the bipartite entanglement entropy $ S(L) $:

\begin{align}
B\sim \begin{cases}
 L^{d-1}\log \epsilon^{-1} & \text{gapped} \\
  L^{d-1}\log L\log \epsilon^{-1} & \text{gapless} 
  \end{cases}
\end{align}

This result holds even for topologically non-trivial Chern band insulators that have an obstruction to forming a fully localized Wannier-basis. Compared to standard adiabatic state-preparation protocols, this method dramatically reduces the number of qubits required ($L^{d-1}B$ vs. $L^d$) to implement the GMPS on a quantum computer~\cite{niu2022holographic}; the compressed state can be then used as an initial state for quantum quench or adiabatic state preparations. 

\subsection{Analytical proof of Thm.~I.1 and implication in complexity theory}

In this section, we provide analytical proof of Thm.~I.1, which says there exists a non-Hermitian $H$ consisting of single qubit terms only and an initial state $\ket{\psi_{0}}$ where the dynamics for $\Theta(\log (n))$ time can be exponentially hard to approximate with a quantum circuit. One explicit example is to take $\ket{\psi_{0}}$ to be a Haar random state $\ket{\psi_{\rm Haar}}$ and $H_{z} =-i\sum_i Z_i$. Our hardness of approximation result comes from two facts:
\begin{enumerate}
    \item Applying the dynamics generated by $H_{z}$ for evolution time\ $t = \Theta(\log(n))$ and measure in the computational basis allows one to distinguish a Haar random state from a maximally mixed state $\rho_m = \mathbbm{I}/d$, where $d = 2^n$.
    \item With overwhelmingly high probability, distinguishing a state sampled from Haar random ensemble from $\rho_m $ is exponentially hard.
\end{enumerate}
The task can be formalized into a decision problem: 
\begin{problem}[Distinguishes a Haar random state from a maximally mixed state]\label{prob: Haar}
An oracle $\mathcal{O}$ prepares copies of a fixed quantum state on a $n$-qubit quantum register that is promised to be either a maximally mixed state or a pure state sampled Haar randomly. Decide which is the case with as few queries as possible.
\end{problem}
We first prove that time evolution with a single qubit dissipation channel can efficiently solve Prob.~\ref{prob: Haar}.
\begin{lemma}[Single qubit dissipation distinguishes a Haar random state from a maximally mixed state] \label{lem:easy}
With $O(k)$ queries and probably $1-e^{-O(k)}$, a quantum state sampled Haar randomly can be distinguished from a maximally mixed state by time evolving the state for $t = \Theta(\log(n))$ with $H_z$. 

%\begin{align}\label{eq:trace_distance}
%    \Tr{\ketbra{\psi_{\rm Haar}(t)}-\rho_m(t))}| > 0.1. 
%\end{align}
%$\rho(t)$ and $\ket{\psi(t)}$ denotes normalized density matrices and states after time evolution.
\end{lemma}
\begin{proof}
    We consider the output distributions of a Haar random state before and after time evolution, measured on the computational basis \begin{align}
        p_s := \left|\braket{s}{\psi_{\rm Haar}}\right|^2,\ q_s:= \left|\braket{s}{\psi_{\rm Haar}(t)}\right|^2
    \end{align}
    $p_s$ is known to fluctuate around its mean value i.e. the distribution of $p_s = 1/d$, but with a variance exponentially small in $n$. This can be seen from the numerical experiments on the left panel of Fig.~\ref{fig:output}. As a result, it is exponentially hard to distinguish a Haar random state and a maximally mixed state with a naive computational basis measurement.

    The effect of the time evolution generated by $H_z$ is that it exponentially re-distributes weights of all output strings according to their Hamming weight $w_s$, or the number of 0's in $s$. This can already be told from a single qubit case: $$e^{-Zt} (a\ket{0}+b\ket{1}) = ae^{-t}\ket{0} + be^t\ket{1}$$.
    
    We begin by finding $t$, such that after applying the evolution generated by $H_z$ to the maximally mixed state, the output weight on the $s=1^n$  can be a constant $c$. By setting:
    \begin{align}\label{eq:weight}
        \left(\frac{e^{2t}}{e^{-2t}+e^{2t}}\right)^{n} = c
    \end{align}
    
    and solving for $t$, we get 
    \begin{align}
        t &= -\frac{1}{4} \log(-c^{-1/ n} (-1 + c^{1/ n}))\\
          & = -\frac{1}{4} \log(c^{-1/ n}-1)\\
          & = -\frac{1}{4} \left[\log(-\frac{1} {n} \log (c))+O(n^{-1}))\right]\\
          & = \Theta(\log(n))
    \end{align}
    What would the output distribution look like for a Haar random state after applying the same time evolution? The normalized weight is just %For states sampled from Haar randomly, the output weight on any fixed string (in this case, $s = 1^n$) is known to follow the Porter-Thomas distribution~\cite{}: 
    %$${\rm PT}(p_s) = 2^n e^{2^n p_s}$$

    \begin{align}
        q_{1^n}  &=\frac{p_{1^n} e^{2n}}{\sum_s p_s e^{4w_s-2n}}\\
       &=\frac{p_{1^n} e^{2n}}{\sum_{i=0}^{n}\sum_{w_s = i}p_s e^{4i-2n}}\\
       &=\frac{p_{1^n} e^{2n}}{p_{1^n} e^{2n}+ p_{0^n} e^{-2n}+\sum_{i=1}^{n-1}2^{-n}\begin{pmatrix}
           n\\
           i
       \end{pmatrix}p_s e^{4i-2n}} \\
       &:=\frac{p_{1^n}}{p_{1^n}+2^{-n}\eta }
       \label{eq:weight_haar}
    \end{align}
    %$\Tr{\ketbra{\psi_{\rm Haar}(t)}-\rho_m(t))}| > 0.1$ 
    In the second from last step, we use the fact that each entry of a Haar state vector consists of two standard normal variables and the weights on each string should follow the Porter-Thomas distribution: ${\rm PT}(p_s) = 2^n e^{-2^n p_s}$, and a sum over $\poly(n)$ terms would converge to its mean value. We may also ignore the $p_{0^n} e^{-2n}$ in the denominator as it is exponentially small in $n$. 

    Under the assumption of Eq.~\ref{eq:weight}, we have 
    \begin{align}
        \frac{1}{1+\eta}= c
    \end{align}

     W.l.o.g., setting $c = 1/2$ gives $\eta  = 1$; therefore the Eq.~\ref{eq:weight_haar} returns
    \begin{align}
        q_{1^n} = \frac{2^{n}p_{1^n}}{2^{n}p_{1^n}+1},
    \end{align}
     
     This means the variance of this distribution is now a constant independent of $n$. For example, the probability of $P[q_{1^n}\geq0.6]$ is 
     \begin{align}
         P[q_{1^n}\geq0.6] &= P[q_{1^n}\geq0.6]\\
         & = P\left[\frac{2^{n}p_{1^n}}{2^{n}p_{1^n}+1}\geq0.6\right]\\
         & = P\left[p_{1^n}\geq 1.5\times2^{-n}\right]\\
         & = \int_{1.5\times2^{-n}}^\infty 2^n e^{-2^n p} dp\\
         & \approx 0.223
     \end{align}
     As evident from Fig.~\ref{fig:output}, with constant probability, a measurement in the computational basis now allows efficient distinguishing between the Haar random state and the maximally mixed state after the time evolution. Notice that, if we choose $t = \poly(n)$, the two states are once again hard to be distinguished because they both purify to a product state. This probability can be boosted to $1-e^{-k}$ by randomly applying a layer of $X$ gates before the time evolution (namely, randomly selecting a string $s$ to amplify and probe its amplitude) and repeating $O(k)$ times. 
\end{proof}
\begin{figure}
    \centering
        \centering
        \includegraphics[width=.75\linewidth]{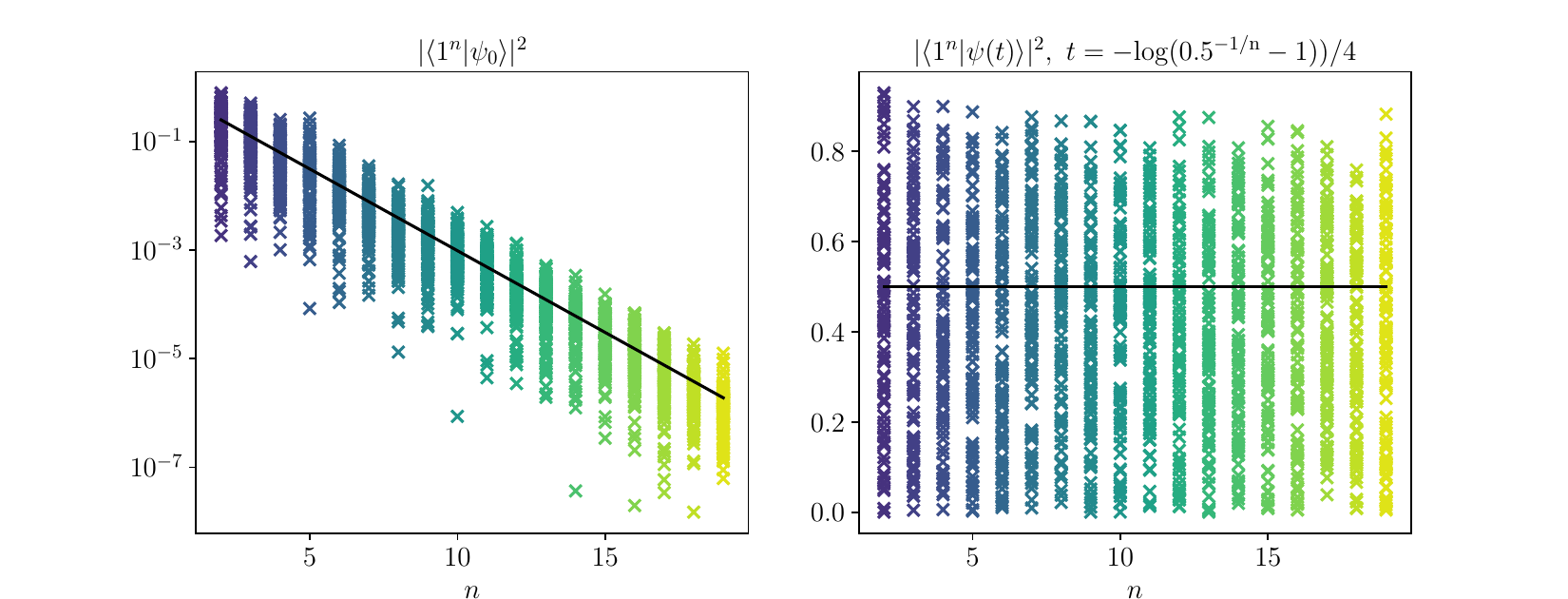}
    \caption{Output weight on the $1^n$ string of 100 Haar random states before (left) and after (right) applying time evolution generated by $H_z$ for a time length $t$; the output weight of a maximally-mixed state, $\rho_m$ is plotted with a black line for comparison. In both cases, the output weight of the Haar states fluctuates around the output value of the maximally mixed state. However, the variance on the left plot is exponentially small in $n$, and the variance after the non-Hermitian time evolution becomes a constant, allowing efficient distinguishment between the Haar random state and the maximally mixed state. This cannot be accomplished with any local quantum channels such as the amplitude damping channel~\cite{fefferman2023effect}.}
    \label{fig:output}
\end{figure}
In the unitary setting, the minimum number of gates required to implement a measurement $M$ that can distinguish a given quantum state $\ket{\psi}$ from the maximally mixed state $\rho_m$ to a certain resolution $\delta$ is defined as the strong state complexity. Formally, let $\beta(r, \ket{\psi})$ be the maximum bias with which $\ket{\psi}$ can be distinguished from the maximally mixed state using a circuit with at most $r$ gates from the gate set $\mathcal{G} \subseteq U(4)$:
\begin{align}
    \beta(r,\ket{\psi}) =& \text{ max}\ |\Tr{M(\ketbra{\psi}{\psi}-\rho_m)}|\\
    &\text {s.t. }M \text{ can be implemented with at most } r \text{ gates}
\end{align}

\noindent Then the strong state complexity is defined as:
\begin{definition}[Strong state complexity]\label{def:sc_state}
    For a given $r \in \mathbb{N}$ and $\eta \in (0, 1)$, a pure state $\ket{\psi}$ has strong $\eta$-state
complexity at most $r$ if and only if
$\beta(r,\ket{\psi}) \geq 1 - 1/d - \delta$.
We denote this $\mathcal{C}_\delta (\ket{\psi}) \leq r$.
\end{definition}
The $1/d$ in the definition comes from that any pure state can be trivially distinguished from a maximally mixed state with trace distance $1/d$. For Haar random pure states, it has been shown that the vast majority of them have exponentially high strong state complexity:
\begin{proposition}[Strong state complexity of Haar random states]
    The probability that $\ket{\psi_{\rm Haar}}$ has strong circuit complexity less than $r$ is 
    \begin{align}
        \Pr [\mathcal{C}_{\delta}(\ket{\psi_{\rm Haar}})\leq r ] \leq 4.0144 d(n + 1)^r|\mathcal{G}|^r \exp{-\frac{d(1-\delta^2)}{9\pi^3}}
    \end{align}
\end{proposition}
The proof given by~\cite{brandao2021models} essentially comes from a counting argument that any fixed measurement can only distinguish a small number of states, and thus to distinguish the vast majority of Haar random states, the complexity of measurement must grow exponential in $n$. Even for $1-\delta^2 = 1/\poly(n)$, this probability of distinguish remains exponentially small in $n$ when $r = \poly (n)$. This means 
\begin{lemma}[Hardness of distinguishing a Haar random state with a $\mathsf{BQP}$ machine]\label{lem:hard}
    With overwhelming probability, a $\mathsf{BQP}$ machine making polynomial queries to $\mathcal{O}$ cannot solve Prob.~\ref{prob: Haar}.
\end{lemma}
Combining Lem.~\ref{lem:easy} and Lem.~\ref{lem:hard} completes the proof of Thm.~I.1.
Further, as we explain in Supplementary Note 1, each single-qubit dissipation channel can be implemented with a two-qubit gate and post-selection on the single ancillary qubit. It turns out that the result we have proved can be summarized in a computational complexity language: With overwhelming probability, Prob.~\ref{prob: Haar} can be solved by a quantum circuit at merely constant depth when postselection is allowed (denoted as $\mathsf{PostQNC_0}$), but also with overwhelming probability, it cannot be solved by a polynomial-size quantum circuit. Namely,
\begin{corollary}[Oracle separation]
    There exists a quantum oracle $\mathcal{O}$ (as defined in Prob.~\ref{prob: Haar}) such that $\mathsf{PostQNC_0}^\mathcal{O}\not\subseteq\mathsf{BQP}^\mathcal{O}$;  relative to the same oracle, $\mathsf{BQP}^\mathcal{O} \subset\mathsf{PostBQP}^\mathcal{O}$.
\end{corollary}
%\YZ{finish proof}

\section{Data Availability}
The generated in this study have been deposited in~\cite{yuxuan_zhang_2025_14834474}.
\section{Code Availability}
The simulation code used in this work can be found at~\cite{yuxuan_zhang_2025_14834474}.
\section{Acknowledgement}
We thank Sarang Gopalakrishnan, Tim Hsieh, Lin Lin, Dvira Segal, Nathan Wiebe, Roeland Wiersema, Cenke Xu, and Yijian Zou for their insightful discussions. We extend gratitude to Michael Foss-Feig and Peter Groszkowski for their invaluable assistance in accessing Quantinuum's hardware resources and thank the reviewers for their valuable feedback. J.C. acknowledges support from the Shared Hierarchical Academic Research Computing Network (SHARCNET), Compute Canada, and the Canadian Institute for Advanced Research (CIFAR) AI chair program. Y.Z., J.C., and Y.B.K. were supported by the Natural Science and Engineering Research Council (NSERC) of Canada. Y.Z. and Y.B.K. acknowledge support from the Center for Quantum Materials at the University of Toronto. Y.Z. was further supported by a CQIQC fellowship at the University of Toronto, and in part by grant NSF PHY-2309135 to the Kavli Institute for Theoretical Physics (KITP). This research used resources of the Oak Ridge Leadership Computing Facility, which is a DOE Office of Science User Facility supported under Contract DE-AC05-00OR22725, and, in part, by the Province of Ontario, the Government of Canada through CIFAR, and companies sponsoring the Vector Institute~\url{www.vectorinstitute.ai/#partners}. 
\section{Author Contributions}
J.C., Y.B.K. and Y.Z. conceived the project and designed the theoretical framework. J.C. and Y.B.K. supervised the project and contributed crucial insights on interpreting the results. Y.Z. performed numerical simulations, executed the quantum circuits, and analyzed the results. Y.Z. wrote the manuscript with input from all authors.

\section{Competing Interests}
The authors declare no competing interests.
\appendix
\section{Methods}\label{app:methods}
This section details the numerical and analytical methods used throughout the work. We give a motivating example of non-Hermitian physics, followed by the discussion of sequential circuits generated by Matrix Product States (MPS) and Gaussian Matrix Product States (GMPS), explaining how they enable efficient quantum representations of near-area-law quantum states. %Next, we highlight the importance of using a proper warm start in all numerical simulations. 
Additionally, we provide an analytical proof of Theorem~I.1.
\subsection{Non-Hermitian physics: a motivating example}
We argued in the main text that a physical way to motivate non-Hermitian physics is from open quantum system dynamics. Usually, the Linbladian equation is used to describe the Markovian evolution of a system interacting with a thermal bath~\cite{manzano2020short}:

\begin{align}
    \frac{d\rho}{dt} &= -i[H,\rho] +\sum_i\gamma_i(L_i\rho L^\dag_i -\frac{1}{2}\{L^\dag_iL_i,\rho\})\\
     &= -i(H_{\rm eff}\rho - \rho H_{\rm eff}^\dag) + \sum_i\gamma_i L_i\rho L^\dag_i. \label{eq: linbladian}
\end{align} 
Here, the system density matrix $\rho$ evolves under $H$, the system Hamiltonian, and $\{L_i\}$ is a set of jump operators corresponding to the dissipations due to the bath. In the second line, we have redefined $H_{\rm eff} = H - \frac{i}{2}\sum_i\gamma_i L^\dag L$. 

The last term in Eq.~\ref{eq: linbladian}, $\sum_i\gamma_i L_i\rho L^\dag_i$, describes quantum jump processes that can be associated with a measurable physical quantity, such as a spontaneous emission of a photon. The stochastic time evolution can be classified into different quantum trajectories depending on the number of quantum jumps. Now, if we \emph{postselect} on the absence of quantum jumps, we end up with a Schrodinger-like time evolution with a non-Hermitian Hamiltonian: 
\begin{align}\label{eq: schrodinger}
    \frac{d\rho}{dt} 
     &= -i(H_{\rm eff}\rho - \rho H_{\rm eff}^\dag)
\end{align}

\subsection{Review of sequential circuits}

In a quantum system consisting of $n$ qubits, the sequential quantum circuits are defined as follows:
\begin{definition}[$\tau$-sequential quantum circuit]
    Given a local universal gate set $ \mathcal{G} \subseteq U(4)$, a quantum circuit consisting of local gates in the set is called a $\tau$-sequential quantum circuit if each qubit is at most acted upon by $\tau$ gates in the circuit.
\end{definition}
Various interesting examples emerge at $\tau\ll n$. For example, when $\tau$ is a constant, from the definition, we can bound its computational power:
\begin{enumerate}
    \item Strictly contains all constant depth circuits;
    \item Is strictly contained in the set of linear depth circuits.
\end{enumerate}
The first inclusion holds because any constant-depth circuit is a sequential circuit and sequential circuits can prepare long-range correlates that cannot be accessed from constant-depth circuits such as the GHZ state. The second bound can be deduced from a) these sequential circuits are at most depth $O(n)$, as their circuit volume is $O(n)$ by definition and b) they cannot generate volume law entanglement states, which are permitted in linear depth circuits.

As such, sequential circuits are of particular interest in the NISQ era: Compared with a constant depth circuit, it has more representation power as its light cone can cover the whole system; compared to a `dense' constant depth circuit, it permits a compressed representation yet accurately captures certain classes of states, such as thermal states of locally interacting spin chains~\cite{zhang2022qubit}, electronic mean-field ground states~\cite{niu2022holographic}, chiral topological orders~\cite{chen2024sequential2}, maps between gapped phases~\cite{chen2024sequential}, projected entangled-pair states~\cite{wei2022sequential}, etc.; they can be combined with adaptive circuits to improve their representation power~\cite{lu2022measurement,foss2023experimental,malz2024preparation}, and experimental proposals on cQED devices have been proposed~\cite{osborne2010holographic,zhang2024sequential}.

\paragraph{Matrix product states} 
The discussion of sequential circuits originated from a celebrated quantum state compression method: MPS. Any pure quantum state $|\Psi\>$ can be expressed as a matrix-product state by sequentially performing Schmidt decompositions between local sites, turning wave-function amplitudes into a 1D tensor train~\footnote{In physical states considered in this work, we focus on open boundary 1D chains with sites $x\in\{1,2,\dots n\}$, although the method can be generalized to infinite systems and larger on-site dimensions.}:
\begin{align}
	|\Psi\> = \sum_{\{s_x\}_{x=1}^n} \ell^T A^{s_1}A^{s_2}\dots |s_1s_2\dots \>.
\end{align} 

In this context, $ s_x \in \{0,1\} $ labels the basis states for site $ x $, and $ A^{s_x} $ are $\chi \times \chi$ matrices. The vectors $\ell$ are $\chi$-dimensional and determine the left boundary conditions. The memory and computational cost of Matrix Product State (MPS) computations scale with the bond dimension, $\chi$, which is lower bounded by the bipartite entanglement entropy across a cut through the bond.

Although representing a generic quantum state, such as the Haar random state still requires the bond dimension $\chi = e^{O(n)}$, many quantum states of physical interests such as 1D short-range correlated, area-law entangled states~\cite{hastings2007area, foss2021entanglement, niu2022holographic}, or thermal mixed states~\cite{hauschild2018finding, zhang2022qubit}, it is possible to truncate the entanglement spectrum to $\chi = O(1)$ independent of system size, allowing for efficient classical simulations of 1D gapped ground states.

Higher-dimensional systems can also be represented as MPS by treating them as a 1D stack of $(d-1)$-dimensional cross-sections. Yet even for area-law entangled states the required bond dimension grows exponentially with the cross-sectional area, making classical simulation impractical. In such cases including classically intractable cases such as 2D and 3D ground-states with symmetry-breaking~\cite{niu2022holographic} or (non-chiral) topological order~\cite{soejima2020isometric}, and finite-time quantum dynamics from any qMPS~\cite{foss2021holographic,chertkov2021holographic}, applying tensor network methods on a quantum computer could offer a significant advantage. 
\paragraph{Sequential quantum circuits generated by MPS (qMPS)}
Properties of any MPS in right-canonical form (RCF)~\cite{perez2006matrix} can be measured by sampling on a quantum computer and implementing its transfer-matrix as a quantum channel~\cite{gyongyosi2012properties} acting on $1$ ``physical" qubit and $q=\log_2\chi$ bond qubits~(see Fig.~2 for a graphical representation). 
Each tensor $A$ is then embedded as a block of larger unitary operator $U_A$ acting on a fixed initial state $|0\>$ of the physical qubits. After the application of $U_A$, the physical qubits can be measured in any desired basis while entanglement information is consistently stored on the bond qubits. 
The process is then repeated for each site in sequence from left to right. In this way, one can measure any product operator of the form $\prod_{x=1}^n \mathcal{O}_x$, which forms a complete basis for general observables. 
Crucially, once measured, the physical qubit for site $x$ can be reset to $|0\>$ and reused as the physical qubits for site $x+1$, enabling a small quantum processor to achieve quantum simulation tasks with sizes far larger than the number of qubits available ~\cite{foss2021holographic}.

To summarize, the qMPS procedure for sampling an observable of the form $\<\psi|\prod_{x=1}^n O_x|\psi\>$ is:
 \begin{algorithm}[H]
	\caption{Generating sequential quantum circuits from MPS} 
 \begin{algorithmic}[1]
\State Prepare the bond qubits in a state corresponding to the left boundary vector $\ell$. This can be done with up to $\log \chi$ ancilla qubits
 \For {$x=1\dots n$}
\State Perform a synthesized quantum circuit $U_A$ at site $x$, entangling the physical and bond qubits. 
\State Measure the physical qubit in the eigenbasis of $O_x$ and weight the measurement outcome by the corresponding eigenvalue of that observable.
\State Reset the physical qubit for site $x$ in a fixed reference state, $|0\>$.
\EndFor
\State Discard the bond qubits.
 \end{algorithmic}
\end{algorithm}

Moreover, the entanglement spectrum of the bond-qubits in between sites $x$ and $x+1$ coincides with the bipartite entanglement spectrum of the physical MPS at that entanglement cut, further enabling measurement of non-local entanglement observables, as recently demonstrated experimentally~\cite{foss2021holographic}. The left boundary vector $\ell$ is prepared by a unitary circuit acting on the bond space. For an open chain, there is no need to specify the right boundary condition as no entanglement is to be stored beyond $x = n$, and thus the bond qubits are disentangled with the physical qubit and can be traced out. When all $U_A$ matrices are set to be the same, the lack of right boundary conditions in the formalism describes a semi-infinite wire~\cite{foss2021holographic}. 

%qMPS methods enable qubit-efficient access to a subset of MPS with exponentially large bond dimension (in qubit number), including classically intractable cases such as 2D and 3D ground-states with symmetry-breaking~\cite{niu2022holographic} or (non-chiral) topological order~\cite{soejima2020isometric}, and finite-time quantum dynamics from any qMPS~\cite{foss2021holographic,chertkov2021holographic}.

By exploiting the efficient compression~\cite{orus2019tensor} of physically interesting states, such as low-energy states of local Hamiltonians, qTNS methods enable simulation of many-body systems relevant to condensed-matter physics and materials science with much smaller quantum memory than would be required to directly encode the many-body wave-function.

\paragraph{Gaussian MPS}%
While MPS is a generic approach to quantum state compression, a subclass of MPS, the Gaussian MPS (GMPS), explores the near-area law entanglement of free fermion systems, enabling even more efficient representations. The Hamiltonian of a Gaussian (i.e. non-interacting) fermion system with $ n $ sites has the form 
$$ H = \sum_{i,j=1}^{n} c^\dagger_i h_{ij} c^{\vphantom\dagger}_j $$ 
This system can be fully characterized by its $ n \times n $ two-point Green's function $ G_{ij} = \langle c^\dagger_i c^{\vphantom\dagger}_j \rangle $ with highly degenerate eigenvalues of either 0 (unoccupied) or 1 (occupied sites). Crucially, $G_{ij}$ remains invariant under any unitary transformation as long as one does not mix occupied and unoccupied states.

The compression scheme presented by Fishman and White~\cite{fishman2015compression} exploits this unitary invariance by progressively disentangling local degrees of freedom in blocks of $ B $ adjacent sites. Moreover, the ground states of Gaussian fermionic systems have near-area-law entanglement. Therefore, choosing block size $ B $ large enough, most of the block eigenvalues must be exponentially close to 0 or 1. This enables a sequence of operations that compresses the correlation matrix:
 \begin{algorithm}[H]
	\caption{GMPS compression} 
 \begin{algorithmic}[1]
 \For {$x=1\dots n$}
\State Start with $G_{xx}$, examine the next $B \times B $ sub-matrix of $ G $ to the bottom right.
\State Identify the eigenvector with an eigenvalue closest to 0 or 1. 
\State Apply a series of $ 2 \times 2 $ single-particle unitary rotations $\prod_{\alpha=1}^{B-1} u_{x,\alpha}^\dagger $ on $G$ to move this eigenvector to the first site of the block. Denote the resultant correlation matrix as $G'$.
\State Set $G\leftarrow G'$
\EndFor
 \end{algorithmic}
\end{algorithm}
%The process begins with the upper-left $ B \times B $ block of $ G $, where the eigenvector with an eigenvalue closest to 0 or 1 can be identified. A series of $ 2 \times 2 $ single-particle unitary rotations is then performed to move this eigenvector to the first site of the block, effectively separating the site from the rest of the system. This procedure is iterated over the remaining $(n-1) \times (n-1)$ sites until the Green's function is approximately diagonalized.

%The composition of these basis rotations results in an $ n \times n $ unitary matrix, $ u^\dagger = (\prod_{\alpha=1}^{(B-1)(N_o - \frac{B}{2})} u_\alpha)^\dagger $, which is made up of $ 2 \times 2 $ single-particle unitaries indexed by $ \alpha $. This unitary approximately diagonalizes the Green's function. 

In each iteration, the compression algorithm separates a site from the rest of the system. At the end of execution, Green's function is approximately diagonalized. Reversing the process allows one to transform a product state into the desired Gaussian fermionic state. 

\paragraph{GMPS as a sequential quantum circuit} To simulate fermions on a quantum computer, these single-particle operations $u_\alpha$ in the fermionic language, can be converted into a circuit for the many-particle Hilbert space of size $ 2^n $ by replacing each $ 2 \times 2 $ unitary $ u_{x,\alpha}^\dagger $ with a two-qubit gate: $$ U_{x,\alpha} = \exp[\sum_{ij} c^\dagger_i (\log u_{x,\alpha}^\dagger)_{ij} c^{\vphantom\dagger}_j] = \exp[\sum_{ij} \sigma^+_i (\log u_{x,\alpha}^\dagger)_{ij} \sigma^-_j] .$$ 

As pictured in Fig.~1(b), the resulting ladder circuit $ U = \prod_{x,\alpha} U_{x,\alpha} $ can be interpreted as a $B$-sequential quantum circuit generated by MPS (Fig.~2) with bond dimension $ \chi = 2^B $ whose causal cone slightly differentiates from generic non-Gaussian (q)MPS of the same bond dimension. In the 1D Hamiltonian considered in Eq.~5 and $J_z = 0$, we numerically find that it suffices to choose block size $B = 3$ to prepare the ground state to energy infidelity $<1\%$ on a $n=18$ chain. 

Directly preparing an arbitrary Gaussian state with a ladder circuit on $ n $ qubits requires $ O(n^2) $ two-qubit gates (as shown in~\cite{arute2020hartree}). To compare, a compressed GMPS ground state can be prepared with $ O(nB) $ two-qubit gates acting on $ O(B) $ qubits when implemented sequentially. The efficiency of this compression depends on the block size $ B $ needed for accurate state approximation. 

Numerical evidence and entanglement-based arguments suggest that for ground states of local Hamiltonians in 1D systems of length $ L $ with a target error threshold $ \epsilon = 1 - \frac{1}{L} \sum_{i,j} |G_{ij}^{(\text{GMPS})} - G_{ij}| $, the required block size $ B $ scales as $ \log \epsilon^{-1} $ for gapped states or $ \log L \log \epsilon^{-1} $ for gapless metallic states. In Niu et al.~\cite{niu2022holographic}, these results are extended to $ d $-dimensional systems, where $ B $ generally scales with the bipartite entanglement entropy $ S(L) $:

\begin{align}
B\sim \begin{cases}
 L^{d-1}\log \epsilon^{-1} & \text{gapped} \\
  L^{d-1}\log L\log \epsilon^{-1} & \text{gapless} 
  \end{cases}
\end{align}

This result holds even for topologically non-trivial Chern band insulators that have an obstruction to forming a fully localized Wannier-basis. Compared to standard adiabatic state-preparation protocols, this method dramatically reduces the number of qubits required ($L^{d-1}B$ vs. $L^d$) to implement the GMPS on a quantum computer~\cite{niu2022holographic}; the compressed state can be then used as an initial state for quantum quench or adiabatic state preparations. 

\subsection{Analytical proof of Thm.~I.1 and implication in complexity theory}

In this section, we provide analytical proof of Thm.~I.1, which says there exists a non-Hermitian $H$ consisting of single qubit terms only and an initial state $\ket{\psi_{0}}$ where the dynamics for $\Theta(\log (n))$ time can be exponentially hard to approximate with a quantum circuit. One explicit example is to take $\ket{\psi_{0}}$ to be a Haar random state $\ket{\psi_{\rm Haar}}$ and $H_{z} =-i\sum_i Z_i$. Our hardness of approximation result comes from two facts:
\begin{enumerate}
    \item Applying the dynamics generated by $H_{z}$ for evolution time\ $t = \Theta(\log(n))$ and measure in the computational basis allows one to distinguish a Haar random state from a maximally mixed state $\rho_m = \mathbbm{I}/d$, where $d = 2^n$.
    \item With overwhelmingly high probability, distinguishing a state sampled from Haar random ensemble from $\rho_m $ is exponentially hard.
\end{enumerate}
The task can be formalized into a decision problem: 
\begin{problem}[Distinguishes a Haar random state from a maximally mixed state]\label{prob: Haar}
An oracle $\mathcal{O}$ prepares copies of a fixed quantum state on a $n$-qubit quantum register that is promised to be either a maximally mixed state or a pure state sampled Haar randomly. Decide which is the case with as few queries as possible.
\end{problem}
We first prove that time evolution with a single qubit dissipation channel can efficiently solve Prob.~\ref{prob: Haar}.
\begin{lemma}[Single qubit dissipation distinguishes a Haar random state from a maximally mixed state] \label{lem:easy}
With $O(k)$ queries and probably $1-e^{-O(k)}$, a quantum state sampled Haar randomly can be distinguished from a maximally mixed state by time evolving the state for $t = \Theta(\log(n))$ with $H_z$. 

%\begin{align}\label{eq:trace_distance}
%    \Tr{\ketbra{\psi_{\rm Haar}(t)}-\rho_m(t))}| > 0.1. 
%\end{align}
%$\rho(t)$ and $\ket{\psi(t)}$ denotes normalized density matrices and states after time evolution.
\end{lemma}
\begin{proof}
    We consider the output distributions of a Haar random state before and after time evolution, measured on the computational basis $$p_s := \left|\braket{s}{\psi_{\rm Haar}}\right|^2,\ q_s:= \left|\braket{s}{\psi_{\rm Haar}(t)}\right|^2$$
    $p_s$ is known to fluctuate around its mean value i.e. the distribution of $p_s = 1/d$, but with a variance exponentially small in $n$. This can be seen from the numerical experiments on the left panel of Fig.~\ref{fig:output}. As a result, it is exponentially hard to distinguish a Haar random state and a maximally mixed state with a naive computational basis measurement.

    The effect of the time evolution generated by $H_z$ is that it exponentially re-distributes weights of all output strings according to their Hamming weight $w_s$, or the number of 0's in $s$. This can already be told from a single qubit case: $$e^{-Zt} (a\ket{0}+b\ket{1}) = ae^{-t}\ket{0} + be^t\ket{1}$$.
    
    We begin by finding $t$, such that after applying the evolution generated by $H_z$ to the maximally mixed state, the output weight on the $s=1^n$  can be a constant $c$. By setting:
    \begin{align}\label{eq:weight}
        \left(\frac{e^{2t}}{e^{-2t}+e^{2t}}\right)^{n} = c
    \end{align}
    
    and solving for $t$, we get 
    \begin{align}
        t &= -\frac{1}{4} \log(-c^{-1/ n} (-1 + c^{1/ n}))\\
          & = -\frac{1}{4} \log(c^{-1/ n}-1)\\
          & = -\frac{1}{4} \left[\log(-\frac{1} {n} \log (c))+O(n^{-1}))\right]\\
          & = \Theta(\log(n))
    \end{align}
    What would the output distribution look like for a Haar random state after applying the same time evolution? The normalized weight is just %For states sampled from Haar randomly, the output weight on any fixed string (in this case, $s = 1^n$) is known to follow the Porter-Thomas distribution~\cite{}: 
    %$${\rm PT}(p_s) = 2^n e^{2^n p_s}$$

    \begin{align}
        q_{1^n}  &=\frac{p_{1^n} e^{2n}}{\sum_s p_s e^{4w_s-2n}}\\
       &=\frac{p_{1^n} e^{2n}}{\sum_{i=0}^{n}\sum_{w_s = i}p_s e^{4i-2n}}\\
       &=\frac{p_{1^n} e^{2n}}{p_{1^n} e^{2n}+ p_{0^n} e^{-2n}+\sum_{i=1}^{n-1}2^{-n}\begin{pmatrix}
           n\\
           i
       \end{pmatrix}p_s e^{4i-2n}} \\
       &:=\frac{p_{1^n}}{p_{1^n}+2^{-n}\eta }
       \label{eq:weight_haar}
    \end{align}
    %$\Tr{\ketbra{\psi_{\rm Haar}(t)}-\rho_m(t))}| > 0.1$ 
    In the second from last step, we use the fact that each entry of a Haar state vector consists of two standard normal variables and the weights on each string should follow the Porter-Thomas distribution: ${\rm PT}(p_s) = 2^n e^{-2^n p_s}$, and a sum over $\poly(n)$ terms would converge to its mean value. We may also ignore the $p_{0^n} e^{-2n}$ in the denominator as it is exponentially small in $n$. 

    Under the assumption of Eq.~\ref{eq:weight}, we have $$ \frac{1}{1+\eta}= c$$

     W.l.o.g., setting $c = 1/2$ gives $\eta  = 1$; therefore the Eq.~\ref{eq:weight_haar} returns $$q_{1^n} = \frac{2^{n}p_{1^n}}{2^{n}p_{1^n}+1},$$
     
     This means the variance of this distribution is now a constant independent of $n$. For example, the probability of $P[q_{1^n}\geq0.6]$ is 
     \begin{align}
         P[q_{1^n}\geq0.6] &= P[q_{1^n}\geq0.6]\\
         & = P\left[\frac{2^{n}p_{1^n}}{2^{n}p_{1^n}+1}\geq0.6\right]\\
         & = P\left[p_{1^n}\geq 1.5\times2^{-n}\right]\\
         & = \int_{1.5\times2^{-n}}^\infty 2^n e^{-2^n p} dp\\
         & \approx 0.223
     \end{align}
     As evident from Fig.~\ref{fig:output}, with constant probability, a measurement in the computational basis now allows efficient distinguishing between the Haar random state and the maximally mixed state after the time evolution~\footnote{Notice that, if we choose $t = \poly(n)$, the two states are once again hard to be distinguished because they both purifies to a product state.}. This probability can be boosted to $1-e^{-k}$ by randomly applying a layer of $X$ gates before the time evolution (namely, randomly selecting a string $s$ to amplify and probe its amplitude) and repeating $O(k)$ times. 
\end{proof}
\begin{figure}
    \centering
        \centering
        \includegraphics[width=.75\linewidth]{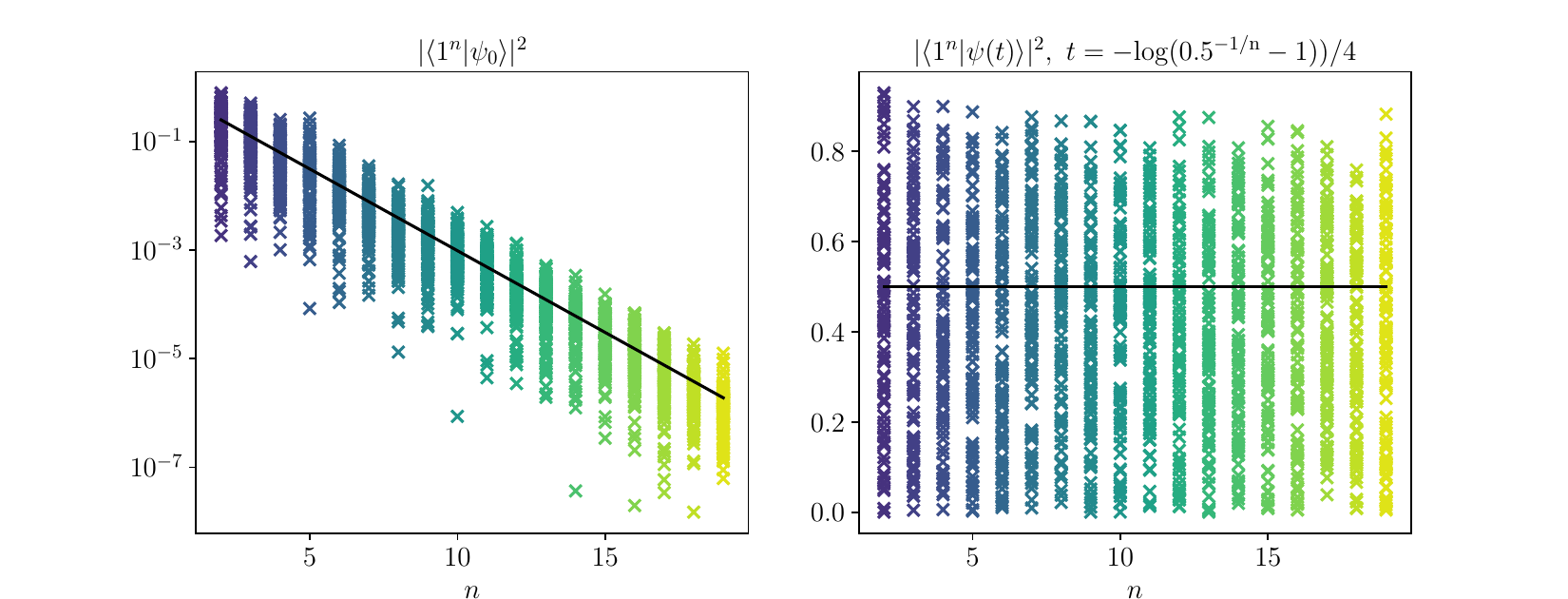}
    \caption{Output weight on the $1^n$ string of 100 Haar random states before (left) and after (right) applying time evolution generated by $H_z$ for a time length $t$; the output weight of a maximally-mixed state, $\rho_m$ is plotted with a black line for comparison. In both cases, the output weight of the Haar states fluctuates around the output value of the maximally mixed state. However, the variance on the left plot is exponentially small in $n$, and the variance after the non-Hermitian time evolution becomes a constant, allowing efficient distinguishment between the Haar random state and the maximally mixed state. This cannot be accomplished with any local quantum channels such as the amplitude damping channel~\cite{fefferman2023effect}.}
    \label{fig:output}
\end{figure}
In the unitary setting, the minimum number of gates required to implement a measurement $M$ that can \emph{distinguish} a given quantum state $\ket{\psi}$ from the maximally mixed state $\rho_m$ to a certain resolution $\delta$ is defined as the \emph{strong} state complexity. Formally, let $\beta(r, \ket{\psi})$ be the maximum bias with which $\ket{\psi}$ can be distinguished from the maximally mixed state using a circuit with at most $r$ gates from the gate set $\mathcal{G} \subseteq U(4)$:
\begin{align}
    \beta(r,\ket{\psi}) =& \text{ max}\ |\Tr{M(\ketbra{\psi}{\psi}-\rho_m)}|\\
    &\text {s.t. }M \text{ can be implemented with at most } r \text{ gates}
\end{align}

\noindent Then the strong state complexity is defined as:
\begin{definition}[Strong state complexity]\label{def:sc_state}
    For a given $r \in \mathbb{N}$ and $\eta \in (0, 1)$, a pure state $\ket{\psi}$ has strong $\eta$-state
complexity at most $r$ if and only if
$\beta(r,\ket{\psi}) \geq 1 - 1/d - \delta$.
We denote this $\mathcal{C}_\delta (\ket{\psi}) \leq r$.
\end{definition}
The $1/d$ in the definition comes from that any pure state can be trivially distinguished from a maximally mixed state with trace distance $1/d$. For Haar random pure states, it has been shown that the vast majority of them have exponentially high strong state complexity:
\begin{proposition}[Strong state complexity of Haar random states]
    The probability that $\ket{\psi_{\rm Haar}}$ has strong circuit complexity less than $r$ is 
    $$\Pr [\mathcal{C}_{\delta}(\ket{\psi_{\rm Haar}})\leq r ] \leq 4.0144 d(n + 1)^r|\mathcal{G}|^r \exp{-\frac{d(1-\delta^2)}{9\pi^3}}$$
\end{proposition}
The proof given by~\cite{brandao2021models} essentially comes from a counting argument that any fixed measurement can only distinguish a small number of states, and thus to distinguish the vast majority of Haar random states, the complexity of measurement must grow exponential in $n$. Even for $1-\delta^2 = 1/\poly(n)$, this probability of distinguish remains exponentially small in $n$ when $r = \poly (n)$. This means 
\begin{lemma}[Hardness of distinguishing a Haar random state with a $\mathsf{BQP}$ machine]\label{lem:hard}
    With overwhelming probability, a $\mathsf{BQP}$ machine making polynomial queries to $\mathcal{O}$ cannot solve Prob.~\ref{prob: Haar}.
\end{lemma}
Combining Lem.~\ref{lem:easy} and Lem.~\ref{lem:hard} completes the proof of Thm.~I.1.
Further, as we explain in Appx.~\ref{app:experimental}, each single-qubit dissipation channel can be implemented with a two-qubit gate and post-selection on the single ancillary qubit. It turns out that the result we have proved can be summarized in a computational complexity language: With overwhelming probability, Prob.~\ref{prob: Haar} can be solved by a quantum circuit at merely constant depth when postselection is allowed (denoted as $\mathsf{PostQNC_0}$), but also with overwhelming probability, it cannot be solved by a polynomial-size quantum circuit. Namely,
\begin{corollary}[Oracle separation]
    There exists a quantum oracle $\mathcal{O}$ (as defined in Prob.~\ref{prob: Haar}) such that $\mathsf{PostQNC_0}^\mathcal{O}\not\subseteq\mathsf{BQP}^\mathcal{O}$;  relative to the same oracle, $\mathsf{BQP}^\mathcal{O} \subset\mathsf{PostBQP}^\mathcal{O}$.
\end{corollary}
%\YZ{finish proof}
\section{Additional details on non-Hermitian VQE with variance minimization}\label{app:numerical}
\subsection{Failure of energy minimization due to non-Hermicity}
In standard VQE, energy minimization is performed with the loss function:
$$E = \braket{\psi({\bm\theta})|H}{\psi({\bm\theta})}$$

This cost function would fail when the system becomes non-Hermitian. Consider a two-level system with right eigenvectors $\{\ket{\psi^+},\ket{\psi^-}\}$ and eigenvalues $\{E^+,E^-\}$, an energy calculation on the trial wave function $\psi = a\ket{\psi^+}+b\ket{\psi^-}$ gives:
$$E = aa^*E^+ + bb^*E^+ + a^*bE^- \braket{\psi^+}{\psi^-}+ab^*E^+ \braket{\psi^-}{\psi^+}$$

In a non-Hermitian system, the right eigenvectors are not guaranteed to be orthonormal to each other. As such, the last two terms on the right-hand side would not be 0 and the result of energy minimization is not necessarily an eigenstate of $H$. We numerically observed this effect with the dissipative Ising chain Hamiltonian, where the energy optimizer can return much lower energy than the value given by exact diagonalization calculation.

\begin{figure*}
    \centering
        \centering
        \includegraphics[width=1\linewidth]{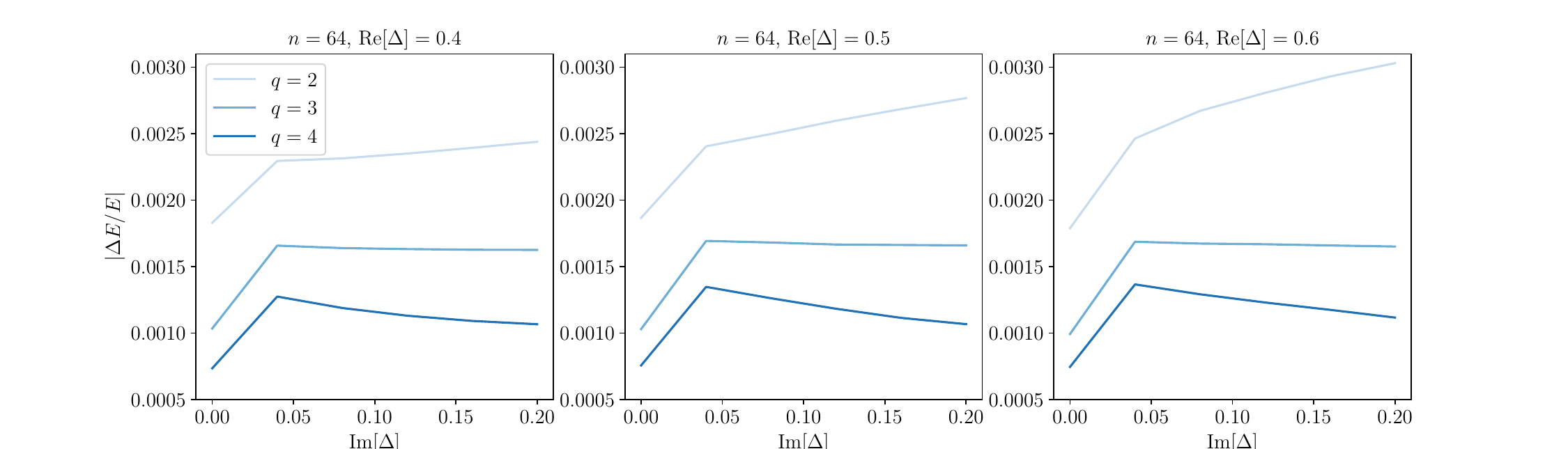}
        \caption{Energy VQE results of a $n=64$ dissipative XXZ chain using various numbers of bond qubits, compared with a $\chi = 100$ NH-DMRG with truncation error set to $10^{-10}$ as the truth value. The local circuit geometry of each $U_A$ in the sequential circuit is fixed to be a $\tau = 10$ brickwall. For each panel, we fix the real part of $\Delta$ and vary the imaginary field.}
        \label{fig:xxz}
\end{figure*}

%\subsection{Ising model}
%We provide some additional data here; First on Ising model: like many non-Hermitian models, this is a model that is very sensitive to boundary conditions; it is also very akin to the even or odd; introduce $Z_2$ symmetry
\subsection{Testing the algorithm with a dissipative XXZ model}
In the main text, we demonstrate that the energy and correlation function of a non-Hermitian Ising chain can be captured accurately with merely 2 bond qubits. To extend our study to a larger system size and bond space, we consider a more challenging task: ground state preparation of the non-Hermitian XXZ spin chain model whose Hamiltonian is  %\YZ{To do: introduce the non-Hermitian XXZ model, describe how the critical exponent can be fit through RR correlators}

\begin{align}\label{eq: xxz}
    H_{\rm XXZ} = \sum_i [X_{i}X_{i+1}+ Y_{i}Y_{i+1} + \Delta Z_{i}Z_{i+1}] 
\end{align}

where and $\Delta$ can be any complex number. Eq.~\ref{eq: xxz} was recently proposed as a prototypical dissipative Tomonaga-Luttinger (TL) liquid~\cite{yamamoto2022universal}. Next, we perform a VQE study with a dissipative XXZ model with Algo.~2. Varying the number of bond qubits in the sequential circuits while keeping the local circuit depth fixed, the results in Fig.~\ref{fig:xxz} demonstrate an improved accuracy in energy. Unlike the dissipative Ising model discussed in the main text, the energy of the XXZ chain becomes complex with even a slight imaginary perturbation in $\Delta$, and this is reflected in the increase in the energy infidelity, suggesting a rise in the hardness of representing the state. 

Comparing three panels where the real part of $\Delta$ varies, the energy infidelity reported by $q = 4$ consistently stays below $0.15\%$, whereas the $q = 2$ result reports increased error with increased imaginary field strength, seemingly suggests that the small number of bond qubits have a hard time capturing the ground state in the strong imaginary field regime.

\subsection{The importance of warm start}

Despite the success of variational solvers in tasks like energy finding and unitary compilation, they often encounter the barren plateau issue. Namely, the average gradient variance of the cost function becomes exponentially small with system size as the circuit depth increases. For local cost functions, such as a single Pauli observable, this issue can arise with a circuit depth of $\poly(n)$. For global cost functions, the challenge can be even more pronounced~\cite{cerezo2021cost}.

Nevertheless, various `warm start' strategies have been proposed and shown to be effective in optimizing PQCs. 
A barren plateau happens when one randomly initializes a deep PQC, which readily forms unitary designs~\cite{mcclean2018barren}. 
On the other hand, a warm start typically utilizes an iterative procedure where the PQC is first set to be a shallow circuit, optimized, and repeatedly fed into a deeper PQC.
In a recent study~\cite{drudis2024variational}, the authors analytically investigate an iterative initialization strategy. They demonstrate that such a variational algorithm maintains at least an inverse polynomial gradient in $n$ at each step even for global cost functions. By establishing convexity guarantees for these regions, their work suggests that polynomial-size circuits could be trainable under certain situations, making the variational approach viable for practical quantum computations.
\begin{figure}
    \centering
        \centering
        \includegraphics[width=.75\linewidth]{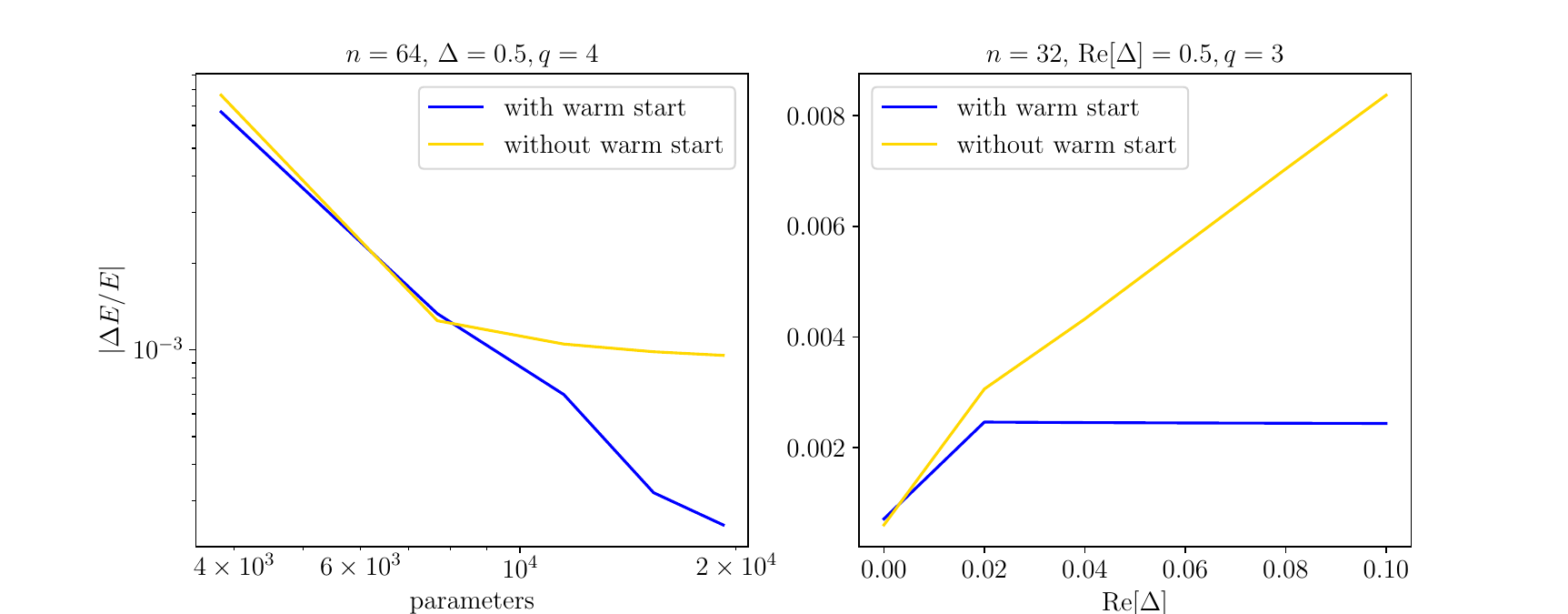}
    \caption{(a) A proper initialization method is crucial to mitigating barren plateaus in a standard Hermitian VQE algorithm. Here we compare the energy optimization result with random initialization and a sequential initialization where one gradually increases the local circuit depth, $\tau$ in the decomposition of $U_A$, and feeds the previously optimized parameters as an initialization. When the initialization succeeds, the relative energy error decays polynomically with the number of variational parameters in the circuit. (b) When the imaginary field is turned on, a good initialization guarantees the variance minimization algorithm finds the ground state, rather than an arbitrary eigenstate.}
    \label{fig:initialization}
\end{figure}
Warm start strategies are both employed in our numerical optimizations of VQC and VQE. Take VQE as an example: each sequential circuit is initially set to local depth $\tau = 2$, which is then optimized by an optimizer, Adam, that has been widely used in tensor network and neural network studies. Next, we add another gate layer of identity while maintaining all gate parameters from the previous optimization. The learning rate of the optimizer is also fine-tuned such that it decreases as the optimizer convergence to the optima at each circuit depth. We repeat this procedure until the desired circuit depth is reached. 

Fig.~\ref{fig:initialization}a gives a comparison of the optimization results using random initialization and warm start: when the initial circuit parameters are completely chosen at random, the optimizer gets stuck after the number of parameters goes beyond $\sim 8,000$; on the other hand, the energy infidelity continues to decay with polynomially with the number of free parameters in the sequential circuit, as report in the previous work~\cite{haghshenas2021variational}.

In the non-Hermitian VQE, Algo.~2, a warm start is used in the variance minimization scheme to ensure the optimizer returns the desired ground state rather than an arbitrary state. Starting with a regular VQE and gradually increasing the imaginary field, both energy and circuit parameters from a previous iteration are fed into the next step. Fig.~\ref{fig:initialization}b shows that without proper use of a warm start, the energy infidelity increases significantly with the imaginary field whereas it stays relatively constant when proper initialization is introduced.

%In this section, we emphasize the two initialization strategies to 1) avoid barren plateaus and 2)make sure we have the right eigenstate. This is a similar optimization method to the one proposed in~\cite{haghshenas2021variational,zhang2022qubit}.
\section{Experimental details}\label{app:experimental}
\subsection{Quantinuum H1 Trapped ion computer}
\label{appendix:circuit}
All experimental data in this study were collected using Quantinuum's System Model H1 trapped-ion quantum processor~\cite{pino2021demonstration}. This processor employs a quantum charge-coupled device (QCCD) architecture with 5 parallel gate zones and 20 qubit ions, enabling all-to-all connectivity and mid-circuit measurement.

The processor can implement arbitrary one-qubit gates, while multi-qubit operations are achieved by compiling into its native entangling two-qubit gate: a M\o lmer-S\o rensen gate wrapped with single-qubit dressing pulses, resulting in a phase-insensitive operation $u_{MS} = \exp[-i\frac{\pi}{4}\sigma^{z}\otimes \sigma^{z}]$ \cite{pino2021demonstration}. The one-qubit gates and two-qubit gates have typical average infidelities of approximately $\epsilon_{1q}\approx10^{-5}$ and $\epsilon_{2q}\approx 10^{-3}$, respectively. Additionally, it recently achieved a $2^{20}$ quantum volume, setting a world record among quantum processors.

However, the high performance of the processor is accompanied by a relatively low clock rate. To mitigate the large sampling overhead of variational algorithms, we chose to perform classical optimization of circuit parameters and implemented only the optimized circuit in hardware. In principle, this optimization loop could potentially be implemented on a quantum processor in the near future. We also computationally simulated the time-evolved state classically, although this could be achieved on a quantum computer using algorithms such as quantum imaginary time evolution~\cite{motta2020determining}. The total experimental cost amounted to approximately 12,000 credits, with half allocated to dynamical simulation and the remainder to ground state preparation.

\paragraph{The fSim gate} In the VQC circuit, we employ fSim gates, the most general particle number-conserving two-qubit gates to enable noise mitigation based on discarding the data with no correct total particle number. This is accomplished with the following unitary:
\begin{equation}\label{eq: fsim}
\begin{pmatrix}
e^{i(\gamma+\phi)} & 0 & 0 & 0 \\
0 & e^{i(-\gamma+\phi+\zeta)}\sin\theta & e^{-i(\chi+\gamma+\phi)}\cos\theta & 0\\
0 & e^{i(\chi-\gamma+\phi)}\cos\theta & e^{-i(\gamma+\phi+\zeta)}\sin\theta & 0\\
0 & 0 & 0 & e^{i(\gamma - \phi)}\\
\end{pmatrix}
\end{equation}
where $(\gamma, \phi, \zeta, \chi, \theta)$ are variational parameters we independently choose for each two-qubit gate. Compared to a SU(4) gate, this gate reduces the number of parameters per gate from 15 to 5 and, crucially to experimental implementation, Eq.~\ref{eq: fsim} can be synthesized with only two CNOT gates instead of three, as demonstrated in Fig.~\ref{fig:synthesis}. Moreover, as we see below, the use of fSim gates naturally allows error mitigation (EM) without any extra cost.
\subsection{Non-Hermitian fermionic quench}
\begin{figure}
    \centering
        \centering
        \includegraphics[width=.9\linewidth]{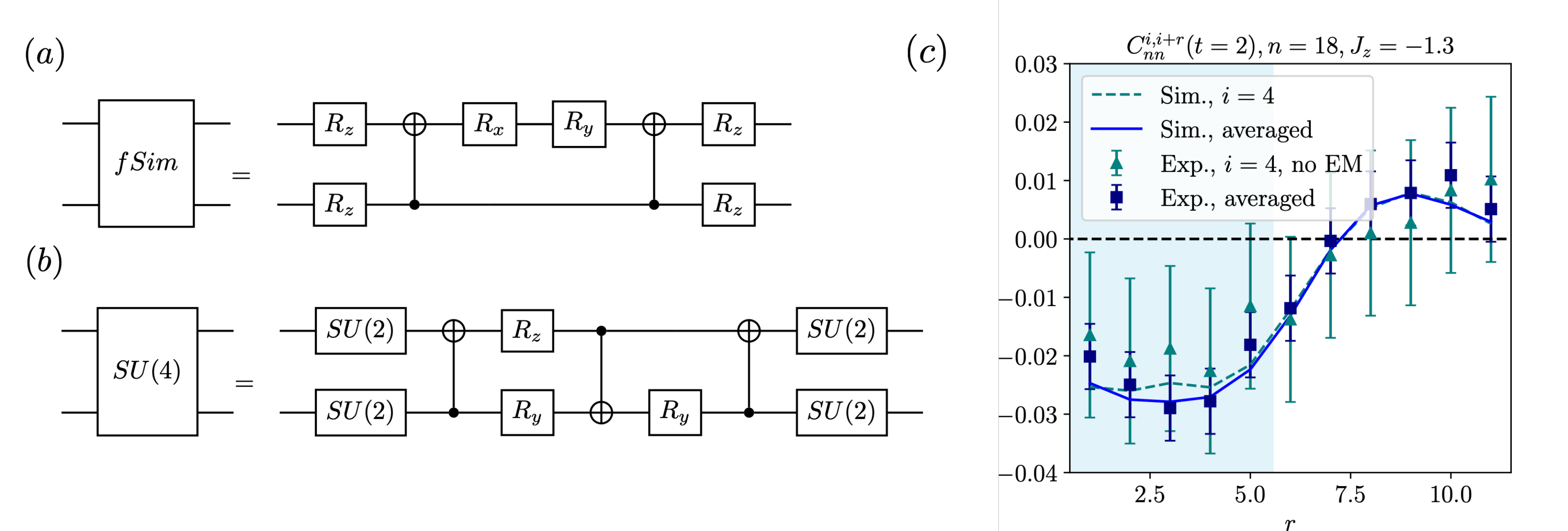}
    \caption{Synthesis of the (a) fSim gate and (b) SU(4) gate into CNOT and single-qubit gates. Each CNOT gate can be further decomposed into single-qubit rotation gates and exactly one $R_{ZZ}$ gate native to the Quantinuum's QCCD processor. (c) Comparing the density-density correlation function between averaged over the bulk sites and simply setting the initial site at $i=4$.}
    \label{fig:synthesis}
\end{figure}
\paragraph{Data processing}
%\subsection{Jordan Wigner transformation }
%The fermi; we rewrite it into a spin model with Jordan-Wigner transformation for the convenience of experimental implementation need to explain a number conserving quantum circuit
 Two methods were used to enhance the signal level in postprocessing experimental data. First, to reduce noise from the physical device, we post-select the measurement outcome in the $z$ direction on being in the right charge sector: we discard a measurement output if its Hamming weight does not equal $n/2$. This allows corrections of bit-flip errors.
 
 This is made possible because a) $C_{nn}$ translates into Pauli-Z correlators in the Jordan Wigner transformation is the density and b) the use of $U(1)$ charge-conserving gates, the fSim gates, in both the initial state preparation and variational time evolution.
 If non-charge preserving gates are used instead, because the compilation result has an error from the true circuit, we wouldn't know if the non-preservation of the total charge from the compilation error or a physical error. 
 
 Second, the correlation function is averaged over sites $i \in [3,4,5,6,13,14,15,16]$ for improved statistical significance. In Fig.~\ref{fig:synthesis}(c) we compare naively reporting the correlation function at $i=3$ to the postprocessed data that we explain above and show in the main text.

\paragraph{Resource comparison between VQC and Trotterization}
 In this section, we perform a resource comparison between the VQC method used in this work and standard Trotterization. As an example, we focus on the experimental implementation of Fig.~1(c) at $t = 2$. In quantum circuit compilation, one approach to implement $e^{-itH}$, where $H$ is a $2^n \times 2^n$ matrix, is to split up $t$ into $m$ smaller time steps $\Delta t$ and approximate the unitary propagator with a product of unitaries,
\begin{align}
    e^{-itH} \approx \exp{-i\Delta t H }^{m} \label{eq:trotter}.
\end{align}
 with $m = \frac{T}{\Delta t}$. To implement the fermionic time evolution of Eq.~5 on a qubit system, one first performs a Jordan-Wigner transformation:
\begin{align}
    H_{\rm fermi} &= \sum_{i}[\frac{1+iJ_z}{2}(c^\dag_i c_{i+1}+h.c.) - i\frac{\pi J_z}{2} n_in_{i+1}]\\
     &= \sum_{i}[\frac{1+iJ_z}{2}(S^+_{i}S^-_{i+1}+h.c.) - i\frac{\pi J_z}{2} (S^z_i+\mathbbm{I}/2)(S^z_{i+1}+\mathbbm{I}/2)]\\
\end{align}
The transformed Hamiltonian is still a sum of local terms, i.e., $H = \sum_j h_{j,j+1}$, and one can perform the following Trotter decomposition suitable for this case:

$$e^{-i \Delta t H} \approx e^{-i h_{1,2} \Delta t/2} e^{-i h_{2,3} \Delta t/2} \cdots e^{-i h_{n-1,n} \Delta t/2}
e^{-i h_{n-1,n} \Delta t/2} e^{-i h_{N-2,N-1} \Delta t/2} \cdots e^{-i h_{1,2} \Delta t/2} + O(\Delta t^3)$$
%Depending on the number of terms $m$ in the product, one can achieve different orders of approximation $p$. In this analysis, we focus on a $p = 2$ implementation that is suitable 
For the total evolution time $t = 2$ considered in the work, %as higher $p$ generally requires a much larger circuit depth to implement. 
we numerically found that a state infidelity of $1-\mathcal{F} = 0.023$ can be achieved at $m=3$, which matches the target infidelity considered by VQC.
\begin{figure}\label{fig:trotter}
    \centering
        \centering
        \includegraphics[width=.8\linewidth]{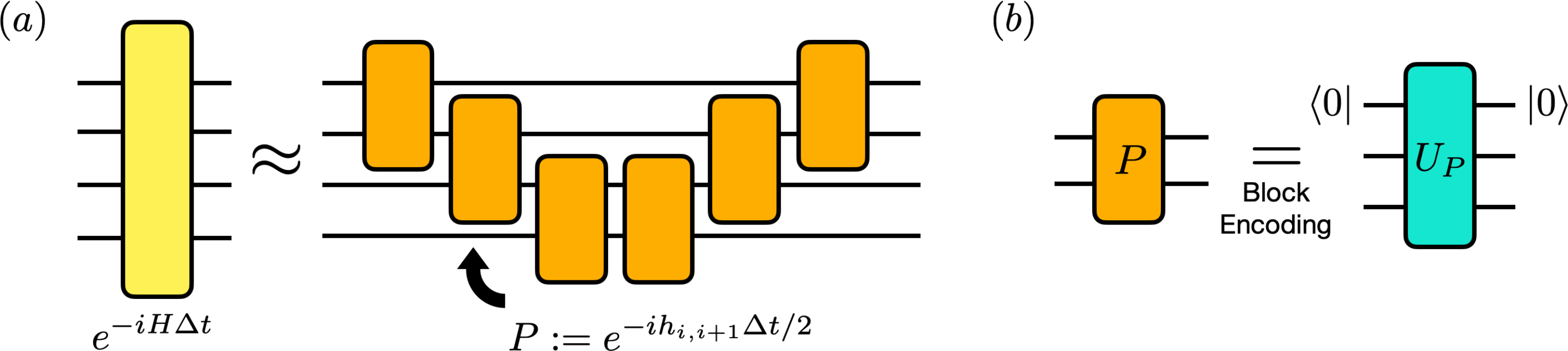}
    \caption{(a) A Trotterization ansatz considered in Appx.~\ref{app:experimental} $e^{-iH\Delta t}$ is approximated by a series of ladder-shaped nearest neighbor time evolutions $P$. (b) To implement a non-unitary operation, the typical way is to implement $P$ in a unitary, $U_P$, in doubled Hilbert space and postselected on the ancilla qubit getting 0 outcome. $U_P$ can further be synthesized into local unitary circuits with, e.g., SU(4) gates.}
\end{figure}
Next, we calculate the number of gates needed to implement the non-Hermitian time evolution of each Trotterized block on two nearest neighbor qubits: $P := e^{-i h_{i,i+1}\Delta t/2}$. Specifically, this can be done with a block encoding~\cite{low2017optimal}; namely $P$ can be embedded in a unitary with double the Hilbert space size plus post-selection: 
\begin{align}
U_P = 
\begin{pmatrix}
    \lambda^{-1/2} P& Q\\
    R &S 
\end{pmatrix}
\end{align}
where $\lambda$ is the largest eigenvalue of $PP^\dag$, which is guaranteed to be positive from the existence of a singular value decomposition $P = U\Sigma V^\dag$. To ensure mutual orthogonality on the first half columns, it suffices to choose $R = U\sqrt{\mathbb{I} - \Sigma^2} V^\dag$. The remaining columns can then be filled in with a QR decomposition. 

We numerically find that $U_P$ can be approximated to a trace distance error $\sim 10^{-4}$ with 4 $\rm{SU}(4)$ gates. 
Typically, the success probability of post-selection varies over initial states, and it averaged to $\sim 0.65$ in our 10,000 numerical tests with Haar random initial states. Putting everything together, we estimate the cost of implementing the dynamics with a Trotterization is about $2(n-1)\times m \times 4$ number of SU(4) gates, or 1224 CNOT gates (see Fig.~\ref{fig:synthesis} for decomposition of a SU(4) gate into CNOT gates), with a post-selection probability $0.65^{103}\approx5.4\times10^{-20}$. 

 Admittedly, due to the freedom in choosing $U_P$, it might be possible to trade off between gate cost and postselection rate; this however will not change the success rate's exponential small scaling in the circuit volume. On the other hand, VQC reaches target infidelity $\epsilon = 0.02$ with merely 4 layers of VQC unitary circuits, and the total gate cost for time evolution is $4\times (n-1) = 68$ fSim gates or 136 CNOT gates. 
\subsection{Qubit-efficient preparation of dissipative Ising ground state}
In this section, we detail the resources used in experimental preparation of the dissipative Ising ground state.

\begin{enumerate}
    \item \emph{Qubit count.} In the implementation of a sequential circuit, qubit can be reused. The total qubit count required is only $q+1$, where $q$ is the number of qubits used as the MPS bond space. In this simulation, we use merely two bond qubits, which means as few as three qubits are sufficient for this implementation.
    \item \emph{Gate count.} The ladder circuit (as shown in Fig.~2) at each physical site has depth $\tau$, and the gate count per site $\sim \tau q$. The gate count has a linear dependence on the system size, and the total gate count $\sim \tau q n$. To be precise, for us, $q = 2, \tau = 2$, and $n = 20$, thus the total SU(4) gate count is 80.
    \item \emph{Gate set.} In our tensor-network algorithm, each tenor is assumed to be an SU(4) gate, which can be synthesized with three hardware-natural two-qubit gates and other single-qubit gates~\cite{vatan2004optimal}. 
    \item \emph{Total shots.} We use 2,000 shots for estimating the correlators to reduce shot noise. The energy can be measured from the $X$-correlators and $Z$-correlators.
\end{enumerate}
%Putting everything together, we would like perform a parameter scan of, e.g. 4 different imaginary field strengths; at each parameter, $\sim$10,000 shots of circuit execution is suggested. Each of the 40,000 total shots makes use of merely $3$ qubits and 60 SU(4) gates, from which we could calculate the energy and physical correlators for the ground state of a 20-site non-Hermitian system within a few percent of the true value.
%\end{appendix}
\onecolumngrid
\bibliography{main}

\begin{thebibliography}{81}%
\makeatletter
\providecommand \@ifxundefined [1]{%
 \@ifx{#1\undefined}
}%
\providecommand \@ifnum [1]{%
 \ifnum #1\expandafter \@firstoftwo
 \else \expandafter \@secondoftwo
 \fi
}%
\providecommand \@ifx [1]{%
 \ifx #1\expandafter \@firstoftwo
 \else \expandafter \@secondoftwo
 \fi
}%
\providecommand \natexlab [1]{#1}%
\providecommand \enquote  [1]{``#1''}%
\providecommand \bibnamefont  [1]{#1}%
\providecommand \bibfnamefont [1]{#1}%
\providecommand \citenamefont [1]{#1}%
\providecommand \href@noop [0]{\@secondoftwo}%
\providecommand \href [0]{\begingroup \@sanitize@url \@href}%
\providecommand \@href[1]{\@@startlink{#1}\@@href}%
\providecommand \@@href[1]{\endgroup#1\@@endlink}%
\providecommand \@sanitize@url [0]{\catcode `\\12\catcode `\$12\catcode
  `\&12\catcode `\#12\catcode `\^12\catcode `\_12\catcode `\%12\relax}%
\providecommand \@@startlink[1]{}%
\providecommand \@@endlink[0]{}%
\providecommand \url  [0]{\begingroup\@sanitize@url \@url }%
\providecommand \@url [1]{\endgroup\@href {#1}{\urlprefix }}%
\providecommand \urlprefix  [0]{URL }%
\providecommand \Eprint [0]{\href }%
\providecommand \doibase [0]{https://doi.org/}%
\providecommand \selectlanguage [0]{\@gobble}%
\providecommand \bibinfo  [0]{\@secondoftwo}%
\providecommand \bibfield  [0]{\@secondoftwo}%
\providecommand \translation [1]{[#1]}%
\providecommand \BibitemOpen [0]{}%
\providecommand \bibitemStop [0]{}%
\providecommand \bibitemNoStop [0]{.\EOS\space}%
\providecommand \EOS [0]{\spacefactor3000\relax}%
\providecommand \BibitemShut  [1]{\csname bibitem#1\endcsname}%
\let\auto@bib@innerbib\@empty
%</preamble>
\bibitem [{\citenamefont {Dalibard}\ \emph {et~al.}(1992)\citenamefont
  {Dalibard}, \citenamefont {Castin},\ and\ \citenamefont
  {M{\o}lmer}}]{dalibard1992wave}%
  \BibitemOpen
  \bibfield  {author} {\bibinfo {author} {\bibfnamefont {J.}~\bibnamefont
  {Dalibard}}, \bibinfo {author} {\bibfnamefont {Y.}~\bibnamefont {Castin}},\
  and\ \bibinfo {author} {\bibfnamefont {K.}~\bibnamefont {M{\o}lmer}},\
  }\bibfield  {title} {\bibinfo {title} {Wave-function approach to dissipative
  processes in quantum optics},\ }\href@noop {} {\bibfield  {journal} {\bibinfo
   {journal} {Physical review letters}\ }\textbf {\bibinfo {volume} {68}},\
  \bibinfo {pages} {580} (\bibinfo {year} {1992})}\BibitemShut {NoStop}%
\bibitem [{\citenamefont {Yang}\ and\ \citenamefont
  {Lee}(1952)}]{yang1952statistical}%
  \BibitemOpen
  \bibfield  {author} {\bibinfo {author} {\bibfnamefont {C.-N.}\ \bibnamefont
  {Yang}}\ and\ \bibinfo {author} {\bibfnamefont {T.-D.}\ \bibnamefont {Lee}},\
  }\bibfield  {title} {\bibinfo {title} {Statistical theory of equations of
  state and phase transitions. i. theory of condensation},\ }\href@noop {}
  {\bibfield  {journal} {\bibinfo  {journal} {Physical Review}\ }\textbf
  {\bibinfo {volume} {87}},\ \bibinfo {pages} {404} (\bibinfo {year}
  {1952})}\BibitemShut {NoStop}%
\bibitem [{\citenamefont {Fisher}(1978)}]{fisher1978yang}%
  \BibitemOpen
  \bibfield  {author} {\bibinfo {author} {\bibfnamefont {M.~E.}\ \bibnamefont
  {Fisher}},\ }\bibfield  {title} {\bibinfo {title} {Yang-lee edge singularity
  and $\phi$ 3 field theory},\ }\href@noop {} {\bibfield  {journal} {\bibinfo
  {journal} {Physical Review Letters}\ }\textbf {\bibinfo {volume} {40}},\
  \bibinfo {pages} {1610} (\bibinfo {year} {1978})}\BibitemShut {NoStop}%
\bibitem [{\citenamefont {Hatano}\ and\ \citenamefont
  {Nelson}(1996)}]{hatano1996localization}%
  \BibitemOpen
  \bibfield  {author} {\bibinfo {author} {\bibfnamefont {N.}~\bibnamefont
  {Hatano}}\ and\ \bibinfo {author} {\bibfnamefont {D.~R.}\ \bibnamefont
  {Nelson}},\ }\bibfield  {title} {\bibinfo {title} {Localization transitions
  in non-hermitian quantum mechanics},\ }\href@noop {} {\bibfield  {journal}
  {\bibinfo  {journal} {Physical review letters}\ }\textbf {\bibinfo {volume}
  {77}},\ \bibinfo {pages} {570} (\bibinfo {year} {1996})}\BibitemShut
  {NoStop}%
\bibitem [{\citenamefont {Hatano}\ and\ \citenamefont
  {Nelson}(1997)}]{hatano1997vortex}%
  \BibitemOpen
  \bibfield  {author} {\bibinfo {author} {\bibfnamefont {N.}~\bibnamefont
  {Hatano}}\ and\ \bibinfo {author} {\bibfnamefont {D.~R.}\ \bibnamefont
  {Nelson}},\ }\bibfield  {title} {\bibinfo {title} {Vortex pinning and
  non-hermitian quantum mechanics},\ }\href@noop {} {\bibfield  {journal}
  {\bibinfo  {journal} {Physical Review B}\ }\textbf {\bibinfo {volume} {56}},\
  \bibinfo {pages} {8651} (\bibinfo {year} {1997})}\BibitemShut {NoStop}%
\bibitem [{\citenamefont {Ashida}\ \emph {et~al.}(2020)\citenamefont {Ashida},
  \citenamefont {Gong},\ and\ \citenamefont {Ueda}}]{ashida2020non}%
  \BibitemOpen
  \bibfield  {author} {\bibinfo {author} {\bibfnamefont {Y.}~\bibnamefont
  {Ashida}}, \bibinfo {author} {\bibfnamefont {Z.}~\bibnamefont {Gong}},\ and\
  \bibinfo {author} {\bibfnamefont {M.}~\bibnamefont {Ueda}},\ }\bibfield
  {title} {\bibinfo {title} {Non-hermitian physics},\ }\href@noop {} {\bibfield
   {journal} {\bibinfo  {journal} {Advances in Physics}\ }\textbf {\bibinfo
  {volume} {69}},\ \bibinfo {pages} {249} (\bibinfo {year} {2020})}\BibitemShut
  {NoStop}%
\bibitem [{\citenamefont {Bender}\ and\ \citenamefont
  {Wu}(1969)}]{bender1969anharmonic}%
  \BibitemOpen
  \bibfield  {author} {\bibinfo {author} {\bibfnamefont {C.~M.}\ \bibnamefont
  {Bender}}\ and\ \bibinfo {author} {\bibfnamefont {T.~T.}\ \bibnamefont
  {Wu}},\ }\bibfield  {title} {\bibinfo {title} {Anharmonic oscillator},\
  }\href@noop {} {\bibfield  {journal} {\bibinfo  {journal} {Physical Review}\
  }\textbf {\bibinfo {volume} {184}},\ \bibinfo {pages} {1231} (\bibinfo {year}
  {1969})}\BibitemShut {NoStop}%
\bibitem [{\citenamefont {Yao}\ and\ \citenamefont {Wang}(2018)}]{yao2018edge}%
  \BibitemOpen
  \bibfield  {author} {\bibinfo {author} {\bibfnamefont {S.}~\bibnamefont
  {Yao}}\ and\ \bibinfo {author} {\bibfnamefont {Z.}~\bibnamefont {Wang}},\
  }\bibfield  {title} {\bibinfo {title} {Edge states and topological invariants
  of non-hermitian systems},\ }\href@noop {} {\bibfield  {journal} {\bibinfo
  {journal} {Phys. Rev. Lett.}\ }\textbf {\bibinfo {volume} {121}},\ \bibinfo
  {pages} {086803} (\bibinfo {year} {2018})}\BibitemShut {NoStop}%
\bibitem [{\citenamefont {Ashida}\ and\ \citenamefont
  {Ueda}(2018)}]{ashida2018full}%
  \BibitemOpen
  \bibfield  {author} {\bibinfo {author} {\bibfnamefont {Y.}~\bibnamefont
  {Ashida}}\ and\ \bibinfo {author} {\bibfnamefont {M.}~\bibnamefont {Ueda}},\
  }\bibfield  {title} {\bibinfo {title} {Full-counting many-particle dynamics:
  Nonlocal and chiral propagation of correlations},\ }\href@noop {} {\bibfield
  {journal} {\bibinfo  {journal} {Physical review letters}\ }\textbf {\bibinfo
  {volume} {120}},\ \bibinfo {pages} {185301} (\bibinfo {year}
  {2018})}\BibitemShut {NoStop}%
\bibitem [{\citenamefont {D{\'o}ra}\ and\ \citenamefont
  {Moca}(2020)}]{dora2020quantum}%
  \BibitemOpen
  \bibfield  {author} {\bibinfo {author} {\bibfnamefont {B.}~\bibnamefont
  {D{\'o}ra}}\ and\ \bibinfo {author} {\bibfnamefont {C.~P.}\ \bibnamefont
  {Moca}},\ }\bibfield  {title} {\bibinfo {title} {Quantum quench in p
  t-symmetric luttinger liquid},\ }\href@noop {} {\bibfield  {journal}
  {\bibinfo  {journal} {Physical Review Letters}\ }\textbf {\bibinfo {volume}
  {124}},\ \bibinfo {pages} {136802} (\bibinfo {year} {2020})}\BibitemShut
  {NoStop}%
\bibitem [{\citenamefont {Rudner}\ and\ \citenamefont
  {Levitov}(2009)}]{rudner2009topological}%
  \BibitemOpen
  \bibfield  {author} {\bibinfo {author} {\bibfnamefont {M.~S.}\ \bibnamefont
  {Rudner}}\ and\ \bibinfo {author} {\bibfnamefont {L.~S.}\ \bibnamefont
  {Levitov}},\ }\bibfield  {title} {\bibinfo {title} {Topological transition in
  a non-hermitian quantum walk},\ }\href@noop {} {\bibfield  {journal}
  {\bibinfo  {journal} {Phys. Rev. Lett.}\ }\textbf {\bibinfo {volume} {102}},\
  \bibinfo {pages} {065703} (\bibinfo {year} {2009})}\BibitemShut {NoStop}%
\bibitem [{\citenamefont {Gopalakrishnan}\ and\ \citenamefont
  {Gullans}(2021)}]{gopalakrishnan2021entanglement}%
  \BibitemOpen
  \bibfield  {author} {\bibinfo {author} {\bibfnamefont {S.}~\bibnamefont
  {Gopalakrishnan}}\ and\ \bibinfo {author} {\bibfnamefont {M.~J.}\
  \bibnamefont {Gullans}},\ }\bibfield  {title} {\bibinfo {title} {Entanglement
  and purification transitions in non-hermitian quantum mechanics},\
  }\href@noop {} {\bibfield  {journal} {\bibinfo  {journal} {Physical review
  letters}\ }\textbf {\bibinfo {volume} {126}},\ \bibinfo {pages} {170503}
  (\bibinfo {year} {2021})}\BibitemShut {NoStop}%
\bibitem [{\citenamefont {Kawabata}\ \emph {et~al.}(2023)\citenamefont
  {Kawabata}, \citenamefont {Numasawa},\ and\ \citenamefont
  {Ryu}}]{kawabata2023entanglement}%
  \BibitemOpen
  \bibfield  {author} {\bibinfo {author} {\bibfnamefont {K.}~\bibnamefont
  {Kawabata}}, \bibinfo {author} {\bibfnamefont {T.}~\bibnamefont {Numasawa}},\
  and\ \bibinfo {author} {\bibfnamefont {S.}~\bibnamefont {Ryu}},\ }\bibfield
  {title} {\bibinfo {title} {Entanglement phase transition induced by the
  non-hermitian skin effect},\ }\href@noop {} {\bibfield  {journal} {\bibinfo
  {journal} {Physical Review X}\ }\textbf {\bibinfo {volume} {13}},\ \bibinfo
  {pages} {021007} (\bibinfo {year} {2023})}\BibitemShut {NoStop}%
\bibitem [{\citenamefont {Guo}\ \emph {et~al.}(2009)\citenamefont {Guo},
  \citenamefont {Salamo}, \citenamefont {Duchesne}, \citenamefont {Morandotti},
  \citenamefont {Volatier-Ravat}, \citenamefont {Aimez}, \citenamefont
  {Siviloglou},\ and\ \citenamefont {Christodoulides}}]{guo2009observation}%
  \BibitemOpen
  \bibfield  {author} {\bibinfo {author} {\bibfnamefont {A.}~\bibnamefont
  {Guo}}, \bibinfo {author} {\bibfnamefont {G.~J.}\ \bibnamefont {Salamo}},
  \bibinfo {author} {\bibfnamefont {D.}~\bibnamefont {Duchesne}}, \bibinfo
  {author} {\bibfnamefont {R.}~\bibnamefont {Morandotti}}, \bibinfo {author}
  {\bibfnamefont {M.}~\bibnamefont {Volatier-Ravat}}, \bibinfo {author}
  {\bibfnamefont {V.}~\bibnamefont {Aimez}}, \bibinfo {author} {\bibfnamefont
  {G.~A.}\ \bibnamefont {Siviloglou}},\ and\ \bibinfo {author} {\bibfnamefont
  {D.~N.}\ \bibnamefont {Christodoulides}},\ }\bibfield  {title} {\bibinfo
  {title} {Observation of p t-symmetry breaking in complex optical
  potentials},\ }\href@noop {} {\bibfield  {journal} {\bibinfo  {journal}
  {Physical review letters}\ }\textbf {\bibinfo {volume} {103}},\ \bibinfo
  {pages} {093902} (\bibinfo {year} {2009})}\BibitemShut {NoStop}%
\bibitem [{\citenamefont {R{\"u}ter}\ \emph {et~al.}(2010)\citenamefont
  {R{\"u}ter}, \citenamefont {Makris}, \citenamefont {El-Ganainy},
  \citenamefont {Christodoulides}, \citenamefont {Segev},\ and\ \citenamefont
  {Kip}}]{ruter2010observation}%
  \BibitemOpen
  \bibfield  {author} {\bibinfo {author} {\bibfnamefont {C.~E.}\ \bibnamefont
  {R{\"u}ter}}, \bibinfo {author} {\bibfnamefont {K.~G.}\ \bibnamefont
  {Makris}}, \bibinfo {author} {\bibfnamefont {R.}~\bibnamefont {El-Ganainy}},
  \bibinfo {author} {\bibfnamefont {D.~N.}\ \bibnamefont {Christodoulides}},
  \bibinfo {author} {\bibfnamefont {M.}~\bibnamefont {Segev}},\ and\ \bibinfo
  {author} {\bibfnamefont {D.}~\bibnamefont {Kip}},\ }\bibfield  {title}
  {\bibinfo {title} {Observation of parity--time symmetry in optics},\
  }\href@noop {} {\bibfield  {journal} {\bibinfo  {journal} {Nature physics}\
  }\textbf {\bibinfo {volume} {6}},\ \bibinfo {pages} {192} (\bibinfo {year}
  {2010})}\BibitemShut {NoStop}%
\bibitem [{\citenamefont {Feng}\ \emph {et~al.}(2011)\citenamefont {Feng},
  \citenamefont {Ayache}, \citenamefont {Huang}, \citenamefont {Xu},
  \citenamefont {Lu}, \citenamefont {Chen}, \citenamefont {Fainman},\ and\
  \citenamefont {Scherer}}]{feng2011nonreciprocal}%
  \BibitemOpen
  \bibfield  {author} {\bibinfo {author} {\bibfnamefont {L.}~\bibnamefont
  {Feng}}, \bibinfo {author} {\bibfnamefont {M.}~\bibnamefont {Ayache}},
  \bibinfo {author} {\bibfnamefont {J.}~\bibnamefont {Huang}}, \bibinfo
  {author} {\bibfnamefont {Y.-L.}\ \bibnamefont {Xu}}, \bibinfo {author}
  {\bibfnamefont {M.-H.}\ \bibnamefont {Lu}}, \bibinfo {author} {\bibfnamefont
  {Y.-F.}\ \bibnamefont {Chen}}, \bibinfo {author} {\bibfnamefont
  {Y.}~\bibnamefont {Fainman}},\ and\ \bibinfo {author} {\bibfnamefont
  {A.}~\bibnamefont {Scherer}},\ }\bibfield  {title} {\bibinfo {title}
  {Nonreciprocal light propagation in a silicon photonic circuit},\ }\href@noop
  {} {\bibfield  {journal} {\bibinfo  {journal} {Science}\ }\textbf {\bibinfo
  {volume} {333}},\ \bibinfo {pages} {729} (\bibinfo {year}
  {2011})}\BibitemShut {NoStop}%
\bibitem [{\citenamefont {Shor}(1994)}]{shor1994algorithms}%
  \BibitemOpen
  \bibfield  {author} {\bibinfo {author} {\bibfnamefont {P.~W.}\ \bibnamefont
  {Shor}},\ }\bibfield  {title} {\bibinfo {title} {Algorithms for quantum
  computation: discrete logarithms and factoring},\ }in\ \href@noop {} {\emph
  {\bibinfo {booktitle} {Proceedings 35th annual symposium on foundations of
  computer science}}}\ (\bibinfo {organization} {Ieee},\ \bibinfo {year}
  {1994})\ pp.\ \bibinfo {pages} {124--134}\BibitemShut {NoStop}%
\bibitem [{\citenamefont {Pirandola}\ \emph {et~al.}(2020)\citenamefont
  {Pirandola}, \citenamefont {Andersen}, \citenamefont {Banchi}, \citenamefont
  {Berta}, \citenamefont {Bunandar}, \citenamefont {Colbeck}, \citenamefont
  {Englund}, \citenamefont {Gehring}, \citenamefont {Lupo}, \citenamefont
  {Ottaviani} \emph {et~al.}}]{pirandola2020advances}%
  \BibitemOpen
  \bibfield  {author} {\bibinfo {author} {\bibfnamefont {S.}~\bibnamefont
  {Pirandola}}, \bibinfo {author} {\bibfnamefont {U.~L.}\ \bibnamefont
  {Andersen}}, \bibinfo {author} {\bibfnamefont {L.}~\bibnamefont {Banchi}},
  \bibinfo {author} {\bibfnamefont {M.}~\bibnamefont {Berta}}, \bibinfo
  {author} {\bibfnamefont {D.}~\bibnamefont {Bunandar}}, \bibinfo {author}
  {\bibfnamefont {R.}~\bibnamefont {Colbeck}}, \bibinfo {author} {\bibfnamefont
  {D.}~\bibnamefont {Englund}}, \bibinfo {author} {\bibfnamefont
  {T.}~\bibnamefont {Gehring}}, \bibinfo {author} {\bibfnamefont
  {C.}~\bibnamefont {Lupo}}, \bibinfo {author} {\bibfnamefont {C.}~\bibnamefont
  {Ottaviani}}, \emph {et~al.},\ }\bibfield  {title} {\bibinfo {title}
  {Advances in quantum cryptography},\ }\href@noop {} {\bibfield  {journal}
  {\bibinfo  {journal} {Advances in optics and photonics}\ }\textbf {\bibinfo
  {volume} {12}},\ \bibinfo {pages} {1012} (\bibinfo {year}
  {2020})}\BibitemShut {NoStop}%
\bibitem [{\citenamefont {Aaronson}\ and\ \citenamefont
  {Arkhipov}(2011)}]{aaronson2011computational}%
  \BibitemOpen
  \bibfield  {author} {\bibinfo {author} {\bibfnamefont {S.}~\bibnamefont
  {Aaronson}}\ and\ \bibinfo {author} {\bibfnamefont {A.}~\bibnamefont
  {Arkhipov}},\ }\bibfield  {title} {\bibinfo {title} {The computational
  complexity of linear optics},\ }in\ \href@noop {} {\emph {\bibinfo
  {booktitle} {Proceedings of the forty-third annual ACM symposium on Theory of
  computing}}}\ (\bibinfo {year} {2011})\ pp.\ \bibinfo {pages}
  {333--342}\BibitemShut {NoStop}%
\bibitem [{\citenamefont {Feynman}(2018)}]{feynman2018simulating}%
  \BibitemOpen
  \bibfield  {author} {\bibinfo {author} {\bibfnamefont {R.~P.}\ \bibnamefont
  {Feynman}},\ }\bibfield  {title} {\bibinfo {title} {Simulating physics with
  computers},\ }in\ \href@noop {} {\emph {\bibinfo {booktitle} {Feynman and
  Computation}}},\ \bibinfo {editor} {edited by\ \bibinfo {editor}
  {\bibfnamefont {A.~J.~G.}\ \bibnamefont {Hey}}}\ (\bibinfo  {publisher} {CRC
  Press},\ \bibinfo {year} {2018})\ pp.\ \bibinfo {pages}
  {133--153}\BibitemShut {NoStop}%
\bibitem [{\citenamefont {Wen}\ \emph {et~al.}(2019)\citenamefont {Wen},
  \citenamefont {Zheng}, \citenamefont {Kong}, \citenamefont {Wei},
  \citenamefont {Xin},\ and\ \citenamefont {Long}}]{wen2019experimental}%
  \BibitemOpen
  \bibfield  {author} {\bibinfo {author} {\bibfnamefont {J.}~\bibnamefont
  {Wen}}, \bibinfo {author} {\bibfnamefont {C.}~\bibnamefont {Zheng}}, \bibinfo
  {author} {\bibfnamefont {X.}~\bibnamefont {Kong}}, \bibinfo {author}
  {\bibfnamefont {S.}~\bibnamefont {Wei}}, \bibinfo {author} {\bibfnamefont
  {T.}~\bibnamefont {Xin}},\ and\ \bibinfo {author} {\bibfnamefont
  {G.}~\bibnamefont {Long}},\ }\bibfield  {title} {\bibinfo {title}
  {Experimental demonstration of a digital quantum simulation of a general
  pt-symmetric system},\ }\href@noop {} {\bibfield  {journal} {\bibinfo
  {journal} {Physical Review A}\ }\textbf {\bibinfo {volume} {99}},\ \bibinfo
  {pages} {062122} (\bibinfo {year} {2019})}\BibitemShut {NoStop}%
\bibitem [{\citenamefont {Shen}\ \emph {et~al.}(2023)\citenamefont {Shen},
  \citenamefont {Chen}, \citenamefont {Yang},\ and\ \citenamefont
  {Lee}}]{shen2023observation}%
  \BibitemOpen
  \bibfield  {author} {\bibinfo {author} {\bibfnamefont {R.}~\bibnamefont
  {Shen}}, \bibinfo {author} {\bibfnamefont {T.}~\bibnamefont {Chen}}, \bibinfo
  {author} {\bibfnamefont {B.}~\bibnamefont {Yang}},\ and\ \bibinfo {author}
  {\bibfnamefont {C.~H.}\ \bibnamefont {Lee}},\ }\bibfield  {title} {\bibinfo
  {title} {Observation of the non-hermitian skin effect and fermi skin on a
  digital quantum computer},\ }\href@noop {} {\bibfield  {journal} {\bibinfo
  {journal} {arXiv preprint arXiv:2311.10143}\ } (\bibinfo {year}
  {2023})}\BibitemShut {NoStop}%
\bibitem [{\citenamefont {Cerezo}\ \emph
  {et~al.}(2021{\natexlab{a}})\citenamefont {Cerezo}, \citenamefont
  {Arrasmith}, \citenamefont {Babbush}, \citenamefont {Benjamin}, \citenamefont
  {Endo}, \citenamefont {Fujii}, \citenamefont {McClean}, \citenamefont
  {Mitarai}, \citenamefont {Yuan}, \citenamefont {Cincio} \emph
  {et~al.}}]{cerezo2021variational}%
  \BibitemOpen
  \bibfield  {author} {\bibinfo {author} {\bibfnamefont {M.}~\bibnamefont
  {Cerezo}}, \bibinfo {author} {\bibfnamefont {A.}~\bibnamefont {Arrasmith}},
  \bibinfo {author} {\bibfnamefont {R.}~\bibnamefont {Babbush}}, \bibinfo
  {author} {\bibfnamefont {S.~C.}\ \bibnamefont {Benjamin}}, \bibinfo {author}
  {\bibfnamefont {S.}~\bibnamefont {Endo}}, \bibinfo {author} {\bibfnamefont
  {K.}~\bibnamefont {Fujii}}, \bibinfo {author} {\bibfnamefont {J.~R.}\
  \bibnamefont {McClean}}, \bibinfo {author} {\bibfnamefont {K.}~\bibnamefont
  {Mitarai}}, \bibinfo {author} {\bibfnamefont {X.}~\bibnamefont {Yuan}},
  \bibinfo {author} {\bibfnamefont {L.}~\bibnamefont {Cincio}}, \emph
  {et~al.},\ }\bibfield  {title} {\bibinfo {title} {Variational quantum
  algorithms},\ }\href@noop {} {\bibfield  {journal} {\bibinfo  {journal}
  {Nature Reviews Physics}\ }\textbf {\bibinfo {volume} {3}},\ \bibinfo {pages}
  {625} (\bibinfo {year} {2021}{\natexlab{a}})}\BibitemShut {NoStop}%
\bibitem [{\citenamefont {Farhi}\ \emph {et~al.}(2014)\citenamefont {Farhi},
  \citenamefont {Goldstone},\ and\ \citenamefont {Gutmann}}]{farhi2014quantum}%
  \BibitemOpen
  \bibfield  {author} {\bibinfo {author} {\bibfnamefont {E.}~\bibnamefont
  {Farhi}}, \bibinfo {author} {\bibfnamefont {J.}~\bibnamefont {Goldstone}},\
  and\ \bibinfo {author} {\bibfnamefont {S.}~\bibnamefont {Gutmann}},\
  }\bibfield  {title} {\bibinfo {title} {A quantum approximate optimization
  algorithm},\ }\href@noop {} {\bibfield  {journal} {\bibinfo  {journal} {arXiv
  preprint arXiv:1411.4028}\ } (\bibinfo {year} {2014})}\BibitemShut {NoStop}%
\bibitem [{\citenamefont {Peruzzo}\ \emph {et~al.}(2014)\citenamefont
  {Peruzzo}, \citenamefont {McClean}, \citenamefont {Shadbolt}, \citenamefont
  {Yung}, \citenamefont {Zhou}, \citenamefont {Love}, \citenamefont
  {Aspuru-Guzik},\ and\ \citenamefont {O'brien}}]{peruzzo2014variational}%
  \BibitemOpen
  \bibfield  {author} {\bibinfo {author} {\bibfnamefont {A.}~\bibnamefont
  {Peruzzo}}, \bibinfo {author} {\bibfnamefont {J.}~\bibnamefont {McClean}},
  \bibinfo {author} {\bibfnamefont {P.}~\bibnamefont {Shadbolt}}, \bibinfo
  {author} {\bibfnamefont {M.-H.}\ \bibnamefont {Yung}}, \bibinfo {author}
  {\bibfnamefont {X.-Q.}\ \bibnamefont {Zhou}}, \bibinfo {author}
  {\bibfnamefont {P.~J.}\ \bibnamefont {Love}}, \bibinfo {author}
  {\bibfnamefont {A.}~\bibnamefont {Aspuru-Guzik}},\ and\ \bibinfo {author}
  {\bibfnamefont {J.~L.}\ \bibnamefont {O'brien}},\ }\bibfield  {title}
  {\bibinfo {title} {A variational eigenvalue solver on a photonic quantum
  processor},\ }\href@noop {} {\bibfield  {journal} {\bibinfo  {journal}
  {Nature communications}\ }\textbf {\bibinfo {volume} {5}},\ \bibinfo {pages}
  {4213} (\bibinfo {year} {2014})}\BibitemShut {NoStop}%
\bibitem [{\citenamefont {Preskill}(2012)}]{preskill2012quantum}%
  \BibitemOpen
  \bibfield  {author} {\bibinfo {author} {\bibfnamefont {J.}~\bibnamefont
  {Preskill}},\ }\bibfield  {title} {\bibinfo {title} {Quantum computing and
  the entanglement frontier},\ }\href@noop {} {\bibfield  {journal} {\bibinfo
  {journal} {arXiv preprint arXiv:1203.5813}\ } (\bibinfo {year}
  {2012})}\BibitemShut {NoStop}%
\bibitem [{\citenamefont {Carrasquilla}\ and\ \citenamefont
  {Melko}(2017)}]{carrasquilla2017machine}%
  \BibitemOpen
  \bibfield  {author} {\bibinfo {author} {\bibfnamefont {J.}~\bibnamefont
  {Carrasquilla}}\ and\ \bibinfo {author} {\bibfnamefont {R.~G.}\ \bibnamefont
  {Melko}},\ }\bibfield  {title} {\bibinfo {title} {Machine learning phases of
  matter},\ }\href@noop {} {\bibfield  {journal} {\bibinfo  {journal} {Nature
  Physics}\ }\textbf {\bibinfo {volume} {13}},\ \bibinfo {pages} {431}
  (\bibinfo {year} {2017})}\BibitemShut {NoStop}%
\bibitem [{\citenamefont {Sch{\"o}n}\ \emph {et~al.}(2005)\citenamefont
  {Sch{\"o}n}, \citenamefont {Solano}, \citenamefont {Verstraete},
  \citenamefont {Cirac},\ and\ \citenamefont {Wolf}}]{schon2005sequential}%
  \BibitemOpen
  \bibfield  {author} {\bibinfo {author} {\bibfnamefont {C.}~\bibnamefont
  {Sch{\"o}n}}, \bibinfo {author} {\bibfnamefont {E.}~\bibnamefont {Solano}},
  \bibinfo {author} {\bibfnamefont {F.}~\bibnamefont {Verstraete}}, \bibinfo
  {author} {\bibfnamefont {J.~I.}\ \bibnamefont {Cirac}},\ and\ \bibinfo
  {author} {\bibfnamefont {M.~M.}\ \bibnamefont {Wolf}},\ }\bibfield  {title}
  {\bibinfo {title} {Sequential generation of entangled multiqubit states},\
  }\href@noop {} {\bibfield  {journal} {\bibinfo  {journal} {Physical review
  letters}\ }\textbf {\bibinfo {volume} {95}},\ \bibinfo {pages} {110503}
  (\bibinfo {year} {2005})}\BibitemShut {NoStop}%
\bibitem [{\citenamefont {Lin}\ \emph {et~al.}(2021)\citenamefont {Lin},
  \citenamefont {Dilip}, \citenamefont {Green}, \citenamefont {Smith},\ and\
  \citenamefont {Pollmann}}]{lin2021real}%
  \BibitemOpen
  \bibfield  {author} {\bibinfo {author} {\bibfnamefont {S.-H.}\ \bibnamefont
  {Lin}}, \bibinfo {author} {\bibfnamefont {R.}~\bibnamefont {Dilip}}, \bibinfo
  {author} {\bibfnamefont {A.~G.}\ \bibnamefont {Green}}, \bibinfo {author}
  {\bibfnamefont {A.}~\bibnamefont {Smith}},\ and\ \bibinfo {author}
  {\bibfnamefont {F.}~\bibnamefont {Pollmann}},\ }\bibfield  {title} {\bibinfo
  {title} {Real-and imaginary-time evolution with compressed quantum
  circuits},\ }\href@noop {} {\bibfield  {journal} {\bibinfo  {journal} {PRX
  Quantum}\ }\textbf {\bibinfo {volume} {2}},\ \bibinfo {pages} {010342}
  (\bibinfo {year} {2021})}\BibitemShut {NoStop}%
\bibitem [{\citenamefont {Zhang}\ \emph
  {et~al.}(2024{\natexlab{a}})\citenamefont {Zhang}, \citenamefont {Wiersema},
  \citenamefont {Carrasquilla}, \citenamefont {Cincio},\ and\ \citenamefont
  {Kim}}]{zhang2024scalable}%
  \BibitemOpen
  \bibfield  {author} {\bibinfo {author} {\bibfnamefont {Y.}~\bibnamefont
  {Zhang}}, \bibinfo {author} {\bibfnamefont {R.}~\bibnamefont {Wiersema}},
  \bibinfo {author} {\bibfnamefont {J.}~\bibnamefont {Carrasquilla}}, \bibinfo
  {author} {\bibfnamefont {L.}~\bibnamefont {Cincio}},\ and\ \bibinfo {author}
  {\bibfnamefont {Y.~B.}\ \bibnamefont {Kim}},\ }\bibfield  {title} {\bibinfo
  {title} {Scalable quantum dynamics compilation via quantum machine
  learning},\ }\href@noop {} {\bibfield  {journal} {\bibinfo  {journal} {arXiv
  preprint arXiv:2409.16346}\ } (\bibinfo {year}
  {2024}{\natexlab{a}})}\BibitemShut {NoStop}%
\bibitem [{\citenamefont {Foss-Feig}\ \emph
  {et~al.}(2021{\natexlab{a}})\citenamefont {Foss-Feig}, \citenamefont {Hayes},
  \citenamefont {Dreiling}, \citenamefont {Figgatt}, \citenamefont {Gaebler},
  \citenamefont {Moses}, \citenamefont {Pino},\ and\ \citenamefont
  {Potter}}]{foss2021holographic}%
  \BibitemOpen
  \bibfield  {author} {\bibinfo {author} {\bibfnamefont {M.}~\bibnamefont
  {Foss-Feig}}, \bibinfo {author} {\bibfnamefont {D.}~\bibnamefont {Hayes}},
  \bibinfo {author} {\bibfnamefont {J.~M.}\ \bibnamefont {Dreiling}}, \bibinfo
  {author} {\bibfnamefont {C.}~\bibnamefont {Figgatt}}, \bibinfo {author}
  {\bibfnamefont {J.~P.}\ \bibnamefont {Gaebler}}, \bibinfo {author}
  {\bibfnamefont {S.~A.}\ \bibnamefont {Moses}}, \bibinfo {author}
  {\bibfnamefont {J.~M.}\ \bibnamefont {Pino}},\ and\ \bibinfo {author}
  {\bibfnamefont {A.~C.}\ \bibnamefont {Potter}},\ }\bibfield  {title}
  {\bibinfo {title} {Holographic quantum algorithms for simulating correlated
  spin systems},\ }\href@noop {} {\bibfield  {journal} {\bibinfo  {journal}
  {Physical Review Research}\ }\textbf {\bibinfo {volume} {3}},\ \bibinfo
  {pages} {033002} (\bibinfo {year} {2021}{\natexlab{a}})}\BibitemShut
  {NoStop}%
\bibitem [{\citenamefont {Perez-Garcia}\ \emph {et~al.}(2006)\citenamefont
  {Perez-Garcia}, \citenamefont {Verstraete}, \citenamefont {Wolf},\ and\
  \citenamefont {Cirac}}]{perez2006matrix}%
  \BibitemOpen
  \bibfield  {author} {\bibinfo {author} {\bibfnamefont {D.}~\bibnamefont
  {Perez-Garcia}}, \bibinfo {author} {\bibfnamefont {F.}~\bibnamefont
  {Verstraete}}, \bibinfo {author} {\bibfnamefont {M.~M.}\ \bibnamefont
  {Wolf}},\ and\ \bibinfo {author} {\bibfnamefont {J.~I.}\ \bibnamefont
  {Cirac}},\ }\bibfield  {title} {\bibinfo {title} {Matrix product state
  representations},\ }\href@noop {} {\bibfield  {journal} {\bibinfo  {journal}
  {arXiv preprint quant-ph/0608197}\ } (\bibinfo {year} {2006})}\BibitemShut
  {NoStop}%
\bibitem [{\citenamefont {Lieb}\ and\ \citenamefont
  {Robinson}(1972)}]{lieb1972finite}%
  \BibitemOpen
  \bibfield  {author} {\bibinfo {author} {\bibfnamefont {E.~H.}\ \bibnamefont
  {Lieb}}\ and\ \bibinfo {author} {\bibfnamefont {D.~W.}\ \bibnamefont
  {Robinson}},\ }\bibfield  {title} {\bibinfo {title} {The finite group
  velocity of quantum spin systems},\ }\href@noop {} {\bibfield  {journal}
  {\bibinfo  {journal} {Communications in mathematical physics}\ }\textbf
  {\bibinfo {volume} {28}},\ \bibinfo {pages} {251} (\bibinfo {year}
  {1972})}\BibitemShut {NoStop}%
\bibitem [{\citenamefont {Eisert}\ \emph {et~al.}(2013)\citenamefont {Eisert},
  \citenamefont {Van Den~Worm}, \citenamefont {Manmana},\ and\ \citenamefont
  {Kastner}}]{eisert2013breakdown}%
  \BibitemOpen
  \bibfield  {author} {\bibinfo {author} {\bibfnamefont {J.}~\bibnamefont
  {Eisert}}, \bibinfo {author} {\bibfnamefont {M.}~\bibnamefont {Van
  Den~Worm}}, \bibinfo {author} {\bibfnamefont {S.~R.}\ \bibnamefont
  {Manmana}},\ and\ \bibinfo {author} {\bibfnamefont {M.}~\bibnamefont
  {Kastner}},\ }\bibfield  {title} {\bibinfo {title} {Breakdown of
  quasilocality in long-range quantum lattice models},\ }\href@noop {}
  {\bibfield  {journal} {\bibinfo  {journal} {Physical review letters}\
  }\textbf {\bibinfo {volume} {111}},\ \bibinfo {pages} {260401} (\bibinfo
  {year} {2013})}\BibitemShut {NoStop}%
\bibitem [{\citenamefont {Richerme}\ \emph {et~al.}(2014)\citenamefont
  {Richerme}, \citenamefont {Gong}, \citenamefont {Lee}, \citenamefont {Senko},
  \citenamefont {Smith}, \citenamefont {Foss-Feig}, \citenamefont {Michalakis},
  \citenamefont {Gorshkov},\ and\ \citenamefont {Monroe}}]{richerme2014non}%
  \BibitemOpen
  \bibfield  {author} {\bibinfo {author} {\bibfnamefont {P.}~\bibnamefont
  {Richerme}}, \bibinfo {author} {\bibfnamefont {Z.-X.}\ \bibnamefont {Gong}},
  \bibinfo {author} {\bibfnamefont {A.}~\bibnamefont {Lee}}, \bibinfo {author}
  {\bibfnamefont {C.}~\bibnamefont {Senko}}, \bibinfo {author} {\bibfnamefont
  {J.}~\bibnamefont {Smith}}, \bibinfo {author} {\bibfnamefont
  {M.}~\bibnamefont {Foss-Feig}}, \bibinfo {author} {\bibfnamefont
  {S.}~\bibnamefont {Michalakis}}, \bibinfo {author} {\bibfnamefont {A.~V.}\
  \bibnamefont {Gorshkov}},\ and\ \bibinfo {author} {\bibfnamefont
  {C.}~\bibnamefont {Monroe}},\ }\bibfield  {title} {\bibinfo {title}
  {Non-local propagation of correlations in quantum systems with long-range
  interactions},\ }\href@noop {} {\bibfield  {journal} {\bibinfo  {journal}
  {Nature}\ }\textbf {\bibinfo {volume} {511}},\ \bibinfo {pages} {198}
  (\bibinfo {year} {2014})}\BibitemShut {NoStop}%
\bibitem [{\citenamefont {Moca}\ and\ \citenamefont
  {D{\'o}ra}(2021)}]{moca2021universal}%
  \BibitemOpen
  \bibfield  {author} {\bibinfo {author} {\bibfnamefont {C.~P.}\ \bibnamefont
  {Moca}}\ and\ \bibinfo {author} {\bibfnamefont {B.}~\bibnamefont
  {D{\'o}ra}},\ }\bibfield  {title} {\bibinfo {title} {Universal conductance of
  a pt-symmetric luttinger liquid after a quantum quench},\ }\href@noop {}
  {\bibfield  {journal} {\bibinfo  {journal} {Physical Review B}\ }\textbf
  {\bibinfo {volume} {104}},\ \bibinfo {pages} {125124} (\bibinfo {year}
  {2021})}\BibitemShut {NoStop}%
\bibitem [{\citenamefont {Trotter}(1959)}]{trotter1959product}%
  \BibitemOpen
  \bibfield  {author} {\bibinfo {author} {\bibfnamefont {H.~F.}\ \bibnamefont
  {Trotter}},\ }\bibfield  {title} {\bibinfo {title} {On the product of
  semi-groups of operators},\ }\href@noop {} {\bibfield  {journal} {\bibinfo
  {journal} {Proceedings of the American Mathematical Society}\ }\textbf
  {\bibinfo {volume} {10}},\ \bibinfo {pages} {545} (\bibinfo {year}
  {1959})}\BibitemShut {NoStop}%
\bibitem [{\citenamefont {Gily{\'e}n}\ \emph {et~al.}(2019)\citenamefont
  {Gily{\'e}n}, \citenamefont {Su}, \citenamefont {Low},\ and\ \citenamefont
  {Wiebe}}]{gilyen2019quantum}%
  \BibitemOpen
  \bibfield  {author} {\bibinfo {author} {\bibfnamefont {A.}~\bibnamefont
  {Gily{\'e}n}}, \bibinfo {author} {\bibfnamefont {Y.}~\bibnamefont {Su}},
  \bibinfo {author} {\bibfnamefont {G.~H.}\ \bibnamefont {Low}},\ and\ \bibinfo
  {author} {\bibfnamefont {N.}~\bibnamefont {Wiebe}},\ }\bibfield  {title}
  {\bibinfo {title} {Quantum singular value transformation and beyond:
  exponential improvements for quantum matrix arithmetics},\ }in\ \href@noop {}
  {\emph {\bibinfo {booktitle} {Proceedings of the 51st Annual ACM SIGACT
  Symposium on Theory of Computing}}}\ (\bibinfo {year} {2019})\ pp.\ \bibinfo
  {pages} {193--204}\BibitemShut {NoStop}%
\bibitem [{\citenamefont {Fishman}\ and\ \citenamefont
  {White}(2015)}]{fishman2015compression}%
  \BibitemOpen
  \bibfield  {author} {\bibinfo {author} {\bibfnamefont {M.~T.}\ \bibnamefont
  {Fishman}}\ and\ \bibinfo {author} {\bibfnamefont {S.~R.}\ \bibnamefont
  {White}},\ }\bibfield  {title} {\bibinfo {title} {Compression of correlation
  matrices and an efficient method for forming matrix product states of
  fermionic gaussian states},\ }\href@noop {} {\bibfield  {journal} {\bibinfo
  {journal} {Physical Review B}\ }\textbf {\bibinfo {volume} {92}},\ \bibinfo
  {pages} {075132} (\bibinfo {year} {2015})}\BibitemShut {NoStop}%
\bibitem [{\citenamefont {Niu}\ \emph {et~al.}(2022)\citenamefont {Niu},
  \citenamefont {Haghshenas}, \citenamefont {Zhang}, \citenamefont {Foss-Feig},
  \citenamefont {Chan},\ and\ \citenamefont {Potter}}]{niu2022holographic}%
  \BibitemOpen
  \bibfield  {author} {\bibinfo {author} {\bibfnamefont {D.}~\bibnamefont
  {Niu}}, \bibinfo {author} {\bibfnamefont {R.}~\bibnamefont {Haghshenas}},
  \bibinfo {author} {\bibfnamefont {Y.}~\bibnamefont {Zhang}}, \bibinfo
  {author} {\bibfnamefont {M.}~\bibnamefont {Foss-Feig}}, \bibinfo {author}
  {\bibfnamefont {G.~K.-L.}\ \bibnamefont {Chan}},\ and\ \bibinfo {author}
  {\bibfnamefont {A.~C.}\ \bibnamefont {Potter}},\ }\bibfield  {title}
  {\bibinfo {title} {Holographic simulation of correlated electrons on a
  trapped-ion quantum processor},\ }\href@noop {} {\bibfield  {journal}
  {\bibinfo  {journal} {PRX Quantum}\ }\textbf {\bibinfo {volume} {3}},\
  \bibinfo {pages} {030317} (\bibinfo {year} {2022})}\BibitemShut {NoStop}%
\bibitem [{\citenamefont {Arute}\ \emph {et~al.}(2019)\citenamefont {Arute},
  \citenamefont {Arya}, \citenamefont {Babbush}, \citenamefont {Bacon},
  \citenamefont {Bardin}, \citenamefont {Barends}, \citenamefont {Biswas},
  \citenamefont {Boixo}, \citenamefont {Brandao}, \citenamefont {Buell},
  \citenamefont {Burkett}, \citenamefont {Chen}, \citenamefont {Chen},
  \citenamefont {Chiaro}, \citenamefont {Collins}, \citenamefont {Courtney},
  \citenamefont {Dunsworth}, \citenamefont {Farhi}, \citenamefont {Foxen},\
  and\ \citenamefont {Martinis}}]{arute2019quantum}%
  \BibitemOpen
  \bibfield  {author} {\bibinfo {author} {\bibfnamefont {F.}~\bibnamefont
  {Arute}}, \bibinfo {author} {\bibfnamefont {K.}~\bibnamefont {Arya}},
  \bibinfo {author} {\bibfnamefont {R.}~\bibnamefont {Babbush}}, \bibinfo
  {author} {\bibfnamefont {D.}~\bibnamefont {Bacon}}, \bibinfo {author}
  {\bibfnamefont {J.}~\bibnamefont {Bardin}}, \bibinfo {author} {\bibfnamefont
  {R.}~\bibnamefont {Barends}}, \bibinfo {author} {\bibfnamefont
  {R.}~\bibnamefont {Biswas}}, \bibinfo {author} {\bibfnamefont
  {S.}~\bibnamefont {Boixo}}, \bibinfo {author} {\bibfnamefont
  {F.}~\bibnamefont {Brandao}}, \bibinfo {author} {\bibfnamefont
  {D.}~\bibnamefont {Buell}}, \bibinfo {author} {\bibfnamefont
  {B.}~\bibnamefont {Burkett}}, \bibinfo {author} {\bibfnamefont
  {Y.}~\bibnamefont {Chen}}, \bibinfo {author} {\bibfnamefont {Z.}~\bibnamefont
  {Chen}}, \bibinfo {author} {\bibfnamefont {B.}~\bibnamefont {Chiaro}},
  \bibinfo {author} {\bibfnamefont {R.}~\bibnamefont {Collins}}, \bibinfo
  {author} {\bibfnamefont {W.}~\bibnamefont {Courtney}}, \bibinfo {author}
  {\bibfnamefont {A.}~\bibnamefont {Dunsworth}}, \bibinfo {author}
  {\bibfnamefont {E.}~\bibnamefont {Farhi}}, \bibinfo {author} {\bibfnamefont
  {B.}~\bibnamefont {Foxen}},\ and\ \bibinfo {author} {\bibfnamefont
  {J.}~\bibnamefont {Martinis}},\ }\bibfield  {title} {\bibinfo {title}
  {Quantum supremacy using a programmable superconducting processor},\ }\href
  {https://doi.org/10.1038/s41586-019-1666-5} {\bibfield  {journal} {\bibinfo
  {journal} {Nature}\ }\textbf {\bibinfo {volume} {574}},\ \bibinfo {pages}
  {505} (\bibinfo {year} {2019})}\BibitemShut {NoStop}%
\bibitem [{\citenamefont {Aaronson}(2005)}]{aaronson2005quantum}%
  \BibitemOpen
  \bibfield  {author} {\bibinfo {author} {\bibfnamefont {S.}~\bibnamefont
  {Aaronson}},\ }\bibfield  {title} {\bibinfo {title} {Quantum computing,
  postselection, and probabilistic polynomial-time},\ }\href@noop {} {\bibfield
   {journal} {\bibinfo  {journal} {Proceedings of the Royal Society A:
  Mathematical, Physical and Engineering Sciences}\ }\textbf {\bibinfo {volume}
  {461}},\ \bibinfo {pages} {3473} (\bibinfo {year} {2005})}\BibitemShut
  {NoStop}%
\bibitem [{\citenamefont {Brand{\~a}o}\ \emph {et~al.}(2021)\citenamefont
  {Brand{\~a}o}, \citenamefont {Chemissany}, \citenamefont {Hunter-Jones},
  \citenamefont {Kueng},\ and\ \citenamefont {Preskill}}]{brandao2021models}%
  \BibitemOpen
  \bibfield  {author} {\bibinfo {author} {\bibfnamefont {F.~G.}\ \bibnamefont
  {Brand{\~a}o}}, \bibinfo {author} {\bibfnamefont {W.}~\bibnamefont
  {Chemissany}}, \bibinfo {author} {\bibfnamefont {N.}~\bibnamefont
  {Hunter-Jones}}, \bibinfo {author} {\bibfnamefont {R.}~\bibnamefont
  {Kueng}},\ and\ \bibinfo {author} {\bibfnamefont {J.}~\bibnamefont
  {Preskill}},\ }\bibfield  {title} {\bibinfo {title} {Models of quantum
  complexity growth},\ }\href@noop {} {\bibfield  {journal} {\bibinfo
  {journal} {PRX Quantum}\ }\textbf {\bibinfo {volume} {2}},\ \bibinfo {pages}
  {030316} (\bibinfo {year} {2021})}\BibitemShut {NoStop}%
\bibitem [{\citenamefont {Sch{\"o}n}\ \emph {et~al.}(2007)\citenamefont
  {Sch{\"o}n}, \citenamefont {Hammerer}, \citenamefont {Wolf}, \citenamefont
  {Cirac},\ and\ \citenamefont {Solano}}]{schon2007sequential}%
  \BibitemOpen
  \bibfield  {author} {\bibinfo {author} {\bibfnamefont {C.}~\bibnamefont
  {Sch{\"o}n}}, \bibinfo {author} {\bibfnamefont {K.}~\bibnamefont {Hammerer}},
  \bibinfo {author} {\bibfnamefont {M.~M.}\ \bibnamefont {Wolf}}, \bibinfo
  {author} {\bibfnamefont {J.~I.}\ \bibnamefont {Cirac}},\ and\ \bibinfo
  {author} {\bibfnamefont {E.}~\bibnamefont {Solano}},\ }\bibfield  {title}
  {\bibinfo {title} {Sequential generation of matrix-product states in cavity
  qed},\ }\href@noop {} {\bibfield  {journal} {\bibinfo  {journal} {Physical
  Review A}\ }\textbf {\bibinfo {volume} {75}},\ \bibinfo {pages} {032311}
  (\bibinfo {year} {2007})}\BibitemShut {NoStop}%
\bibitem [{\citenamefont {Xie}\ \emph {et~al.}(2023)\citenamefont {Xie},
  \citenamefont {Xue},\ and\ \citenamefont {Zhang}}]{xie2023variational}%
  \BibitemOpen
  \bibfield  {author} {\bibinfo {author} {\bibfnamefont {X.-D.}\ \bibnamefont
  {Xie}}, \bibinfo {author} {\bibfnamefont {Z.-Y.}\ \bibnamefont {Xue}},\ and\
  \bibinfo {author} {\bibfnamefont {D.-B.}\ \bibnamefont {Zhang}},\ }\bibfield
  {title} {\bibinfo {title} {Variational quantum eigensolvers for the
  non-hermitian systems by variance minimization},\ }\href@noop {} {\bibfield
  {journal} {\bibinfo  {journal} {arXiv preprint arXiv:2305.19807}\ } (\bibinfo
  {year} {2023})}\BibitemShut {NoStop}%
\bibitem [{\citenamefont {Von~Gehlen}(1991)}]{von1991critical}%
  \BibitemOpen
  \bibfield  {author} {\bibinfo {author} {\bibfnamefont {G.}~\bibnamefont
  {Von~Gehlen}},\ }\bibfield  {title} {\bibinfo {title} {Critical and
  off-critical conformal analysis of the ising quantum chain in an imaginary
  field},\ }\href@noop {} {\bibfield  {journal} {\bibinfo  {journal} {Journal
  of Physics A: Mathematical and General}\ }\textbf {\bibinfo {volume} {24}},\
  \bibinfo {pages} {5371} (\bibinfo {year} {1991})}\BibitemShut {NoStop}%
\bibitem [{\citenamefont {Fishman}\ \emph {et~al.}(2022)\citenamefont
  {Fishman}, \citenamefont {White},\ and\ \citenamefont
  {Stoudenmire}}]{fishman2022itensor}%
  \BibitemOpen
  \bibfield  {author} {\bibinfo {author} {\bibfnamefont {M.}~\bibnamefont
  {Fishman}}, \bibinfo {author} {\bibfnamefont {S.~R.}\ \bibnamefont {White}},\
  and\ \bibinfo {author} {\bibfnamefont {E.~M.}\ \bibnamefont {Stoudenmire}},\
  }\bibfield  {title} {\bibinfo {title} {The itensor software library for
  tensor network calculations, scipost phys},\ }\href@noop {} {\bibfield
  {journal} {\bibinfo  {journal} {Codebases}\ }\textbf {\bibinfo {volume}
  {4}},\ \bibinfo {pages} {1} (\bibinfo {year} {2022})}\BibitemShut {NoStop}%
\bibitem [{\citenamefont {Li}\ \emph {et~al.}(2018)\citenamefont {Li},
  \citenamefont {Chen},\ and\ \citenamefont {Fisher}}]{li2018quantum}%
  \BibitemOpen
  \bibfield  {author} {\bibinfo {author} {\bibfnamefont {Y.}~\bibnamefont
  {Li}}, \bibinfo {author} {\bibfnamefont {X.}~\bibnamefont {Chen}},\ and\
  \bibinfo {author} {\bibfnamefont {M.~P.}\ \bibnamefont {Fisher}},\ }\bibfield
   {title} {\bibinfo {title} {Quantum zeno effect and the many-body
  entanglement transition},\ }\href@noop {} {\bibfield  {journal} {\bibinfo
  {journal} {Physical Review B}\ }\textbf {\bibinfo {volume} {98}},\ \bibinfo
  {pages} {205136} (\bibinfo {year} {2018})}\BibitemShut {NoStop}%
\bibitem [{\citenamefont {Liu}\ \emph {et~al.}(2021)\citenamefont {Liu},
  \citenamefont {Kolden}, \citenamefont {Krovi}, \citenamefont {Loureiro},
  \citenamefont {Trivisa},\ and\ \citenamefont {Childs}}]{liu2021efficient}%
  \BibitemOpen
  \bibfield  {author} {\bibinfo {author} {\bibfnamefont {J.-P.}\ \bibnamefont
  {Liu}}, \bibinfo {author} {\bibfnamefont {H.~{\O}.}\ \bibnamefont {Kolden}},
  \bibinfo {author} {\bibfnamefont {H.~K.}\ \bibnamefont {Krovi}}, \bibinfo
  {author} {\bibfnamefont {N.~F.}\ \bibnamefont {Loureiro}}, \bibinfo {author}
  {\bibfnamefont {K.}~\bibnamefont {Trivisa}},\ and\ \bibinfo {author}
  {\bibfnamefont {A.~M.}\ \bibnamefont {Childs}},\ }\bibfield  {title}
  {\bibinfo {title} {Efficient quantum algorithm for dissipative nonlinear
  differential equations},\ }\href@noop {} {\bibfield  {journal} {\bibinfo
  {journal} {Proceedings of the National Academy of Sciences}\ }\textbf
  {\bibinfo {volume} {118}},\ \bibinfo {pages} {e2026805118} (\bibinfo {year}
  {2021})}\BibitemShut {NoStop}%
\bibitem [{\citenamefont {Yoshida}\ \emph {et~al.}(2019)\citenamefont
  {Yoshida}, \citenamefont {Kudo},\ and\ \citenamefont
  {Hatsugai}}]{yoshida2019non}%
  \BibitemOpen
  \bibfield  {author} {\bibinfo {author} {\bibfnamefont {T.}~\bibnamefont
  {Yoshida}}, \bibinfo {author} {\bibfnamefont {K.}~\bibnamefont {Kudo}},\ and\
  \bibinfo {author} {\bibfnamefont {Y.}~\bibnamefont {Hatsugai}},\ }\bibfield
  {title} {\bibinfo {title} {Non-hermitian fractional quantum hall states},\
  }\href@noop {} {\bibfield  {journal} {\bibinfo  {journal} {Scientific
  reports}\ }\textbf {\bibinfo {volume} {9}},\ \bibinfo {pages} {16895}
  (\bibinfo {year} {2019})}\BibitemShut {NoStop}%
\bibitem [{\citenamefont {Yin}\ \emph {et~al.}(2017)\citenamefont {Yin},
  \citenamefont {Huang}, \citenamefont {Lo},\ and\ \citenamefont
  {Chen}}]{yin2017kibble}%
  \BibitemOpen
  \bibfield  {author} {\bibinfo {author} {\bibfnamefont {S.}~\bibnamefont
  {Yin}}, \bibinfo {author} {\bibfnamefont {G.-Y.}\ \bibnamefont {Huang}},
  \bibinfo {author} {\bibfnamefont {C.-Y.}\ \bibnamefont {Lo}},\ and\ \bibinfo
  {author} {\bibfnamefont {P.}~\bibnamefont {Chen}},\ }\bibfield  {title}
  {\bibinfo {title} {Kibble-zurek scaling in the yang-lee edge singularity},\
  }\href@noop {} {\bibfield  {journal} {\bibinfo  {journal} {Physical Review
  Letters}\ }\textbf {\bibinfo {volume} {118}},\ \bibinfo {pages} {065701}
  (\bibinfo {year} {2017})}\BibitemShut {NoStop}%
\bibitem [{\citenamefont {D{\'o}ra}\ \emph {et~al.}(2019)\citenamefont
  {D{\'o}ra}, \citenamefont {Heyl},\ and\ \citenamefont
  {Moessner}}]{dora2019kibble}%
  \BibitemOpen
  \bibfield  {author} {\bibinfo {author} {\bibfnamefont {B.}~\bibnamefont
  {D{\'o}ra}}, \bibinfo {author} {\bibfnamefont {M.}~\bibnamefont {Heyl}},\
  and\ \bibinfo {author} {\bibfnamefont {R.}~\bibnamefont {Moessner}},\
  }\bibfield  {title} {\bibinfo {title} {The kibble-zurek mechanism at
  exceptional points},\ }\href@noop {} {\bibfield  {journal} {\bibinfo
  {journal} {Nature communications}\ }\textbf {\bibinfo {volume} {10}},\
  \bibinfo {pages} {2254} (\bibinfo {year} {2019})}\BibitemShut {NoStop}%
\bibitem [{\citenamefont {Hastings}(2007)}]{hastings2007area}%
  \BibitemOpen
  \bibfield  {author} {\bibinfo {author} {\bibfnamefont {M.~B.}\ \bibnamefont
  {Hastings}},\ }\bibfield  {title} {\bibinfo {title} {An area law for
  one-dimensional quantum systems},\ }\href
  {https://doi.org/10.1088/1742-5468/2007/08/p08024} {\bibfield  {journal}
  {\bibinfo  {journal} {Journal of Statistical Mechanics: Theory and
  Experiment}\ }\textbf {\bibinfo {volume} {2007}},\ \bibinfo {pages} {P08024}
  (\bibinfo {year} {2007})}\BibitemShut {NoStop}%
\bibitem [{\citenamefont {Manzano}(2020)}]{manzano2020short}%
  \BibitemOpen
  \bibfield  {author} {\bibinfo {author} {\bibfnamefont {D.}~\bibnamefont
  {Manzano}},\ }\bibfield  {title} {\bibinfo {title} {A short introduction to
  the lindblad master equation},\ }\href@noop {} {\bibfield  {journal}
  {\bibinfo  {journal} {Aip Advances}\ }\textbf {\bibinfo {volume} {10}}
  (\bibinfo {year} {2020})}\BibitemShut {NoStop}%
\bibitem [{\citenamefont {Zhang}\ \emph {et~al.}(2022)\citenamefont {Zhang},
  \citenamefont {Jahanbani}, \citenamefont {Niu}, \citenamefont {Haghshenas},\
  and\ \citenamefont {Potter}}]{zhang2022qubit}%
  \BibitemOpen
  \bibfield  {author} {\bibinfo {author} {\bibfnamefont {Y.}~\bibnamefont
  {Zhang}}, \bibinfo {author} {\bibfnamefont {S.}~\bibnamefont {Jahanbani}},
  \bibinfo {author} {\bibfnamefont {D.}~\bibnamefont {Niu}}, \bibinfo {author}
  {\bibfnamefont {R.}~\bibnamefont {Haghshenas}},\ and\ \bibinfo {author}
  {\bibfnamefont {A.~C.}\ \bibnamefont {Potter}},\ }\bibfield  {title}
  {\bibinfo {title} {Qubit-efficient simulation of thermal states with quantum
  tensor networks},\ }\href@noop {} {\bibfield  {journal} {\bibinfo  {journal}
  {Physical Review B}\ }\textbf {\bibinfo {volume} {106}},\ \bibinfo {pages}
  {165126} (\bibinfo {year} {2022})}\BibitemShut {NoStop}%
\bibitem [{\citenamefont {Chen}\ \emph
  {et~al.}(2024{\natexlab{a}})\citenamefont {Chen}, \citenamefont {Hermele},\
  and\ \citenamefont {Stephen}}]{chen2024sequential2}%
  \BibitemOpen
  \bibfield  {author} {\bibinfo {author} {\bibfnamefont {X.}~\bibnamefont
  {Chen}}, \bibinfo {author} {\bibfnamefont {M.}~\bibnamefont {Hermele}},\ and\
  \bibinfo {author} {\bibfnamefont {D.~T.}\ \bibnamefont {Stephen}},\
  }\bibfield  {title} {\bibinfo {title} {Sequential adiabatic generation of
  chiral topological states},\ }\href@noop {} {\bibfield  {journal} {\bibinfo
  {journal} {arXiv preprint arXiv:2402.03433}\ } (\bibinfo {year}
  {2024}{\natexlab{a}})}\BibitemShut {NoStop}%
\bibitem [{\citenamefont {Chen}\ \emph
  {et~al.}(2024{\natexlab{b}})\citenamefont {Chen}, \citenamefont {Dua},
  \citenamefont {Hermele}, \citenamefont {Stephen}, \citenamefont
  {Tantivasadakarn}, \citenamefont {Vanhove},\ and\ \citenamefont
  {Zhao}}]{chen2024sequential}%
  \BibitemOpen
  \bibfield  {author} {\bibinfo {author} {\bibfnamefont {X.}~\bibnamefont
  {Chen}}, \bibinfo {author} {\bibfnamefont {A.}~\bibnamefont {Dua}}, \bibinfo
  {author} {\bibfnamefont {M.}~\bibnamefont {Hermele}}, \bibinfo {author}
  {\bibfnamefont {D.~T.}\ \bibnamefont {Stephen}}, \bibinfo {author}
  {\bibfnamefont {N.}~\bibnamefont {Tantivasadakarn}}, \bibinfo {author}
  {\bibfnamefont {R.}~\bibnamefont {Vanhove}},\ and\ \bibinfo {author}
  {\bibfnamefont {J.-Y.}\ \bibnamefont {Zhao}},\ }\bibfield  {title} {\bibinfo
  {title} {Sequential quantum circuits as maps between gapped phases},\
  }\href@noop {} {\bibfield  {journal} {\bibinfo  {journal} {Physical Review
  B}\ }\textbf {\bibinfo {volume} {109}},\ \bibinfo {pages} {075116} (\bibinfo
  {year} {2024}{\natexlab{b}})}\BibitemShut {NoStop}%
\bibitem [{\citenamefont {Wei}\ \emph {et~al.}(2022)\citenamefont {Wei},
  \citenamefont {Malz},\ and\ \citenamefont {Cirac}}]{wei2022sequential}%
  \BibitemOpen
  \bibfield  {author} {\bibinfo {author} {\bibfnamefont {Z.-Y.}\ \bibnamefont
  {Wei}}, \bibinfo {author} {\bibfnamefont {D.}~\bibnamefont {Malz}},\ and\
  \bibinfo {author} {\bibfnamefont {J.~I.}\ \bibnamefont {Cirac}},\ }\bibfield
  {title} {\bibinfo {title} {Sequential generation of projected entangled-pair
  states},\ }\href@noop {} {\bibfield  {journal} {\bibinfo  {journal} {Physical
  Review Letters}\ }\textbf {\bibinfo {volume} {128}},\ \bibinfo {pages}
  {010607} (\bibinfo {year} {2022})}\BibitemShut {NoStop}%
\bibitem [{\citenamefont {Lu}\ \emph {et~al.}(2022)\citenamefont {Lu},
  \citenamefont {Lessa}, \citenamefont {Kim},\ and\ \citenamefont
  {Hsieh}}]{lu2022measurement}%
  \BibitemOpen
  \bibfield  {author} {\bibinfo {author} {\bibfnamefont {T.-C.}\ \bibnamefont
  {Lu}}, \bibinfo {author} {\bibfnamefont {L.~A.}\ \bibnamefont {Lessa}},
  \bibinfo {author} {\bibfnamefont {I.~H.}\ \bibnamefont {Kim}},\ and\ \bibinfo
  {author} {\bibfnamefont {T.~H.}\ \bibnamefont {Hsieh}},\ }\bibfield  {title}
  {\bibinfo {title} {Measurement as a shortcut to long-range entangled quantum
  matter},\ }\href@noop {} {\bibfield  {journal} {\bibinfo  {journal} {PRX
  Quantum}\ }\textbf {\bibinfo {volume} {3}},\ \bibinfo {pages} {040337}
  (\bibinfo {year} {2022})}\BibitemShut {NoStop}%
\bibitem [{\citenamefont {Foss-Feig}\ \emph {et~al.}(2023)\citenamefont
  {Foss-Feig}, \citenamefont {Tikku}, \citenamefont {Lu}, \citenamefont
  {Mayer}, \citenamefont {Iqbal}, \citenamefont {Gatterman}, \citenamefont
  {Gerber}, \citenamefont {Gilmore}, \citenamefont {Gresh}, \citenamefont
  {Hankin} \emph {et~al.}}]{foss2023experimental}%
  \BibitemOpen
  \bibfield  {author} {\bibinfo {author} {\bibfnamefont {M.}~\bibnamefont
  {Foss-Feig}}, \bibinfo {author} {\bibfnamefont {A.}~\bibnamefont {Tikku}},
  \bibinfo {author} {\bibfnamefont {T.-C.}\ \bibnamefont {Lu}}, \bibinfo
  {author} {\bibfnamefont {K.}~\bibnamefont {Mayer}}, \bibinfo {author}
  {\bibfnamefont {M.}~\bibnamefont {Iqbal}}, \bibinfo {author} {\bibfnamefont
  {T.~M.}\ \bibnamefont {Gatterman}}, \bibinfo {author} {\bibfnamefont {J.~A.}\
  \bibnamefont {Gerber}}, \bibinfo {author} {\bibfnamefont {K.}~\bibnamefont
  {Gilmore}}, \bibinfo {author} {\bibfnamefont {D.}~\bibnamefont {Gresh}},
  \bibinfo {author} {\bibfnamefont {A.}~\bibnamefont {Hankin}}, \emph
  {et~al.},\ }\bibfield  {title} {\bibinfo {title} {Experimental demonstration
  of the advantage of adaptive quantum circuits},\ }\href@noop {} {\bibfield
  {journal} {\bibinfo  {journal} {arXiv preprint arXiv:2302.03029}\ } (\bibinfo
  {year} {2023})}\BibitemShut {NoStop}%
\bibitem [{\citenamefont {Malz}\ \emph {et~al.}(2024)\citenamefont {Malz},
  \citenamefont {Styliaris}, \citenamefont {Wei},\ and\ \citenamefont
  {Cirac}}]{malz2024preparation}%
  \BibitemOpen
  \bibfield  {author} {\bibinfo {author} {\bibfnamefont {D.}~\bibnamefont
  {Malz}}, \bibinfo {author} {\bibfnamefont {G.}~\bibnamefont {Styliaris}},
  \bibinfo {author} {\bibfnamefont {Z.-Y.}\ \bibnamefont {Wei}},\ and\ \bibinfo
  {author} {\bibfnamefont {J.~I.}\ \bibnamefont {Cirac}},\ }\bibfield  {title}
  {\bibinfo {title} {Preparation of matrix product states with log-depth
  quantum circuits},\ }\href@noop {} {\bibfield  {journal} {\bibinfo  {journal}
  {Physical Review Letters}\ }\textbf {\bibinfo {volume} {132}},\ \bibinfo
  {pages} {040404} (\bibinfo {year} {2024})}\BibitemShut {NoStop}%
\bibitem [{\citenamefont {Osborne}\ \emph {et~al.}(2010)\citenamefont
  {Osborne}, \citenamefont {Eisert},\ and\ \citenamefont
  {Verstraete}}]{osborne2010holographic}%
  \BibitemOpen
  \bibfield  {author} {\bibinfo {author} {\bibfnamefont {T.~J.}\ \bibnamefont
  {Osborne}}, \bibinfo {author} {\bibfnamefont {J.}~\bibnamefont {Eisert}},\
  and\ \bibinfo {author} {\bibfnamefont {F.}~\bibnamefont {Verstraete}},\
  }\bibfield  {title} {\bibinfo {title} {Holographic quantum states},\
  }\href@noop {} {\bibfield  {journal} {\bibinfo  {journal} {Physical review
  letters}\ }\textbf {\bibinfo {volume} {105}},\ \bibinfo {pages} {260401}
  (\bibinfo {year} {2010})}\BibitemShut {NoStop}%
\bibitem [{\citenamefont {Zhang}\ \emph
  {et~al.}(2024{\natexlab{b}})\citenamefont {Zhang}, \citenamefont {Jahanbani},
  \citenamefont {Riswadkar}, \citenamefont {Shankar},\ and\ \citenamefont
  {Potter}}]{zhang2024sequential}%
  \BibitemOpen
  \bibfield  {author} {\bibinfo {author} {\bibfnamefont {Y.}~\bibnamefont
  {Zhang}}, \bibinfo {author} {\bibfnamefont {S.}~\bibnamefont {Jahanbani}},
  \bibinfo {author} {\bibfnamefont {A.}~\bibnamefont {Riswadkar}}, \bibinfo
  {author} {\bibfnamefont {S.}~\bibnamefont {Shankar}},\ and\ \bibinfo {author}
  {\bibfnamefont {A.~C.}\ \bibnamefont {Potter}},\ }\bibfield  {title}
  {\bibinfo {title} {Sequential quantum simulation of spin chains with a single
  circuit qed device},\ }\href@noop {} {\bibfield  {journal} {\bibinfo
  {journal} {Physical Review A}\ }\textbf {\bibinfo {volume} {109}},\ \bibinfo
  {pages} {022606} (\bibinfo {year} {2024}{\natexlab{b}})}\BibitemShut
  {NoStop}%
\bibitem [{\citenamefont {Foss-Feig}\ \emph
  {et~al.}(2021{\natexlab{b}})\citenamefont {Foss-Feig}, \citenamefont
  {Ragole}, \citenamefont {Potter}, \citenamefont {Dreiling}, \citenamefont
  {Figgatt}, \citenamefont {Gaebler}, \citenamefont {Hall}, \citenamefont
  {Moses}, \citenamefont {Pino}, \citenamefont {Spaun} \emph
  {et~al.}}]{foss2021entanglement}%
  \BibitemOpen
  \bibfield  {author} {\bibinfo {author} {\bibfnamefont {M.}~\bibnamefont
  {Foss-Feig}}, \bibinfo {author} {\bibfnamefont {S.}~\bibnamefont {Ragole}},
  \bibinfo {author} {\bibfnamefont {A.}~\bibnamefont {Potter}}, \bibinfo
  {author} {\bibfnamefont {J.}~\bibnamefont {Dreiling}}, \bibinfo {author}
  {\bibfnamefont {C.}~\bibnamefont {Figgatt}}, \bibinfo {author} {\bibfnamefont
  {J.}~\bibnamefont {Gaebler}}, \bibinfo {author} {\bibfnamefont
  {A.}~\bibnamefont {Hall}}, \bibinfo {author} {\bibfnamefont {S.}~\bibnamefont
  {Moses}}, \bibinfo {author} {\bibfnamefont {J.}~\bibnamefont {Pino}},
  \bibinfo {author} {\bibfnamefont {B.}~\bibnamefont {Spaun}}, \emph {et~al.},\
  }\bibfield  {title} {\bibinfo {title} {Entanglement from tensor networks on a
  trapped-ion qccd quantum computer},\ }\href@noop {} {\bibfield  {journal}
  {\bibinfo  {journal} {arXiv preprint arXiv:2104.11235}\ } (\bibinfo {year}
  {2021}{\natexlab{b}})}\BibitemShut {NoStop}%
\bibitem [{\citenamefont {Hauschild}\ \emph {et~al.}(2018)\citenamefont
  {Hauschild}, \citenamefont {Leviatan}, \citenamefont {Bardarson},
  \citenamefont {Altman}, \citenamefont {Zaletel},\ and\ \citenamefont
  {Pollmann}}]{hauschild2018finding}%
  \BibitemOpen
  \bibfield  {author} {\bibinfo {author} {\bibfnamefont {J.}~\bibnamefont
  {Hauschild}}, \bibinfo {author} {\bibfnamefont {E.}~\bibnamefont {Leviatan}},
  \bibinfo {author} {\bibfnamefont {J.~H.}\ \bibnamefont {Bardarson}}, \bibinfo
  {author} {\bibfnamefont {E.}~\bibnamefont {Altman}}, \bibinfo {author}
  {\bibfnamefont {M.~P.}\ \bibnamefont {Zaletel}},\ and\ \bibinfo {author}
  {\bibfnamefont {F.}~\bibnamefont {Pollmann}},\ }\bibfield  {title} {\bibinfo
  {title} {Finding purifications with minimal entanglement},\ }\href
  {https://doi.org/10.1103/PhysRevB.98.235163} {\bibfield  {journal} {\bibinfo
  {journal} {Phys. Rev. B}\ }\textbf {\bibinfo {volume} {98}},\ \bibinfo
  {pages} {235163} (\bibinfo {year} {2018})}\BibitemShut {NoStop}%
\bibitem [{\citenamefont {Soejima}\ \emph {et~al.}(2020)\citenamefont
  {Soejima}, \citenamefont {Siva}, \citenamefont {Bultinck}, \citenamefont
  {Chatterjee}, \citenamefont {Pollmann}, \citenamefont {Zaletel} \emph
  {et~al.}}]{soejima2020isometric}%
  \BibitemOpen
  \bibfield  {author} {\bibinfo {author} {\bibfnamefont {T.}~\bibnamefont
  {Soejima}}, \bibinfo {author} {\bibfnamefont {K.}~\bibnamefont {Siva}},
  \bibinfo {author} {\bibfnamefont {N.}~\bibnamefont {Bultinck}}, \bibinfo
  {author} {\bibfnamefont {S.}~\bibnamefont {Chatterjee}}, \bibinfo {author}
  {\bibfnamefont {F.}~\bibnamefont {Pollmann}}, \bibinfo {author}
  {\bibfnamefont {M.~P.}\ \bibnamefont {Zaletel}}, \emph {et~al.},\ }\bibfield
  {title} {\bibinfo {title} {Isometric tensor network representation of
  string-net liquids},\ }\href@noop {} {\bibfield  {journal} {\bibinfo
  {journal} {Physical Review B}\ }\textbf {\bibinfo {volume} {101}},\ \bibinfo
  {pages} {085117} (\bibinfo {year} {2020})}\BibitemShut {NoStop}%
\bibitem [{\citenamefont {Chertkov}\ \emph {et~al.}(2021)\citenamefont
  {Chertkov}, \citenamefont {Bohnet}, \citenamefont {Francois}, \citenamefont
  {Gaebler}, \citenamefont {Gresh}, \citenamefont {Hankin}, \citenamefont
  {Lee}, \citenamefont {Hayes}, \citenamefont {Neyenhuis}, \citenamefont
  {Stutz} \emph {et~al.}}]{chertkov2021holographic}%
  \BibitemOpen
  \bibfield  {author} {\bibinfo {author} {\bibfnamefont {E.}~\bibnamefont
  {Chertkov}}, \bibinfo {author} {\bibfnamefont {J.}~\bibnamefont {Bohnet}},
  \bibinfo {author} {\bibfnamefont {D.}~\bibnamefont {Francois}}, \bibinfo
  {author} {\bibfnamefont {J.}~\bibnamefont {Gaebler}}, \bibinfo {author}
  {\bibfnamefont {D.}~\bibnamefont {Gresh}}, \bibinfo {author} {\bibfnamefont
  {A.}~\bibnamefont {Hankin}}, \bibinfo {author} {\bibfnamefont
  {K.}~\bibnamefont {Lee}}, \bibinfo {author} {\bibfnamefont {D.}~\bibnamefont
  {Hayes}}, \bibinfo {author} {\bibfnamefont {B.}~\bibnamefont {Neyenhuis}},
  \bibinfo {author} {\bibfnamefont {R.}~\bibnamefont {Stutz}}, \emph {et~al.},\
  }\bibfield  {title} {\bibinfo {title} {Holographic dynamics simulations with
  a trapped ion quantum computer},\ }in\ \href@noop {} {\emph {\bibinfo
  {booktitle} {Quantum Information and Measurement}}}\ (\bibinfo {organization}
  {Optical Society of America},\ \bibinfo {year} {2021})\ pp.\ \bibinfo {pages}
  {W3A--3}\BibitemShut {NoStop}%
\bibitem [{\citenamefont {Gyongyosi}\ and\ \citenamefont
  {Imre}(2012)}]{gyongyosi2012properties}%
  \BibitemOpen
  \bibfield  {author} {\bibinfo {author} {\bibfnamefont {L.}~\bibnamefont
  {Gyongyosi}}\ and\ \bibinfo {author} {\bibfnamefont {S.}~\bibnamefont
  {Imre}},\ }\bibfield  {title} {\bibinfo {title} {Properties of the quantum
  channel},\ }\href@noop {} {\bibfield  {journal} {\bibinfo  {journal} {arXiv
  preprint arXiv:1208.1270}\ } (\bibinfo {year} {2012})}\BibitemShut {NoStop}%
\bibitem [{\citenamefont {Or{\'u}s}(2019)}]{orus2019tensor}%
  \BibitemOpen
  \bibfield  {author} {\bibinfo {author} {\bibfnamefont {R.}~\bibnamefont
  {Or{\'u}s}},\ }\bibfield  {title} {\bibinfo {title} {Tensor networks for
  complex quantum systems},\ }\href@noop {} {\bibfield  {journal} {\bibinfo
  {journal} {Nature Reviews Physics}\ }\textbf {\bibinfo {volume} {1}},\
  \bibinfo {pages} {538} (\bibinfo {year} {2019})}\BibitemShut {NoStop}%
\bibitem [{\citenamefont {Arute}\ \emph {et~al.}(2020)\citenamefont {Arute},
  \citenamefont {Arya}, \citenamefont {Babbush}, \citenamefont {Bacon},
  \citenamefont {Bardin}, \citenamefont {Barends}, \citenamefont {Boixo},
  \citenamefont {Broughton}, \citenamefont {Buckley}, \citenamefont {Buell}
  \emph {et~al.}}]{arute2020hartree}%
  \BibitemOpen
  \bibfield  {author} {\bibinfo {author} {\bibfnamefont {F.}~\bibnamefont
  {Arute}}, \bibinfo {author} {\bibfnamefont {K.}~\bibnamefont {Arya}},
  \bibinfo {author} {\bibfnamefont {R.}~\bibnamefont {Babbush}}, \bibinfo
  {author} {\bibfnamefont {D.}~\bibnamefont {Bacon}}, \bibinfo {author}
  {\bibfnamefont {J.~C.}\ \bibnamefont {Bardin}}, \bibinfo {author}
  {\bibfnamefont {R.}~\bibnamefont {Barends}}, \bibinfo {author} {\bibfnamefont
  {S.}~\bibnamefont {Boixo}}, \bibinfo {author} {\bibfnamefont
  {M.}~\bibnamefont {Broughton}}, \bibinfo {author} {\bibfnamefont {B.~B.}\
  \bibnamefont {Buckley}}, \bibinfo {author} {\bibfnamefont {D.~A.}\
  \bibnamefont {Buell}}, \emph {et~al.},\ }\bibfield  {title} {\bibinfo {title}
  {Hartree-fock on a superconducting qubit quantum computer},\ }\href@noop {}
  {\bibfield  {journal} {\bibinfo  {journal} {Science}\ }\textbf {\bibinfo
  {volume} {369}},\ \bibinfo {pages} {1084} (\bibinfo {year}
  {2020})}\BibitemShut {NoStop}%
\bibitem [{\citenamefont {Fefferman}\ \emph {et~al.}(2023)\citenamefont
  {Fefferman}, \citenamefont {Ghosh}, \citenamefont {Gullans}, \citenamefont
  {Kuroiwa},\ and\ \citenamefont {Sharma}}]{fefferman2023effect}%
  \BibitemOpen
  \bibfield  {author} {\bibinfo {author} {\bibfnamefont {B.}~\bibnamefont
  {Fefferman}}, \bibinfo {author} {\bibfnamefont {S.}~\bibnamefont {Ghosh}},
  \bibinfo {author} {\bibfnamefont {M.}~\bibnamefont {Gullans}}, \bibinfo
  {author} {\bibfnamefont {K.}~\bibnamefont {Kuroiwa}},\ and\ \bibinfo {author}
  {\bibfnamefont {K.}~\bibnamefont {Sharma}},\ }\bibfield  {title} {\bibinfo
  {title} {Effect of non-unital noise on random circuit sampling},\ }\href@noop
  {} {\bibfield  {journal} {\bibinfo  {journal} {arXiv preprint
  arXiv:2306.16659}\ } (\bibinfo {year} {2023})}\BibitemShut {NoStop}%
\bibitem [{\citenamefont {Zhang}(2025)}]{yuxuan_zhang_2025_14834474}%
  \BibitemOpen
  \bibfield  {author} {\bibinfo {author} {\bibfnamefont {Y.}~\bibnamefont
  {Zhang}},\ }\href {https://doi.org/10.5281/zenodo.14834474} {\bibinfo {title}
  {yuxuanzhang1995/non\_hermitian: Variational circuits for non-hermitian
  dynamics}} (\bibinfo {year} {2025})\BibitemShut {NoStop}%
\bibitem [{\citenamefont {Yamamoto}\ \emph {et~al.}(2022)\citenamefont
  {Yamamoto}, \citenamefont {Nakagawa}, \citenamefont {Tezuka}, \citenamefont
  {Ueda},\ and\ \citenamefont {Kawakami}}]{yamamoto2022universal}%
  \BibitemOpen
  \bibfield  {author} {\bibinfo {author} {\bibfnamefont {K.}~\bibnamefont
  {Yamamoto}}, \bibinfo {author} {\bibfnamefont {M.}~\bibnamefont {Nakagawa}},
  \bibinfo {author} {\bibfnamefont {M.}~\bibnamefont {Tezuka}}, \bibinfo
  {author} {\bibfnamefont {M.}~\bibnamefont {Ueda}},\ and\ \bibinfo {author}
  {\bibfnamefont {N.}~\bibnamefont {Kawakami}},\ }\bibfield  {title} {\bibinfo
  {title} {Universal properties of dissipative tomonaga-luttinger liquids: Case
  study of a non-hermitian xxz spin chain},\ }\href@noop {} {\bibfield
  {journal} {\bibinfo  {journal} {Physical Review B}\ }\textbf {\bibinfo
  {volume} {105}},\ \bibinfo {pages} {205125} (\bibinfo {year}
  {2022})}\BibitemShut {NoStop}%
\bibitem [{\citenamefont {Cerezo}\ \emph
  {et~al.}(2021{\natexlab{b}})\citenamefont {Cerezo}, \citenamefont {Sone},
  \citenamefont {Volkoff}, \citenamefont {Cincio},\ and\ \citenamefont
  {Coles}}]{cerezo2021cost}%
  \BibitemOpen
  \bibfield  {author} {\bibinfo {author} {\bibfnamefont {M.}~\bibnamefont
  {Cerezo}}, \bibinfo {author} {\bibfnamefont {A.}~\bibnamefont {Sone}},
  \bibinfo {author} {\bibfnamefont {T.}~\bibnamefont {Volkoff}}, \bibinfo
  {author} {\bibfnamefont {L.}~\bibnamefont {Cincio}},\ and\ \bibinfo {author}
  {\bibfnamefont {P.~J.}\ \bibnamefont {Coles}},\ }\bibfield  {title} {\bibinfo
  {title} {Cost function dependent barren plateaus in shallow parametrized
  quantum circuits},\ }\href@noop {} {\bibfield  {journal} {\bibinfo  {journal}
  {Nature communications}\ }\textbf {\bibinfo {volume} {12}},\ \bibinfo {pages}
  {1} (\bibinfo {year} {2021}{\natexlab{b}})}\BibitemShut {NoStop}%
\bibitem [{\citenamefont {McClean}\ \emph {et~al.}(2018)\citenamefont
  {McClean}, \citenamefont {Boixo}, \citenamefont {Smelyanskiy}, \citenamefont
  {Babbush},\ and\ \citenamefont {Neven}}]{mcclean2018barren}%
  \BibitemOpen
  \bibfield  {author} {\bibinfo {author} {\bibfnamefont {J.~R.}\ \bibnamefont
  {McClean}}, \bibinfo {author} {\bibfnamefont {S.}~\bibnamefont {Boixo}},
  \bibinfo {author} {\bibfnamefont {V.~N.}\ \bibnamefont {Smelyanskiy}},
  \bibinfo {author} {\bibfnamefont {R.}~\bibnamefont {Babbush}},\ and\ \bibinfo
  {author} {\bibfnamefont {H.}~\bibnamefont {Neven}},\ }\bibfield  {title}
  {\bibinfo {title} {Barren plateaus in quantum neural network training
  landscapes},\ }\href@noop {} {\bibfield  {journal} {\bibinfo  {journal}
  {Nature communications}\ }\textbf {\bibinfo {volume} {9}},\ \bibinfo {pages}
  {4812} (\bibinfo {year} {2018})}\BibitemShut {NoStop}%
\bibitem [{\citenamefont {Drudis}\ \emph {et~al.}(2024)\citenamefont {Drudis},
  \citenamefont {Thanasilp}, \citenamefont {Holmes} \emph
  {et~al.}}]{drudis2024variational}%
  \BibitemOpen
  \bibfield  {author} {\bibinfo {author} {\bibfnamefont {M.}~\bibnamefont
  {Drudis}}, \bibinfo {author} {\bibfnamefont {S.}~\bibnamefont {Thanasilp}},
  \bibinfo {author} {\bibfnamefont {Z.}~\bibnamefont {Holmes}}, \emph
  {et~al.},\ }\bibfield  {title} {\bibinfo {title} {Variational quantum
  simulation: a case study for understanding warm starts},\ }\href@noop {}
  {\bibfield  {journal} {\bibinfo  {journal} {arXiv preprint arXiv:2404.10044}\
  } (\bibinfo {year} {2024})}\BibitemShut {NoStop}%
\bibitem [{\citenamefont {Haghshenas}\ \emph {et~al.}(2021)\citenamefont
  {Haghshenas}, \citenamefont {Gray}, \citenamefont {Potter},\ and\
  \citenamefont {Chan}}]{haghshenas2021variational}%
  \BibitemOpen
  \bibfield  {author} {\bibinfo {author} {\bibfnamefont {R.}~\bibnamefont
  {Haghshenas}}, \bibinfo {author} {\bibfnamefont {J.}~\bibnamefont {Gray}},
  \bibinfo {author} {\bibfnamefont {A.~C.}\ \bibnamefont {Potter}},\ and\
  \bibinfo {author} {\bibfnamefont {G.~K.-L.}\ \bibnamefont {Chan}},\
  }\href@noop {} {\bibinfo {title} {The variational power of quantum circuit
  tensor networks}} (\bibinfo {year} {2021}),\ \Eprint
  {https://arxiv.org/abs/2107.01307} {arXiv:2107.01307 [quant-ph]} \BibitemShut
  {NoStop}%
\bibitem [{\citenamefont {Pino}\ \emph {et~al.}(2021)\citenamefont {Pino},
  \citenamefont {Dreiling}, \citenamefont {Figgatt}, \citenamefont {Gaebler},
  \citenamefont {Moses}, \citenamefont {Allman}, \citenamefont {Baldwin},
  \citenamefont {Foss-Feig}, \citenamefont {Hayes}, \citenamefont {Mayer} \emph
  {et~al.}}]{pino2021demonstration}%
  \BibitemOpen
  \bibfield  {author} {\bibinfo {author} {\bibfnamefont {J.~M.}\ \bibnamefont
  {Pino}}, \bibinfo {author} {\bibfnamefont {J.~M.}\ \bibnamefont {Dreiling}},
  \bibinfo {author} {\bibfnamefont {C.}~\bibnamefont {Figgatt}}, \bibinfo
  {author} {\bibfnamefont {J.~P.}\ \bibnamefont {Gaebler}}, \bibinfo {author}
  {\bibfnamefont {S.~A.}\ \bibnamefont {Moses}}, \bibinfo {author}
  {\bibfnamefont {M.}~\bibnamefont {Allman}}, \bibinfo {author} {\bibfnamefont
  {C.}~\bibnamefont {Baldwin}}, \bibinfo {author} {\bibfnamefont
  {M.}~\bibnamefont {Foss-Feig}}, \bibinfo {author} {\bibfnamefont
  {D.}~\bibnamefont {Hayes}}, \bibinfo {author} {\bibfnamefont
  {K.}~\bibnamefont {Mayer}}, \emph {et~al.},\ }\bibfield  {title} {\bibinfo
  {title} {Demonstration of the trapped-ion quantum ccd computer
  architecture},\ }\href@noop {} {\bibfield  {journal} {\bibinfo  {journal}
  {Nature}\ }\textbf {\bibinfo {volume} {592}},\ \bibinfo {pages} {209}
  (\bibinfo {year} {2021})}\BibitemShut {NoStop}%
\bibitem [{\citenamefont {Motta}\ \emph {et~al.}(2020)\citenamefont {Motta},
  \citenamefont {Sun}, \citenamefont {Tan}, \citenamefont {O’Rourke},
  \citenamefont {Ye}, \citenamefont {Minnich}, \citenamefont {Brandao},\ and\
  \citenamefont {Chan}}]{motta2020determining}%
  \BibitemOpen
  \bibfield  {author} {\bibinfo {author} {\bibfnamefont {M.}~\bibnamefont
  {Motta}}, \bibinfo {author} {\bibfnamefont {C.}~\bibnamefont {Sun}}, \bibinfo
  {author} {\bibfnamefont {A.~T.~K.}\ \bibnamefont {Tan}}, \bibinfo {author}
  {\bibfnamefont {M.~J.}\ \bibnamefont {O’Rourke}}, \bibinfo {author}
  {\bibfnamefont {E.}~\bibnamefont {Ye}}, \bibinfo {author} {\bibfnamefont
  {A.~J.}\ \bibnamefont {Minnich}}, \bibinfo {author} {\bibfnamefont {F.~G.
  S.~L.}\ \bibnamefont {Brandao}},\ and\ \bibinfo {author} {\bibfnamefont
  {G.~K.-L.}\ \bibnamefont {Chan}},\ }\bibfield  {title} {\bibinfo {title}
  {Determining eigenstates and thermal states on a quantum computer using
  quantum imaginary time evolution},\ }\href
  {https://doi.org/10.1038/s41567-019-0704-4} {\bibfield  {journal} {\bibinfo
  {journal} {Nature Physics}\ }\textbf {\bibinfo {volume} {16}},\ \bibinfo
  {pages} {205} (\bibinfo {year} {2020})}\BibitemShut {NoStop}%
\bibitem [{\citenamefont {Low}\ and\ \citenamefont
  {Chuang}(2017)}]{low2017optimal}%
  \BibitemOpen
  \bibfield  {author} {\bibinfo {author} {\bibfnamefont {G.~H.}\ \bibnamefont
  {Low}}\ and\ \bibinfo {author} {\bibfnamefont {I.~L.}\ \bibnamefont
  {Chuang}},\ }\bibfield  {title} {\bibinfo {title} {Optimal hamiltonian
  simulation by quantum signal processing},\ }\href@noop {} {\bibfield
  {journal} {\bibinfo  {journal} {Physical review letters}\ }\textbf {\bibinfo
  {volume} {118}},\ \bibinfo {pages} {010501} (\bibinfo {year}
  {2017})}\BibitemShut {NoStop}%
\bibitem [{\citenamefont {Vatan}\ and\ \citenamefont
  {Williams}(2004)}]{vatan2004optimal}%
  \BibitemOpen
  \bibfield  {author} {\bibinfo {author} {\bibfnamefont {F.}~\bibnamefont
  {Vatan}}\ and\ \bibinfo {author} {\bibfnamefont {C.}~\bibnamefont
  {Williams}},\ }\bibfield  {title} {\bibinfo {title} {Optimal quantum circuits
  for general two-qubit gates},\ }\href@noop {} {\bibfield  {journal} {\bibinfo
   {journal} {Physical Review A}\ }\textbf {\bibinfo {volume} {69}},\ \bibinfo
  {pages} {032315} (\bibinfo {year} {2004})}\BibitemShut {NoStop}%
\end{thebibliography}%

\end{document}